\documentclass[12pt]{article}
\usepackage[utf8]{inputenc}
\usepackage{amsmath}
\usepackage{amsfonts,mathrsfs}
\usepackage{amssymb,amsthm,xcolor}
\usepackage[pdftex,plainpages=false,hypertexnames=false,pdfpagelabels]{hyperref}
\usepackage{xcolor}
\usepackage{enumerate}

\usepackage{tikz}
\usetikzlibrary{calc}
\usepackage{tocvsec2}
\usepackage{caption}
\definecolor{dark-red}{rgb}{0.7,0.25,0.25}
\definecolor{dark-blue}{rgb}{0.15,0.15,0.55}
\definecolor{medium-blue}{rgb}{0,0,.8}
\definecolor{DarkGreen}{RGB}{0,150,0}
\definecolor{rho}{named}{red}
\hypersetup{
   colorlinks, linkcolor={purple},
   citecolor={medium-blue}, urlcolor={medium-blue}
}
\usepackage[top=1in, bottom=1in, left=0.92in, right=0.92in]{geometry}

\usepackage[shortlabels]{enumitem}
\newcommand{\frm}{\mathfrak{m}}
\newcommand{\frg}{\mathfrak{g}}

\newcommand{\Diffcinfty}{\Diff_c^{(\infty)}(S^1)}
\newcommand{\Diffc}{\Diff_c(S^1)}
\newcommand{\Ran}{\operatorname{Ran}}
\newcommand{\Int}{\operatorname{int}}
\newcommand{\Interior}{\operatorname{interior}}
\newcommand{\Ann}{\operatorname{Ann}}
\newcommand{\GAnn}{\operatorname{GAnn}}
\newcommand{\scA}{\mathscr{A}}

\newcommand{\Rot}{\operatorname{Rot}}
\newcommand{\scU}{\mathscr{U}}

\newcommand{\cB}{\mathcal{B}}
\newcommand{\cC}{\mathcal{C}}
\newcommand{\cF}{\mathcal{F}}
\newcommand{\cL}{\mathcal{L}}

\newcommand{\cM}{\mathcal{M}}
\newcommand{\cY}{\mathcal{Y}}

\newcommand{\cH}{\mathcal{H}}

\newcommand{\scG}{\mathscr{G}}
\newcommand{\cK}{\mathcal{K}}
\newcommand{\C}{\mathbb{C}}
\newcommand{\bbC}{\mathbb{C}}
\newcommand{\cA}{\mathcal{A}}

\newcommand{\cI}{\mathcal{I}}
\newcommand{\cU}{\mathcal{U}}

\newcommand{\g}{\mathfrak{g}}

\newcommand{\dom}{\operatorname{dom}}

\newcommand{\Pair}{\operatorname{Pair}}

\newcommand{\bbF}{\mathbb{F}}
\newcommand{\R}{\mathbb{R}}
\newcommand{\bbR}{\mathbb{R}}
\newcommand{\Z}{\mathbb{Z}}
\newcommand{\bbZ}{\mathbb{Z}}

\newcommand{\cl}{\operatorname{cl}}

\newcommand{\id}{\operatorname{id}}
\newcommand{\abs}[1]{\left|#1\right|}
\newcommand{\babs}[1]{\big|#1\big|}

\newcommand{\norm}[1]{\left\|#1\right\|}

\newcommand{\vertex}[3]{\binom{#1}{#2 \, #3}}
\newcommand{\ip}[1]{\left\langle#1\right\rangle}

\newcommand{\bip}[1]{\big\langle#1\big\rangle}

\newcommand{\bbD}{\mathbb{D}}

\newcommand{\cDR}{\mathcal{DR}}
\newcommand{\cDA}{\mathcal{DA}}
\newcommand{\cDP}{\mathcal{DP}}

\newcommand{\cO}{\mathcal{O}}

\newcommand{\Diff}{\operatorname{Diff}}
\newcommand{\CAR}{\operatorname{CAR}}
\newcommand{\Span}{\operatorname{span}}

\newcommand{\End}{\operatorname{End}}
\newcommand{\Hom}{\operatorname{Hom}}

\newcommand{\Vir}{\operatorname{Vir}}
\newcommand{\Mob}{\operatorname{M\ddot{o}b}}

\newcommand{\grotimes}{\hat \otimes}

\newcommand{\interior}[1]{\mathring{#1}}

\newcommand{\hotimes}{\hat \otimes}

\newtheorem{thmalpha}{Theorem}

\newtheorem{Theorem}{Theorem}[section]
\newtheorem*{Theorem*}{Theorem}
\newtheorem{Lemma}[Theorem]{Lemma}
\newtheorem{Proposition}[Theorem]{Proposition}
\newtheorem{Corollary}[Theorem]{Corollary}
\newtheorem*{Corollary*}{Corollary}
\theoremstyle{definition}
\newtheorem{Remark}[Theorem]{Remark}
\newtheorem{Definition}[Theorem]{Definition}
\newtheorem*{Definition*}{Definition}
\newtheorem{Example}[Theorem]{Example}

\newtheorem*{Problem*}{Problem}

\newtheorem{Convention}[Theorem]{Convention}
\numberwithin{equation}{section}
\numberwithin{figure}{section}

\usepackage[nottoc]{tocbibind}
\settocbibname{References}

\begin{document}

\title{Representation theory in chiral conformal field theory: from fields to observables}
\author{James E. Tener}
\date{}

\maketitle

\abstract{
This article develops new techniques for understanding the relationship between the three different mathematical formulations of two-dimensional chiral conformal field theory: conformal nets (axiomatizing local observables), vertex operator algebras (axiomatizing fields), and Segal CFTs.
It builds upon previous work \cite{GRACFT1}, which introduced a geometric interpolation procedure for constructing conformal nets from VOAs via Segal CFT, simultaneously relating all three frameworks.
In this article, we extend this construction to study the relationship between the representation theory of conformal nets and the representation theory of vertex operator algebras.
We define a correspondence between representations in the two contexts, and show how to construct representations of conformal nets from VOAs.
We also show that this correspondence is rich enough to relate the respective `fusion product' theories for conformal nets and VOAs, by constructing local intertwiners (in the sense of conformal nets) from intertwining operators (in the sense of VOAs).
We use these techniques to show that all WZW conformal nets can be constructed using our geometric interpolation procedure.
}

\newpage

\setcounter{tocdepth}{2}
\tableofcontents
\newpage



\section{Introduction} 

\subsection{Overview and context}

The mathematical study of two-dimensional chiral conformal field theories (CFTs) has developed into a broad undertaking, linking diverse mathematical fields.
It occupies a middle-ground in the larger landscape of quantum field theory, being both sufficiently tractable to permit rigorous analysis, yet simultaneous rich enough to encode interesting mathematical structures (such as braided and modular tensor categories, vector-valued modular forms, subfactors, and so on).

There are three approaches to the study of CFTs.
The most developed definitions are \emph{conformal nets}, which axiomatize algebras of observables in the spirit of the Haag-Kastler approach, and \emph{vertex operator algebras}, which axiomatize the fields of a CFT in the spirit of the Wightman axioms.
The third framework is \emph{Segal CFTs}, which are functorial field theories (a relative of the better known Atiyah-Segal axioms of topological quantum field theory).

Many fundamental physical facts are theorems in one context, but conjectures in the others.
For example, the statement ``the fixed points of a rational CFT under a finite group of automorphisms is rational'' has been proven for conformal nets, but not for vertex operator algebras, and the statement ``the WZW model corresponding to a simple Lie group at positive integral level is rational'' has been proven for vertex operator algebras but not conformal nets.
Translated into the different frameworks, each of these statements has mathematical implications, specifically regarding the modularity of tensor categories and the finiteness of the indices of subfactors.
It is an important and ongoing project to rigorously develop the relationship between these frameworks to the point where important theorems, when proven in one context, can thereby be deduced in the others, with the goal of eventually obtaining a single unified framework for the mathematical study of CFTs.

The first systematic comparison of conformal nets and VOAs was recently undertaken by Carpi-Kawahigashi-Longo-Weiner \cite{CKLW18}.
One of the goals of their approach is to construct algebras of local observables by integrating vertex operators against compactly supported test functions.
This approach encounters technical challenges, as the smeared vertex operators do not act continuously on the Hilbert space, and as a result it is difficult to deduce the locality axiom of a conformal net (i.e., that observables localized in disjoint regions commute) from the locality axiom of a VOA.
A significant achievement of \cite{CKLW18} is to provide tools for demonstrating this locality axiom of conformal nets, and they are able to produce conformal nets from most of the important examples of unitary vertex operator algebras.

This article is the second in a series (initiated in \cite{GRACFT1}) presenting an alternative, geometric approach to relating conformal nets, VOAs, and Segal CFT.
It differs from other approaches (e.g. \cite{CKLW18}) in several respects, the most prominent being:

\begin{itemize}
\item we simultaneously relate conformal nets, vertex operator algebras, and Segal CFTs, and in fact we show that Segal CFTs can be used to interpolate between conformal nets and VOAs
\item all operators which appear are continuous, which alleviates many technical challenges present in other approaches
\end{itemize}

One outcome of this series is to relate the representation theory of conformal nets with the representation theory of vertex operator algebras.
In particular, relating the tensor products (`fusion products') of representations of conformal nets and VOAs is a challenging problem which has been studied for several decades, and which has major implications for many unsolved problems in the study of vertex operator algebras, conformal nets, subfactors, and modular forms (including the two problems regarding rationality of theories mentioned above).
As an analogy, we think of the VOA as playing the role of Lie algebra relative to the conformal net's role of Lie group.
The first article of this series \cite{GRACFT1} was concerned with the `exponential map' from VOAs to conformal nets.
In this article, we address the `exponentiation' of representations, and provide first results indicating that this exponentiation is compatible with the respective theories of tensor products developed for VOAs and conformal nets.

The first major results to compare tensor products of VOAs and conformal nets was the landmark paper of Wassermann \cite{Wa98}, in which he directly related the fusion products of the conformal net and vertex operator algebra corresponding to WZW models of type $A$ at positive integral level.
Wassermann's ideas have been developed and extended in work of Toledano-Laredo for WZW models of type $D$ \cite{TL97}, in work of Loke for unitary minimal models \cite{Loke}, and in recent articles of Gui for WZW models of type $B$,$C$, and $G$ \cite{GuiUnitarityI,GuiUnitarityII,GuiG2}.

A key step in the comparison of fusion products for VOAs and conformal nets is to relate intertwining operators (encoding the fusion rules for VOAs) with \emph{local intertwiners} for conformal nets (see Section \ref{sec: main results}).
The way that this is done in the above referenced works is specific to the examples studied, and does not easily generalize even to all WZW models, with $E_8$ being particularly challenging%
\footnote{See \cite[\S 6]{GuiG2} for a discussion of some of the difficulties involved.}%
.
In particular, these techniques have not been demonstrated to apply, even in principle, to models which do not come from Lie algebras (affine Lie algebras or the Virasoro algebra), which excludes many of the most interesting examples, including the Moonshine VOA.

In this article, we present a framework for comparing representations of VOAs and conformal nets which is model independent, and we apply it to many examples, including type $E$ WZW models and other models which do not arise from Lie algebras.
As an application of our work, we explicitly construct local intertwiners of conformal net representations from VOA intertwining operators, without relying on special analytic properties intrinsic to certain specific models.
This is the first link between tensor product theories for VOAs and conformal nets not being obtained on a model-by-model basis.
We also use our results to attack the problem of extending analytic properties of VOAs from subalgebras to extensions, and obtain general results for extensions of \emph{code type}.

Our results are described in more detail in Section \ref{sec: main results}, but first we give a short outline of our geometric construction of conformal nets from VOAs in Section \ref{sec: BLVO intro}.

\subsection{Bounded localized vertex operators}\label{sec: BLVO intro}

In \cite{GRACFT1}, we introduced an analytic condition on unitary VOAs called \emph{bounded localized vertex operators}.
Vertex operators algebras with this property provide conformal nets, via a procedure that we will soon describe.
First, we briefly recall the notions of conformal net and VOA (precise definitions may be found in Section \ref{secPreliminaries}).
The primary data of a conformal net is a family of von Neumann algebras $\cA(I)$ acting on a Hilbert space $\cH$ indexed by intervals $I \subset S^1$, along with a unitary representation $U$ of the centrally extended group of orientation preserving diffeomorphisms $\Diff_c(S^1)$.
The primary data of a VOA is a vector space $V$ along with a state-field correspondence $Y:V \to \End(V)[[x^{\pm 1}]]$ denoted $Y(v,x)$, which includes a representation of the Virasoro algebra $\Vir_c$.
A VOA is called unitary if $V$ is equipped with an inner product which is invariant for the fields $Y(v,x)$.

The Virasoro algebra $\Vir_c$ is the complexified Lie algebra of $\Diff_c(S^1)$, and it is widely agreed that there does not exist a group which deserves to be called the complexification of $\Diff_c(S^1)$.
On the other hand, Segal proposed that the \emph{semigroup of annuli} \cite[\S 2]{SegalDef} should be regarded as a subsemigroup of the (non-existent) group $\Diff_c(S^1)_{\bbC}$.
The semigroup of annuli consists of compact Riemann surfaces which are topologically annuli, equipped with parametrizations of their boundary circles.
Neretin showed that irreducible highest weight positive energy representations of $\Diff_c(S^1)$ admit natural extensions to representations of the semigroup of annuli \cite{Neretin90}.

In a conformal net, the unitary operators $U(\gamma)$ corresponding to $\gamma \in \Diff_c(I)$ (that is, $\gamma \in \Diff_c(S^1)$ which act as the identity on $I^c$) actually lie in the local algebra $\cA(I)$.
On the other hand, $U(\gamma)$ for a general diffeomorphism does not lie in any $\cA(I)$, and the same is true for the operators assigned to annuli in Neretin's representation.
Henriques proposed enlarging the semigroup of annuli to include what we will call \emph{degenerate annuli}, which are annuli that have been `pinched thin' in places so that the incoming and outgoing boundaries overlap (see \cite{Henriques14}).

\begin{figure}[!ht]
$$
\tikz[scale=.7,,baseline={([yshift=-.5ex]current bounding box.center)}]{
\coordinate (a) at ($(0,0)+(-45:2cm and 1cm)$);
\coordinate (b) at ($(0,0)+(225:2cm and 1cm)$);
\coordinate (c) at ($(0,1)+(-45:2cm and 1cm)$);
\coordinate (d) at ($(0,1)+(225:2cm and 1cm)$);
\coordinate (e) at ($(0,.5)+(-70:2cm and 1cm)$);
\coordinate (f) at ($(0,.5)+(-110:2cm and 1cm)$);
\fill[red!10!blue!20!gray!30!white] (-2,0) arc (180:0:2cm and 1cm) -- +(0,1) arc (0:180:2cm and 1cm) -- cycle;
\draw[line width=1](0,0)+(-45:2cm and 1cm) arc (-45:180+45:2cm and 1cm)
(0,1)+(-45:2cm and 1cm) arc (-45:180+45:2cm and 1cm)
(0,.5)+(-70:2cm and 1cm) arc (-70:-110:2cm and 1cm);
\filldraw[line width=1, fill=red!10!blue!20!gray!30!white]
(2,0) arc (0:-45:2cm and 1cm) to[out=195, in=0] (e) to[out=20, in=203] (c) arc (-45:0:2cm and 1cm) -- cycle
(-2,0) arc (180:225:2cm and 1cm) to[out=-20, in=180] (f) to[out=160, in=-23] (d) arc (225:180:2cm and 1cm) -- cycle;
}
\qquad \qquad \qquad
\begin{tikzpicture}[baseline={([yshift=-.5ex]current bounding box.center)}]
	\coordinate (a) at (120:1cm);
	\coordinate (b) at (240:1cm);
	\coordinate (c) at (180:.25cm);
	\fill[fill=red!10!blue!20!gray!30!white] (0,0) circle (1cm);
	\draw (0,0) circle (1cm);
	\fill[fill=white] (a)  .. controls ++(210:.6cm) and ++(90:.4cm) .. (c) .. controls ++(270:.4cm) and ++(150:.6cm) .. (b) -- ([shift=(240:1cm)]0,0) arc (240:480:1cm);
	\draw ([shift=(240:1cm)]0,0) arc (240:480:1cm);
	\draw (a) .. controls ++(210:.6cm) and ++(90:.4cm) .. (c);
	\draw (b) .. controls ++(150:.6cm) and ++(270:.4cm) .. (c);
\end{tikzpicture}
$$
\captionsetup{justification=centering,width=0.8\linewidth}
\caption{A pair of degenerate annuli, one (from \cite{Henriques14}) depicted in three space, and another embedded in the complex plane.}
\label{fig: IntroDegenerateAnnuli}
\end{figure}
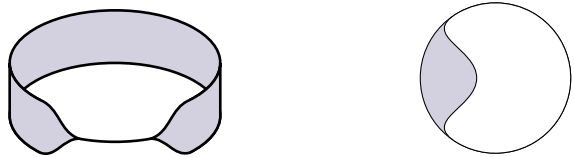

The representations of $\Diff_c(S^1)$ which are part of the data of conformal nets are expected to extend not just to the semigroup of annuli, but to the larger semigroup of degenerate annuli.
There are several ways of approaching the problem of constructing such an extension, but it is not hard to see that the extension is unique if it exists.
Moreover, operators assigned to degenerate annuli lie in $\cA(I)$ when the annuli are \emph{localized in} $I$, meaning that the incoming and outgoing parametrizations agree on $I^c$ (which necessitates that $I^c$ parametrizes a thin part of the degenerate annulus for both the incoming and outgoing boundary).
At the extreme end of things, one has totally thin annuli (i.e. diffeomorphisms), which are localized in $I$ in the above sense precisely when they lie in $\Diff(I)$, in which case they do in fact lie in the local algebra $\cA(I)$.

We now describe our method for constructing a conformal net from a unitary VOA $V$.
Suppose that the representation of $\Diff_c(S^1)$ on the Hilbert space completion $\cH_V$(which integrates the representation of $\Vir_c$ that $V$ comes equipped with) extends to a representation of the semigroup of degenerate annuli.
The connection with vertex operators arises by considering operators which insert a state into the bulk of a degenerate annulus, which should be localized when the annulus is.

To make this precise, we first consider insertions into the standard annulus $\{R > \abs{z} > r\}$ when $R > 1 > r$.
The annuli $\{1 > \abs{z} > r\}$ and $\{ R > \abs{z} > 1\}$ correspond to the bounded operators $r^{L_0}$ and $R^{-L_0}$, respectively, where $L_0$ is the conformal Hamiltonian.
If $V$ is a VOA, the operator $R^{-L_0}Y(v,z)r^{L_0}$ corresponds to inserting a state inside the annulus, and is depicted:
$$
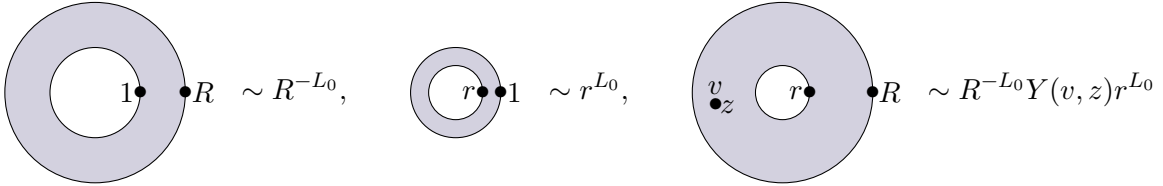
\hspace{-.1in}
\begin{tikzpicture}[scale=0.6,baseline={([yshift=-.5ex]current bounding box.center)}]
	\fill[fill=red!10!blue!20!gray!30!white] (0,0) circle (2cm);
	\draw (0,0) circle (2cm);
	\node at (2,0) {\textbullet};
	\node at (2.39,0) {$R$};

	\fill[fill=white] (0,0) circle (1cm);
	\draw (0,0) circle (1cm);
	\node at (1,0) {\textbullet};
	\node at (0.7,0) {$1$};
\end{tikzpicture}
\,\, \sim R^{-L_0},
\qquad
\begin{tikzpicture}[scale=0.6,baseline={([yshift=-.5ex]current bounding box.center)}]
	\fill[fill=red!10!blue!20!gray!30!white] (0,0) circle (1cm);
	\draw (0,0) circle (1cm);
	\node at (1,0) {\textbullet};
	\node at (1.3,0) {$1$};

	\fill[fill=white] (0,0) circle (0.6cm);
	\draw (0,0) circle (0.6cm);
	\node at (0.6,0) {\textbullet};
	\node at (0.3,0) {$r$};
\end{tikzpicture}
\,\, \sim r^{L_0}, 
\qquad
\begin{tikzpicture}[scale=0.6,baseline={([yshift=-.5ex]current bounding box.center)}]
	\fill[fill=red!10!blue!20!gray!30!white] (0,0) circle (2cm);
	\draw (0,0) circle (2cm);
	\node at (2,0) {\textbullet};
	\node at (2.37,0) {$R$};

	\fill[fill=white] (0,0) circle (0.6cm);
	\draw (0,0) circle (0.6cm);
	\node at (0.6,0) {\textbullet};
	\node at (0.3,0) {$r$};

	\node at ([shift=(207:1.37cm)]0,0.3cm) {$z$};
	\node at ([shift=(187:1.50cm)]0,0.3cm) {$v$};
	\node at ([shift=(190:1.5cm)]0,0cm) {\textbullet};
\end{tikzpicture}
\,\,
\sim
R^{-L_0}Y(v,z)r^{L_0}
$$

Instead, we will consider a pair $(B,A)$ of degenerate annuli, with the outgoing boundary of $A$ and the incoming boundary of $B$ both the unit circle, and such that the composition $B \circ A$ in the semigroup of degenerate annuli is localized in $I$.
Then the operator $BY(v,z)A$ which inserts the state $v$ at the point $z$ inside $B \circ A$ should be bounded, and localized in $I$ in the sense of algebraic quantum field theory.
$$
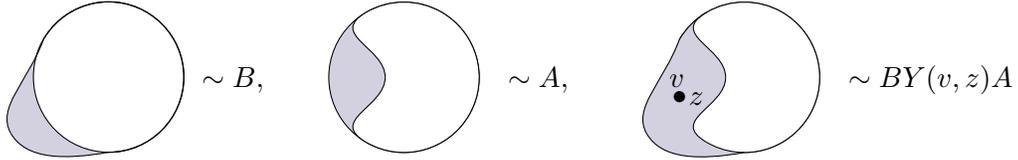
\begin{tikzpicture}[baseline={([yshift=-.5ex]current bounding box.center)}]
	\coordinate (a) at (150:1cm);
	\coordinate (b) at (270:1cm);
	\coordinate (c) at (210:1.5cm);
	\draw (0,0) circle (1cm);
	\fill[red!10!blue!20!gray!30!white] (a)  .. controls ++(250:.6cm) and ++(120:.4cm) .. (c) .. controls ++(300:.4cm) and ++(190:.6cm) .. (b) arc (270:510:1cm);
	\filldraw[fill=white] (0,0) circle (1cm);
	\draw ([shift=(270:1cm)]0,0) arc (270:510:1cm);
	\draw (a) .. controls ++(250:.6cm) and ++(120:.4cm) .. (c);
	\draw (b) .. controls ++(190:.6cm) and ++(300:.4cm) .. (c);
\end{tikzpicture}
\,\, \sim B,
\qquad
\begin{tikzpicture}[baseline={([yshift=-.5ex]current bounding box.center)}]
	\coordinate (a) at (120:1cm);
	\coordinate (b) at (240:1cm);
	\coordinate (c) at (180:.25cm);
	\fill[fill=red!10!blue!20!gray!30!white] (0,0) circle (1cm);
	\draw (0,0) circle (1cm);
	\fill[fill=white] (a)  .. controls ++(210:.6cm) and ++(90:.4cm) .. (c) .. controls ++(270:.4cm) and ++(150:.6cm) .. (b) -- ([shift=(240:1cm)]0,0) arc (240:480:1cm);
	\draw ([shift=(240:1cm)]0,0) arc (240:480:1cm);
	\draw (a) .. controls ++(210:.6cm) and ++(90:.4cm) .. (c);
	\draw (b) .. controls ++(150:.6cm) and ++(270:.4cm) .. (c);
\end{tikzpicture}
\,\, \sim A,
\qquad
\begin{tikzpicture}[baseline={([yshift=-.5ex]current bounding box.center)}]
	\coordinate (a1) at (150:1cm);
	\coordinate (b1) at (270:1cm);
	\coordinate (c1) at (210:1.5cm);
	\fill[red!10!blue!20!gray!30!white] (a1)  .. controls ++(250:.6cm) and ++(120:.4cm) .. (c1) .. controls ++(300:.4cm) and ++(190:.6cm) .. (b1) arc (270:510:1cm);
	\draw ([shift=(270:1cm)]0,0) arc (270:510:1cm);
	\coordinate (a) at (120:1cm);
	\coordinate (b) at (240:1cm);
	\coordinate (c) at (180:.25cm);
	\fill[fill=red!10!blue!20!gray!30!white] (0,0) circle (1cm);
	\fill[fill=white] (a)  .. controls ++(210:.6cm) and ++(90:.4cm) .. (c) .. controls ++(270:.4cm) and ++(150:.6cm) .. (b) -- ([shift=(240:1cm)]0,0) arc (240:480:1cm);
	\draw ([shift=(240:1cm)]0,0) arc (240:480:1cm);
	\draw (a) .. controls ++(210:.6cm) and ++(90:.4cm) .. (c);
	\draw (b) .. controls ++(150:.6cm) and ++(270:.4cm) .. (c);
	\draw (a) arc (120:150:1cm) (a1) .. controls ++(250:.6cm) and ++(120:.4cm) .. (c1);
	\draw (b1) .. controls ++(190:.6cm) and ++(300:.4cm) .. (c1);
	
%
	\node at (197:.9cm) {\textbullet};
	\node at (204:.7cm) {$z$};
	\node at (183:.9cm) {$v$};
%
%
%
\end{tikzpicture}
\,\, \sim
BY(v,z)A
$$
We now attempt to construct a conformal net $\cA_V$ whose local algebras $\cA_V(I)$ are generated by operators of the form $BY(v,z)A$, as $v$ runs over all of the states of $V$, $(B,A)$ runs over all pairs with $B \circ A$ localized in $I$, and $z \in B \circ A$.

The main result of \cite{GRACFT1} (adjusted slightly in \S3-4 of this article, and extended significantly in \S7) asserts that this construction produces conformal nets from many unitary VOAs.
The minor adjustment in this article is that we do not focus on the semigroup of degenerate annuli, and instead consider families of pairs of operators $(B,A)$ which are \emph{generalized annuli}%
\footnote{%
The definition of generalized annulus is introduced in order to avoid certain scalar ambiguities arising from the projective nature of the representations of $\Diff_c(S^1)$ under consideration.
These ambiguities were not problematic in \cite{GRACFT1}, but they make it quite difficult to formulate Definition \ref{def: piM} capturing the relationship between VOA modules and conformal net representations.
While the introduction of generalized annuli accounts for increased complexity up front in Section \ref{sec: generalized annuli}, it significantly simplifies the statements of results for the remainder of the article.
}%
.
Generalized annuli are operators which abstract away the specifics of representations of the semigroup of degenerate annuli, allowing us to place greater focus on the vertex operators.
We show in Section \ref{sec: example of system of generalized annuli} that the operators assigned to degenerate annuli in \cite{GRACFT1}, which are given by explicit exponentials of the smeared Virasoro field, satisfy the axioms of a generalized annuli.

Specifically, we will require as input a system of generalized annuli $\scA_I$, which provides for every interval $I$ a family of pairs of operators $(B,A)$ on $\cH_V$ such that $BA$ lies in $\cA_c(I)$, where $\cA_c$ is the Virasoro net.
We require that each pair $(B,A)$ be equipped with a choice of interior, which is an open set $\Int(B,A) \subset \bbD$, representing the annulus.

Given a choice of a system of generalized annuli, one says that $V$ has bounded localized vertex operators if whenever $(B,A) \in \scA_I$ and $z \in \Int(B,A)$, the operator $BY(v,z)A \in \cB(\cH_V)$, and if the algebras
$$
\cA_V(I) := \{ BY(v,z)A : v \in V, \, (B,A) \mbox{ and } z \mbox{ as above } \}^{\prime\prime}
$$
satisfy the locality condition that $\cA_V(I)$ and $\cA_V(J)$ commute when $I$ and $J$ are disjoint.
In this case, $\cA_V$ is a conformal net.

In \cite{GRACFT1} (adjusted slightly in Section \ref{sec: BLVO}), we showed that the property of bounded localized vertex operators is inherited by tensor products and subalgebras, and that the free fermion ($bc$-system) VOA has bounded localized vertex operators with respect to the system of generalized annuli $\scA_I$ constructed in Section \ref{sec: example of system of generalized annuli}.
From there, we deduced that many familiar VOAs also enjoy this property.
In this article, we extend that result to include many more VOAs, including all WZW models, in Section \ref{sec: examples} (see also Section \ref{sec: main results} for a summary).

\subsection{Main results}\label{sec: main results}

In this article, we investigate the representation theory of the conformal nets $\cA_V$ described in the previous section, and its connection with the representation theory of the unitary VOA $V$.
Recall that a representation of a conformal net is given by a family of compatible representation $\pi_I$ of all of the local algebras $\cA(I)$, and that a unitary module $M$ for a unitary VOA is an inner product space $M$ equipped with a state-field correspondence $Y^M:V \to \End(M)[[x^{\pm 1}]]$.
Given a VOA module $M$, the corresponding representation $\pi^M$ of $\cA_V$, if it exists, is the one which satisfies
$$
\pi^M(B Y(v,z) A) = B Y^M (v,z) A,
$$
where we must use the localizability of the generalized annulus $(B,A)$ and the Virasoro conformal net to allow $A$ and $B$ to act on $\cH_M$ (see Section \ref{sec: generalized annuli}).
We establish several results about the correspondence $M \longleftrightarrow \pi^M$ in Section \ref{sec: BLVO modules}, the most important of which (Theorem \ref{thmConstructionOfNetRepresentations}) provides examples of modules $M$ for which $\pi^M$ exists.

\begin{thmalpha}\label{thm: intro module existence}
Let $V$ be a simple unitary vertex operator superalgebra with bounded localized vertex operators, let $W$ be a unitary subalgebra of $V$, and let $M$ be a $W$-submodule of $V$.
Then $\pi^M$ exists as a representation of $\cA_W$.
\end{thmalpha}

This also extends to $W$-submodules of $V$-modules $N$ such that $\pi^N$ exists as a representation of $\cA_V$.
Using this, we are able to find examples of VOAs $V$ such that $\pi^M$ exists for every $V$-module $M$.
These examples include the WZW models of type $A$ and $E$, and as a corollary we obtain the new result that every unitary module for the affine Lie algebra $\mathfrak{e}_n$ at level $k$ ($n=6,7,8$) exponentiates to a representation of the $E_8$ conformal net (see Section \ref{sec: modules and local equivalence}).

An important consequence of the construction of representations $\pi^M$ from Theorem \ref{thm: intro module existence} is that it provides a link between intertwining operators (for $V$) and local intertwiners (for $\cA_V$).
Recall that the local intertwiners between two representations $\pi$ and $\lambda$ of a net $\cA$ are bounded maps in $\Hom_{\cA(I)}(\cH_\pi,\cH_\lambda)$ for some interval $I$.
If $N$ and $K$ are unitary $V$-modules such that $\pi^N$ and $\pi^K$ exist, we define $I_{loc} \binom{K}{M \, N}$ to be the class of intertwining operators $\cY$ such that
$$
B \cY(a,z) A \in \Hom_{\cA(I)}(\cH_N, \cH_K)
$$
whenever $(B,A) \in \scA_I$ and $z \in \Int(B,A)$.
We demonstrate the existence of such intertwining operators in Theorem \ref{thm: intertwiners local}.

\begin{thmalpha}\label{thm: intro local intertwiners}
Let $V$ be a simple unitary vertex operator superalgebra, let $W$ be a unitary subalgebra of $V$, and let $M$, $N$, and $K$ be simple $W$-submodules of $V$.
Then the intertwining operator $\cY \in I \binom{K}{M \, N}$ obtained by projecting $Y^V$ onto these submodules lies in $I_{loc} \binom{K}{M \, N}$.
\end{thmalpha}

As a result, the representations $\pi^N$ and $\pi^K$ come with lots of local intertwiners which are explicitly derived from vertex operators, a necessary step in analyzing the fusion rules of these representations (and a step which appears intractable, except in specific families of models, using other approaches).

Using the localized intertwining operators constructed in Theorem \ref{thm: intro local intertwiners}, we study the problem of extending the property of bounded localized vertex operators from a VOA $W$ to a simple current extension of $W$.
While it is commonplace for subalgebras to inherit nice analytic properties from larger algebras, the reverse process is inherently much more difficult.

If $M$ is a self-dual simple current of a unitary VOA $W$ (intuitively, if $M \boxtimes M = W$), then for every binary string $i \in \bbF_2^n$ we associate the $W^{\otimes n}$ module $M_i = M_{i(1)} \otimes \cdots \otimes M_{i(n)}$ (where $M_1 = M$ and $M_0 = W$).
A simple current extension of $W^{\otimes n} \subset V$ of code type, with respect to a code $C \subset \bbF_2^n$ and the module $M$, is one where $V = \bigoplus_{i \in C} M_i$ as a $W^{\otimes n}$-module.
We consider when the simple current $M$ is either bosonic, fermionic, or semionic, which essentially corresponds to having conformal weights lying in $\bbZ$, $\tfrac12 + \bbZ$, or $\pm \tfrac14 + \bbZ$, respectively.
The following extension result for bounded localized vertex operators is a combination of Proposition \ref{prop: BLVO for easier codes} and Theorem \ref{thm: BLVO for harder codes}.

\begin{thmalpha}
Let $W$ be a simple vertex operator algebra, and let $M$ be a self-dual simple current which is bosonic, fermionic, or semionic.
Suppose there exists a simple unitary vertex operator superalgebra $V$ with bounded localized vertex operators, and an extension $W^{\otimes m} \subset V$ (where $m=1$ if $M$ is bosonic or fermionic, and $m=2$ if $M$ is semionic).
Then for every $n \in \bbZ_{> 0}$, any simple current extension of $W^{\otimes mn}$ of code type based on the module $M$ has bounded localized vertex operators.
\end{thmalpha}

In Section \ref{sec: examples lattices in codes}, we apply this theorem to produce many new examples of VOAs with bounded localized vertex operators.
We also verify that the conformal nets produced by our construction are isomorphic to the ones constructed in \cite{CKLW18} in all of the examples considered in Section \ref{sec: examples} which both constructions apply to.

\begin{Corollary*}
The class of VOAs with bounded localized vertex operators (with respect to the system of generalized annuli constructed in Section \ref{sec: example of system of generalized annuli}) includes:
\begin{itemize}
\item all WZW models $V(\mathfrak{g}, k)$ corresponding to simple Lie algebras $\mathfrak{g}$ at positive integral level $k$
\item all simple current extensions of copies of the Ising model $L(\frac12, 0)^{\otimes n}$
\item all simple current extensions of copies of the WZW model $V(\mathfrak{a}_1,1)^{\otimes n}$ (i.e. VOAs associated to lattices with $A_1^n$-framings)
\end{itemize}
\end{Corollary*}

Of course, by earlier results, all subalgebras of tensor products of these algebras again have bounded localized vertex operators as well. 
In Section \ref{sec: modules and local equivalence}, we combine the Corollary and Theorem \ref{thmConstructionOfNetRepresentations} with results from the VOA literature (most notably, \cite{KrauelMiyamoto15} and \cite{ArakawaCreutzigLinshaw19}) to produce examples of VOAs $V$ such that $\pi^M$ exists for every $V$-module $M$, and in particular resolve the local equivalence problem for representations of the loop group $LG$ when $G$ is of type $E_6$, $E_7$, or $E_8$.

\subsection{Structure of the article}

The article is structured as follows.
In Section \ref{secPreliminaries}, we review the definitions of conformal nets and unitary vertex operator algebras, as well as their $\bbZ/2\bbZ$-graded relatives.
In Sections \ref{sec: generalized annuli} and \ref{sec: BLVO}, we review and extend the definition of VOAs with bounded localized vertex operators, and the construction of conformal nets from such VOAs.
In particular, Section \ref{sec: generalized annuli} gives the definition of a system of generalized annuli as well as our motivating example of a system of generalized annuli comprised of explicit exponentials of smeared Virasoro fields.
In Section \ref{sec: BLVO}, we study insertion operators relative to systems of generalized annuli, and this leads us to the definition of bounded localized vertex operators.
In Section \ref{sec: BLVO modules}, we define the correspondence between $V$-modules and $\cA_V$-representations (denoted $M \mapsto \pi^M$), and provide existence theorems for representations $\pi^M$.
We also study the relationship between intertwining operators (for VOAs) and local intertwiners (for conformal nets).
In Section \ref{sec: code extensions}, we show that certain simple current extensions preserve bounded localized vertex operators, and in Section \ref{sec: examples} we apply our theorems to specific models.

\subsection{Acknowledgments}

A significant portion of the research for this article was undertaken at the Max Planck Institute for Mathematics, Bonn, between 2014 and 2016, and I would like to gratefully acknowledge their hospitality and support during this period.
This work was also supported in part by an AMS-Simons travel grant.
I am grateful to many people for enlightening conversations which improved this article, including Marcel Bischoff, Sebastiano Carpi, Thomas Creutzig, Terry Gannon, Andr\'{e} Henriques, and Robert McRae.

\newpage

\section{Preliminaries}\label{secPreliminaries}

\subsection{Unitary vertex operator superalgebras and their representation theory}

\subsubsection{Unitary vertex operator superalgebras}

In this section we will present definitions and basic results pertaining to unitary vertex operator superalgebras and their modules.
The definitions first appeared in \cite{DongLin14} (for the case of vertex operator algebras) and \cite{AiLin17} (for the general case).
The theory of unitary vertex operator algebras was extended in \cite[\S4-5]{CKLW18}, and the results presented below are adapted from these three references, along with \cite{FHL93}.

\begin{Definition}[Vertex operator superalgebras]
A  vertex operator superalgebra is given by:
\begin{enumerate}
\item a $\Z/2\Z$-graded vector space $V = V^0 \oplus V^1$. Elements of $V^0 \cup V^1$ are called parity homogeneous vectors, and elements of $V^0$ (resp. $V^1$) are called even (resp. odd) vectors.
If $a \in V^i$, we denote the parity $p(a) = i$.
\item even vectors $\Omega,\nu \in V^0$ called the \emph{vacuum vector} and the \emph{conformal vector}, respectively.
\item a state-field correspondence $Y:V  \to \End(V)[[x^{\pm 1}]]$, denoted 
\begin{equation}\label{eqnStateFieldForm}
Y(a,x) = \sum_{n \in \Z} a_{(n)} x^{-n-1}.
\end{equation}
Here $\End(V)[[x^{\pm 1}]]$ is the vector space of formal series of the form \eqref{eqnStateFieldForm}.
\end{enumerate}
This data must satisfy:
\begin{enumerate}
\item For every $a \in V$, if $a$ is even (resp. odd) then $a_{(n)}$ is even (resp. odd) for all $n \in \Z$.
\item For every $a,b \in V$, we have $a_{(n)}b = 0$ for $n$ sufficiently large. 
\item For every $a \in V$, we have $a_{(n)}\Omega = 0$ for $n \ge 0$ and $a_{(-1)}\Omega = a$.
\item $Y(\Omega,x) = 1_V$. That is, $\Omega_{(n)} = \delta_{n,-1} 1_V$.
\item For every  $a,b,c \in V$ and every $m,k,n \in \Z$, we have the Borcherds (or Jacobi) identity:
\begin{align*}
\sum_{j = 0}^\infty \binom{m}{j} \big(a_{(n+j)}b\big)_{(m+k-j)}&c = \sum_{j=0}^\infty (-1)^j \binom{n}{j} a_{(m+n-j)}b_{(k+j)}c \\ 
&-(-1)^{p(a)p(b)} \sum_{j=0}^\infty(-1)^{j+n} \binom{n}{j} b_{(n+k-j)}a_{(m+j)}c.
\end{align*}
\item If we write $Y(\nu,x) = \sum_{n \in \Z} L_n x^{-n-2}$, then the $L_n$ give a representation of the Virasoro algebra. 
That is, 
$$
[L_m,L_n] = (m-n)L_{m+n} + \tfrac{c}{12} (m^3-m)\delta_{m,-n}1_V
$$
for a number $c \in \C$, called the \emph{central charge}.
\item If we write $V_\alpha = \ker (L_0 - \alpha 1_V)$, then we have a decomposition of $V$ as an algebraic direct sum
$$
V^0 = \bigoplus_{\alpha \in \Z_{\ge 0}} V_\alpha, \qquad V^1 = \bigoplus_{\alpha \in \tfrac12 + \Z_{\ge 0}} V_\alpha
$$
with $\dim V_\alpha < \infty$.
\item For every $a \in V$ we have $[L_{-1}, Y(a,x)] = \frac{d}{dx} Y(a,x)$.
\end{enumerate}
\end{Definition}
We will often abuse terminology by referring to $V$ as a vertex operator superalgebra, instead of referring to the quadruple $(V, Y, \Omega, \nu)$.
If $V^1 = \{0\}$, then $V$ is called a \emph{vertex operator algebra}, and we say that $V$ is \emph{even}.

If $a \in V_\alpha$, then we say that $a$ is homogeneous of conformal weight $\alpha =: \Delta_a$.
It follows from the definition that if $a$ is homogeneous, then $a_{(n)}V_{\beta} \subset V_{\beta - n - 1 + \Delta_a}$.
More generally, expressions involving $\Delta_a$ and $p(a)$ should be understood as being for homogeneous $a$, and extended linearly otherwise.

We will use the special cases of the Borcherds identity corresponding to $m=0$ and $n=0$, which say:
\begin{equation}\label{eqn: BPF}
\big(a_{(n)}b\big)_{(k)} c= \sum_{j=0}^\infty (-1)^j \binom{n}{j} \left(a_{(n-j)}b_{(k+j)} - (-1)^{p(a)p(b)+n}b_{(n+k-j)}a_{(j)}\right)c
\end{equation}
and
\begin{equation}\label{eqn: BCF}
a_{(m)}b_{(k)}c - (-1)^{p(a)p(b)} b_{(k)}a_{(m)}c = \sum_{j =0}^\infty \binom{m}{j} \big(a_{(j)}b\big)_{(m+k-j)}c,
\end{equation}
which are called the Borcherds Product Formula and Borcherds Commutator Formula, respectively.

\begin{Definition}[Subalgebras, ideals, etc.]
A subalgebra of a vertex operator superalgebra $V$ is a subspace $W \subset V$ such that:
\begin{enumerate}
\item $W = (W \cap V^0) \oplus (W \cap V^1)$
\item $a_{(n)}b \in W$ for all $a,b \in W$ and $n \in \Z$.
\item $\Omega \in W$
\end{enumerate}
If $\nu \in W$, then $W$ is called a \emph{conformal subalgebra} of $V$.

A vertex subalgebra $W$ is called an \emph{ideal} if we have $a_{(n)}b \in W$ for every $a \in V$ and $b \in W$.
A  vertex operator superalgebra $V$ is called \emph{simple} if its only ideals are $\{0\}$ and $V$.

A homomorphism (resp. antilinear homomorphism) from a vertex operator superalgebra $V$ to a vertex operator superalgebra $W$ is a complex linear (resp. antilinear) map $\phi:V \to W$ which satisfies $\phi(\Omega_V) = \Omega_W$, $\phi(\nu_V) = \nu_W$, and $\phi(a_{(n)}b) = \phi(a)_{(n)}\phi(b)$ for all $a,b \in V$.
We also have the obvious notion of (antilinear) isomorphism and automorphism.
\end{Definition}

Since homomorphisms preserve the $L_0$ grading, which in turn determines the parity grading, our definition of homomorphism is always even.

The grading operator $\Gamma = (-1)^{2L_0}$ is always an automorphism of a vertex operator superalgebra.
We will make use of the Klein transform $\kappa$, a square root of $\Gamma$ defined by
$$
\kappa = \frac{1 - i \Gamma}{1 - i},
$$
which acts on homogeneous vectors $a$ by 
$
\kappa a = i^{p(a)} a.
$
Throughout the paper we will use the convention that $(-1)^z = e^{i \pi z}$, and so $\kappa a = (-1)^{\tfrac12 p(a)}a$.

\begin{Definition}\label{defUnitaryVOSA}
A \emph{unitary vertex operator superalgebra} is a vertex operator superalgebra $V$ equipped with an inner product and an antilinear automorphism $\theta$ satisfying:
\begin{enumerate}
\item $\bip{b, Y(\theta a, \overline{x})c} = \bip{Y(e^{x L_1} (-x^{-2})^{L_0} \kappa a, x^{-1})b, c}$ for all $a,b,c \in V$.
\item $\ip{\Omega,\Omega} = 1$
\end{enumerate}
We write $\cH_V$ for the Hilbert space completion of $V$, consisting of vectors $\xi = \sum_{n \in \tfrac12\bbZ} v_n$ with $v_n \in V_n$ and $\sum \norm{v_n}^2 < \infty$.

An isomorphism $\phi:V \to W$ of unitary vertex operator superalgebras is called unitary if $\ip{\phi a, \phi b} = \ip{a, b}$ for all $a,b \in V$.
If $V^1 = \{0\}$ then we refer to $V$ as a \emph{unitary vertex operator algebra}.

A subalgebra $W$ of $V$ is called a \emph{unitary subalgebra} if $\theta(W) \subset W$ and $L_1 W \subset W$.
\end{Definition}
Note that $x$ is treated as a formal, complex variable in the statement of the invariance property, and that $(-1)^{L_0}$ is defined as $e^{i \pi L_0}$ by our convention above.

Let $(V,Y,\Omega,\nu,\ip{\,\cdot\, , \, \cdot \,}, \theta)$ be a simple unitary vertex operator superalgebra, and let $W$ be a unitary subalgebra.
Then we may obtain a unitary vertex operator superalgebra structure on $W$ as follows.
Let $e_W$ be the orthogonal projection of $\cH_V$ onto $\cH_W$, let $Y^W$ and $\theta^W$ be the restrictions of $Y$ and $\theta$ to $W$, and let $\nu^W = e_W \nu$.
The following is \cite[Prop. 5.29]{CKLW18}:
\begin{Proposition}\label{propUnitarySubalgebraIsUnitaryVOA}
$(W, Y^W, \nu^W, \ip{\,\cdot\, , \, \cdot \,}, \theta^W)$ is a simple unitary vertex operator superalgebra.
Moreover, $L_i^{W} = L_i|_{W}$ for $i \in \{-1, 0 ,1\}$, and in particular the $\tfrac12 \Z_{\ge 0}$ grading of $W$ coincides with the one inherited from $V$.
\end{Proposition}

Note that unitary subalgebras of simple unitary vertex operator superalgebras are again simple by \cite[Prop. 5.3]{CKLW18}:
\begin{Proposition}\label{propUnitarySimple}
Let $V$ be a unitary vertex operator superalgebra.
Then $V$ is simple if and only if $V_0 = \C \Omega$.
\end{Proposition}

We now briefly introduce the tensor product and coset constructions for unitary vertex operator superalgebras; for more detail, see \cite[\S2.2]{GRACFT1} or \cite[\S5]{CKLW18}.

\begin{Definition}[Coset]
Let $(V,Y,\Omega, \nu)$ be a vertex operator superalgebra and let $W$ be a subalgebra.
The \emph{coset} $W^c \subset V$ is given by
$$
W^c = \{a \in V : Y(a,x)Y(b,y) - (-1)^{p(a)p(b)}Y(b,y)Y(a,x) = 0 \mbox{ for all } b \in W \}.
$$
\end{Definition}

\begin{Proposition}[{\cite[Ex. 5.27]{CKLW18}}]\label{propUVOSACoset}
Let $W$ be a simple unitary vertex operator superalgebra, and let $W \subset V$ be a unitary subalgebra.
Then $W^c$ is a unitary subalgebra and $\nu = \nu^W + \nu^{W^c}$.
\end{Proposition}

Unitary tensor products of vertex operator superalgebras were discussed in \cite{AiLin17}:

\begin{Proposition}[\cite{AiLin17}]\label{propUVOSATensorProduct}
For $i \in \{1,2\}$, let $(V_i, Y^i, \Omega_i, \nu_i, \ip{\, \cdot\, , \, \cdot \,}, \theta_i)$ be unitary vertex operator superalgebras.
For $a_i \in V_i$ homogeneous vectors with parity $p(a_i)$, let $Y(a_1 \otimes a_2, x) = Y^1(a_1, x)\Gamma_{V_1}^{p(a_2)} \otimes Y^2(a_2, x)$.
Then $(V_1 \otimes V_2, Y, \Omega_1 \otimes \Omega_2, \nu_1 \otimes \Omega_2 + \Omega_1 \otimes \nu_2, \ip{\, \cdot \, , \, \cdot \,}, \theta_1 \otimes \theta_2)$ is a unitary vertex operator superalgebra.
\end{Proposition}
We will sometimes use the notation
$$
Y^1(a_1, x) \hotimes Y^2(a_2, x) := Y^1(a_1, x)\Gamma_{V_1}^{p(a_2)} \otimes Y^2(a_2, x).
$$
Note that by Proposition \ref{propUnitarySimple}, the tensor product of simple unitary vertex operator superalgebras is again simple.

The following observation is well-known, but we were unable to find a statement in the literature, and so a proof was given in \cite[Prop. 2.21]{GRACFT1}.
\begin{Proposition}\label{propWandCommutantGenerateTensorProduct}
Let $V$ be a simple unitary vertex operator superalgebra, and let $W$ be a unitary subalgebra. 
Let $\tilde W = \Span \{ a_{(-1)}b : a \in W, b \in W^c\}$.
Then $\tilde W$ is a unitary conformal subalgebra of $V$, unitarily isomorphic to $W \otimes W^c$.
\end{Proposition}

\subsubsection{Unitary modules}

\begin{Definition}
Let $V$ be a vertex operator superalgebra.
A \emph{generalized $V$-module} is given by a $\Z/2\Z$-graded vector space $M = M^0 \oplus M^1$ along with a state-field correspondence $Y^M:V \to \End(M)[[x^{\pm 1}]]$, written
$$
Y^M(a,x) = \sum_{n \in \Z} a_{(n)}^M x^{-n-1},
$$
which is required to be linear, and satisfy the following additional properties.
\begin{enumerate}
\item $Y^M(\Omega,x) = 1_M$.
\item If $a \in V$ is even (resp. odd) then $a^M_{(n)}$ is even (resp. odd) for all $n \in \Z$.
\item For every $a \in V$ and $b \in M$, we have $a^M_{(n)}b = 0$ for $n$ sufficiently large.
\item For every homogeneous $a,b \in V$, $c \in M$, and $m,k,n \in \Z$, the Borcherds(/Jacobi) identity holds:
\begin{align}
\sum_{j = 0}^\infty \binom{m}{j} \big(a_{(n+j)}b\big)^M_{(m+k-j)}&c = \sum_{j=0}^\infty (-1)^j \binom{n}{j} a_{(m+n-j)}^Mb_{(k+j)}^M c \nonumber\\ 
&-(-1)^{p(a)p(b)} \sum_{j=0}^\infty(-1)^{j+n} \binom{n}{j} b_{(n+k-j)}^M a_{(m+j)}^M c. \label{eqnBorcherdsIdentityModule}
\end{align}
\item $M$ is compatibly graded by conformal weights and parity.
That is, if we write $Y^M(\nu,x) =: \sum_{n \in \Z} L_n^M x^{-n-2}$,  $M_\alpha := \ker(L_0^M - \alpha 1_M)$, and $M_\alpha^i = M^i \cap M_\alpha$ then we require that 
$$
M^i = \bigoplus_{\alpha \in \bbC} M^i_\alpha.
$$

If $a \in M^i_\alpha$ then we say that $a$ is homogeneous with conformal weight $\Delta_a := \alpha$ and parity $p(a):=i$.
\end{enumerate}

A generalized module is called a \emph{(strong, separable) module} if each space $M_\alpha$ is finite-dimensional, and $M_\alpha$ is non-zero for at most countably many $\alpha$.
\end{Definition}

There are obvious notions of submodules and direct sums of $V$-modules, as well as (antilinear) $V$-module homomorphisms.
We insist, however, that $V$-module homomorphisms be even (i.e. preserve the parity grading).

If $M$ has no proper, non-trivial submodules then it is called a \emph{simple} module.

As with vertex operator superalgebras, we define a grading operator and Klein transform acting on $a \in M$ by
$$
\Gamma a = (-1)^{p(a)} a, \qquad \kappa a = i^{p(a)} a = (-1)^{\tfrac12 p(a)} a.
$$

We also have the following well-known basic properties of modules \cite[\S4]{FHL93}.
\begin{Proposition}\label{propBasicModuleOperatorProperties}
Let $M$ be a module over a vertex operator superalgebra $V$, and let $a,b \in V$ and $d \in M$ be homogeneous elements.
\begin{enumerate}
\item $Y^M$ satisfies a Borcherds product formula:
\begin{equation}\label{eqnModuleBorcherdsProductFormula}
\big(a_{(n)}b\big)^M_{(k)} d= \sum_{j=0}^\infty (-1)^j \binom{n}{j} \left(a_{(n-j)}^Mb_{(k+j)}^M - (-1)^{p(a)p(b)+n}b_{(n+k-j)}^M a_{(j)}^M\right)d
\end{equation}
for all $n,k \in \Z$.
\item $Y^M$ satisfies a Borcherds commutator formula:
\begin{equation}\label{eqnModuleBorcherdsCommutatorFormula}
a_{(m)}^Mb_{(k)}^Md - (-1)^{p(a)p(b)} b_{(k)}^Ma_{(m)}^Md = \sum_{j =0}^\infty \binom{m}{j} \big(a_{(j)}b\big)^M_{(m+k-j)}d
\end{equation}
for all $m,k \in \Z$.
\item $Y^M$ satisfies the $L_{-1}^M$-derivative property:
\begin{equation}\label{eqnModuleDerivativeCommutator}
\big[L_{-1}^M, Y^M(a,x)\big] = Y^M(L_{-1} a, x) = \frac{d}{dx} Y^M(a,x)
\end{equation}
\item If $a$ is homogeneous of conformal weight $\Delta_a$, then $a_{(n)}^M M_\alpha \subset M_{\alpha - n -1 + \Delta_a}$.
\item The modes $L_n^M = \nu^M_{(n+1)}$ of the conformal vector satisfy the Virasoro relations 
$$
[L^M_m,L^M_n] = (m-n)L^M_{m+n} + \tfrac{c}{12} (m^3-m)\delta_{m,-n}1_M
$$
where $c$ is the central charge of $V$.
\end{enumerate}
\end{Proposition}

We will be interested in \emph{unitary} modules over \emph{unitary} vertex operator superalgebras, which first appeared in \cite{DongLin14,AiLin17}.
\begin{Definition}
Let $V$ be a unitary vertex operator superalgebra, and let $M$ be a generalized $V$-module.
A sesquilinear form $\ip{ \, \cdot \, , \, \cdot \,}$ on $M$ is called \emph{invariant} if
\begin{equation}\label{eqnModuleInnerProductInvariance}
\bip{b,Y^M(\theta a,\overline{x})c} = \bip{Y^M(e^{x L_1} (-x^{-2})^{L_0} \kappa a, x^{-1})b,c},
\end{equation}
for all $a \in V$ and $b,c \in M$.
We call $M$ a \emph{unitary} generalized module if it is equipped with an invariant inner product.
A (strong, separable) unitary  module is a module which is unitary in the same sense.
\end{Definition}
Observe that the requirement we imposed on $V$-modules that $L^M_0$ have only countably many distinct eigenvalues is equivalent to the separability of the Hilbert space completion $\cH_M$ in the presence of the other properties.

There are natural notions of orthogonal direct sum and unitary isomorphism of unitary modules.
When talking about unitary vertex operator superalgebras and unitary modules, we will reserve the symbol $\bigoplus$ for orthogonal direct sums.

As one would hope, unitary modules provide unitary representations of the Virasoro algebra (see \cite[Lem. 2.5]{DongLin14}).
\begin{Proposition}
Let $M$ be a unitary generalized module over a unitary vertex operator superalgebra $V$ with conformal vector $\nu$, and let $L_n^M = \nu_{(n+1)}^M$ be the representation of the Virasoro algebra on $M$.
Then we have
\begin{enumerate}
\item
$
\bip{L_n^M a, b} = \bip{a, L_{-n}^Mb}
$
for all $a,b \in M$ and $n \in \Z$.
\item
We have an orthogonal decomposition $M = \bigoplus_{\alpha \in \R_{\ge 0}} M_\alpha$.
\end{enumerate}
\end{Proposition}
\begin{proof}
The first item follows immediately from the definition of invariant inner product, and it follows that distinct eigenspaces for $L_0$ are orthogonal.
All that remains is to show that all eigenvalues of $L_0$ are non-negative real numbers.
If not, and we had an eigenvector $v$ of $L_0$ with negative eigenvalue, then since $L_n v = 0$ for $n$ sufficiently large we could find an eigenvector of $L_0$ with negative eigenvalue such that $L_nv = 0$ for $n \ge 1$.
Such a vector would generate a highest weight module which violated the classification of unitary representations of the Virasoro algebra (see \cite[Lec. 8]{KacRaina}).
\end{proof}

Complete reducibility for unitary modules was shown in \cite[Prop. 2.2]{AiLin17}.
\begin{Proposition}\label{prop: complete reducibility}
Let $M$ be a unitary module over a unitary vertex operator superalgebra $V$.
If $N \subset M$ is a submodule, then
$$
N^\perp := \{m \in M : \ip{m,n} = 0 \mbox{ for all } n \in N\}
$$
is also a submodule and $M = N \oplus N^\perp$.
Every unitary $V$ module can be decomposed as an at most countable orthogonal direct sum of simple modules.
\end{Proposition}
\begin{proof}
Everything except the restriction to countably many simple modules was shown in \cite[Prop. 2.2]{AiLin17} (and this reference did not include the requirement that $\cH_V$ be separable).
However it is clear that if $\cH_M$ is separable, then it cannot decompose into an uncountable orthogonal direct sum of non-zero Hilbert spaces, so the decomposition of $M$ into simple $V$-modules must be countable.
\end{proof}

Either via the same argument, or as a special case, we have:
\begin{Corollary}
Let $M$ be a unitary module over a unitary vertex operator superalgebra $V$.
Then $M$ decomposes as a countable direct sum of irreducible highest weight unitary representations of the Virasoro algebra.
\end{Corollary}
Note that the countability of the direct sum again follows from our insistence that modules have only countably many distinct conformal weights.

We now take a short detour to understand the relationship between the invariance of the inner product on $M$ and the induced map between the dual of $M$ and its complex conjugate. 

\begin{Definition}
Let $V$ be a vertex operator superalgebra, and let $M$ be a $V$-module.
The \emph{graded dual} $M^\prime$ is defined by $M^\prime = \bigoplus_{\alpha \in \C} M_\alpha^*$.
If we write $\left( \, \cdot \, , \, \cdot \,\right)$ for the pairing between $M^\prime$ and $M$, then the \emph{contragredient module} (\cite{FHL93}, and also \cite{Yamauchi2014} for the super case) structure on $M^\prime$ is the unique state-field correspondence $Y^{M^\prime}$ satisfying
$$
\big( Y^{M^\prime}(a,x)b^\prime, c \big) = \big( b^\prime, Y^M(e^{x L_1} (-x^{-2})^{L_0} \kappa b, x^{-1})a\big)
$$
for all $a \in V$, $b^\prime \in M^\prime$ and $c \in M$.
\end{Definition}

If $M$ is a complex vector space, we will write $M^\dagger$ for the complex conjugate vector space.
If $a \in M$, we will again write $a$ for the corresponding element of $M^\dagger$; the conjugation on the vector space is merely used to adjust which maps are linear and which are antilinear.
\begin{Definition}
Let $V$ be a vertex operator superalgebra equipped with PCT automorphism $\theta$, and let $M$ be a generalized $V$-module.
Then the conjugate module $M^\dagger$ is defined by $Y^{M^\dagger}(a,x) = Y^M(\theta a, x)$.
\end{Definition}
It is straightforward to verify that $M^\dagger$ is indeed a $V$-module.

\begin{Lemma}\label{lem: invariant inner product as iso}
Let $V$ be a unitary vertex operator superalgebra, and let $M$ be a $V$-module equipped with a sesquilinear form $\ip{ \, \cdot \, , \, \cdot \,}$ such that the $L_0$-eigenspaces of $M$ are orthogonal.
Then the form is invariant if and only if the map $M^\dagger \to M^\prime$ induced by the form is a homomorphism of $V$-modules.
\end{Lemma}
\begin{proof}
Note that the condition that the $L_0$-eigenspaces of $M$ are orthogonal and finite-dimensional ensures that the form induces a map $M^\dagger \to M^\prime$.
Given $c \in M$, we denote by $c^\prime \in M^\prime$ the linear functional $c^\prime(a) = \ip{a,c}$.
Now for $a \in V$ and $b,c \in M$, we have by the definition of the contragredient module 
\begin{align*}
\bip{Y^M(e^{x L_1} (-x^{-2})^{L_0} \kappa a, x^{-1})b,c}_M 
&=\big(Y^{M^\prime}(a,x)c^\prime, b\big).
\end{align*}
By the definition of the conjugate module, we have
$$
\bip{b,Y^M(\theta a,\overline{x})c}_M = \bip{b, Y^{M^\dagger}(a,\overline{x})c}_{M}  =  \big( (Y^{M^\dagger}(a,\overline{x})c)^\prime, b \big).
$$
The equality of the left-hand sides of the two equations is the definition of invariance.
On the other hand, equality of the right-hand terms is equivalent to
$$
(a_n^{M^\dagger}c)^\prime = a_n^{M^\prime}c^\prime, \quad \mbox{ for all } n \in \Z,
$$
which is in turn equivalent to the map $c \mapsto c^\prime$ intertwining the actions of $Y^{M^\dagger}$ and $Y^{M^\prime}$.
\end{proof}

\begin{Proposition}\label{prop: invariant inner product unique}
Let $V$ be a simple unitary vertex operator algebra, and let $M$ be a simple unitary $V$-module.
Then the invariant inner product on $M$ is unique up to a scalar.
Hence every isomorphism between simple unitary $V$-modules is a scalar multiple of a unitary.
\end{Proposition}
\begin{proof}
If $M$ is simple, so is $M^\dagger$, and thus when $M$ is unitarizable there is only a one-dimensional space of isomorphisms $M^\prime \cong M^\dagger$.
By Lemma \ref{lem: invariant inner product as iso}, this implies that the inner product on $M$ is unique up to scalar.
\end{proof}

If $V_1$ and $V_2$ are vertex operator algebras, and $M_i$ is a generalized $V_i$-module, then $M_1 \otimes M_2$ is naturally a generalized $V_1 \otimes V_2$-module with action 
$$
Y^{M_1 \otimes M_2}(a_1 \otimes a_2,x) =Y^{M_1}(a_1,x) \hotimes Y^{M_2}(a_2,x) := Y^{M_1}(a_1,x)\Gamma^{p(a_2)} \otimes Y^{M_2}(a_2,x).
$$
Even when $M_1$ and $M_2$ are strong modules, $M_1 \otimes M_2$ may fail to be, but if the $M_i$ are simple then $M_1 \otimes M_2$ is a simple strong module (see \cite[Cor. 4.7.3]{FHL93}).
It is straightforward to check that if the $V_i$ and $M_i$ are unitary, then $M_1 \otimes M_2$ is a unitary module under the natural tensor product (see \cite[Prop. 2.10]{DongLin14} for the even case).
As a converse, we have the following.

\begin{Proposition}\label{prop: tensor splitting modules}
Let $V_1$ and $V_2$ be unitary vertex operator superalgebras, and let $M$ be a simple unitary $V_1 \otimes V_2$ module.
Then there exist simple unitary $V_i$-modules $M_i$ such that $M$ is unitarily equivalent to $M_1 \otimes M_2$.
\end{Proposition}
\begin{proof}
By \cite[Thm. 4.7.4]{FHL93}%
\footnote{%
The condition on the rationality of the lowest eigenvalues in the cited result is not essential, and is only present because rationality of eigenvalues is included in the definition of module in \cite{FHL93}.
}%
 $M$ is isomorphic to a tensor product of simple modules $M_1 \otimes M_2$.
We suppress this isomorphism and assume $M = M_1 \otimes M_2$.
We must show that the inner product on $M$ factors as a tensor product of invariant inner products on the $M_i$.
If we fix homogeneous $a_2,b_2 \in M_2$, the form $\ip{a_1,b_1}_{M_1} := \ip{a_1 \otimes a_2, b_1 \otimes b_2}_M$ is invariant for the action of $V_1$ on $M_1$.
By Proposition \ref{prop: invariant inner product unique}, the invariant forms on $M_1$ are unique up to scalar, so we have 
$$
\ip{a_1,b_1}_{M_1} \ip{a_2,b_2}_{M_2} = \ip{a_1 \otimes a_2, b_1 \otimes b_2}_M
$$
for some form $\ip{ \cdot , \cdot}_{M_2}$.
It is then clear that the form on $M_2$ is invariant as well.
Since the form on $M$ is an inner product, and the spaces of forms on $M_i$ are one-dimensional, we may adjust each by a scalar to make them into an inner product.
\end{proof}

As a corollary, we have the following.

\begin{Proposition}\label{propTensorDecomposition}
Let $V$ be a unitary vertex operator superalgebra, let $W$ be a unitary subalgebra, and let $M$ be a unitary $V$-module.
Then there exist countable families $N_i$ (resp. $K_i$) of simple unitary $W$-modules (resp. $W^c$-modules) such that $M$ is unitarily equivalent to $\bigoplus_{i=0}^\infty N_i \otimes K_i$ as a $W \otimes W^c$ module.
\end{Proposition}
\begin{proof}
By Proposition \ref{propWandCommutantGenerateTensorProduct}, $V$ has a unitary conformal subalgebra which is unitarily equivalent to $W \otimes W^c$, and $M$ is a $W \otimes W^c$-module.
Thus it decomposes as an orthogonal direct sum of simple unitary modules by Proposition \ref{prop: complete reducibility}, and all of the simple summands are of the indicated form by Proposition \ref{prop: tensor splitting modules}.
\end{proof}

We will frequently be in the situation of Proposition \ref{propTensorDecomposition}, where $M$ is a unitary $V$-module and $W$ is a (not necessarily conformal) unitary subalgebra of $V$.
In this case, we have two actions of the Virasoro algebra on $M$, coming from $Y^M(\nu^V,x)$ and $Y^M(\nu^W,x)$.
We will write $L_n^V$ and $L_n^W$, respectively, for the modes of these two fields.
Note that both operators act on $M$, and the superscript only indicates which conformal vector produces the representation.

\begin{Definition}
Let $V$ be a unitary vertex operator superalgebra, $W$ a unitary subalgebra, and $M$ a unitary $V$-module.
Then a generalized $W$-submodule $N$ of $M$ is a $\bbZ/2$-graded subspace $N^i \subset M^i$ which is invariant under $a^M_{(n)}$ for all $a \in W$ and $n \in \bbZ$ and also under $L_0^V$.
A $W$-submodule of $M$ is a generalized $W$-submodule of $M$ which is a $W$-module under the inherited action.
\end{Definition}

Since $L_0^W$ and $L_0^V$ commute, and the $L_0^V$ eigenspaces of $M$ are finite-dimensional, it follows that $L_0^W$ is diagonalizable on $M$, so any generalized $W$-submodule of $M$ is a unitary generalized $W$-module.
The Hilbert space completion of $N$ is separable (since $\cH_M$ is separable), and thus $N$ is a $W$-submodule of $M$ precisely when the eigenspaces of $L_0^W$ are finite-dimensional.
In light of Proposition \ref{propTensorDecomposition}, we can decompose $M = \bigoplus N_i \otimes K_i$, and so the generalized $W$-submodules $N$ of $M$ are precisely of the form $N = \bigoplus N_i \otimes S_i$, where $S_i \subset K_i$ is a $L_0^{W_c}$-graded subspace of $K_i$.
If $N$ is a $W$-submodule then each $S_i$ is finite-dimensional, although this condition alone does not imply that $N$ is a module.
Observe that $M$ itself is a generalized $W$-submodule of $M$, but it will not be a (strong) $W$-submodule unless $W$ is a conformal subalgebra of $V$.

\begin{Proposition}\label{propModuleProjection}
Let $V$ be a simple unitary vertex operator superalgebra, let $W$ be a unitary subalgebra, and let $N$ be a $W$-submodule of a unitary $V$-module $M$.
Let $\cH_M$ and $\cH_N$ be the Hilbert space completions of $M$ and $N$, respectively, and let $p_N$ be the projection of $\cH_M$ onto $\cH_N$. 
Then we have:
\begin{itemize}
\item[i)] $p_N M \subseteq N$, and $p_N$ is an even map when regarded as an endomorphism of $M$.
\item[ii)] For all $a \in W$ and $n \in \Z$, $a^M_{(n)}$ commutes with $p_N$ as an endomorphism of $M$.
\item[iii)] $p_N$ commutes with $L_0^V$ as an endomorphism of $M$
\end{itemize}
\end{Proposition}
\begin{proof}
As remarked above, by Proposition \ref{propTensorDecomposition} we can write $M = \bigoplus N_i \otimes K_i$ as a module for $W \otimes W^c$, and $N = \bigoplus N_i \otimes S_i$ where $S_i$ are finite-dimensional $L_0^{W_c}$-graded and parity-graded subspaces of $K_i$.
Then if $p_i$ is the projection of $\cH_{K_i}$ onto $S_i$, we have $p_N = \bigoplus 1 \otimes p_i$.
Thus indeed we have $p_N M \subset N$, and since the $S_i$ are parity-graded, $p_N$ is even.
For $a \in W$, we have $a^M_{(n)} = \bigoplus a^{M_i}_{(n)} \otimes 1$, and thus (ii) holds.
Finally, since $S_i$ is $L_0^{W_c}$ graded, $p_N$ commutes with $L_0^{W_c}$.
By (ii), $p_N$ commutes with $L_0^W$, and thus $p_N$ commutes with $L_0^V = L_0^W + L_0^{W_c}$.
\end{proof}

\subsubsection{Intertwining operators}

If $M$ and $N$ are vector spaces, we write $\cL(M,N)$ for the space of linear maps from $M$ to $N$, and $\cL(M,N)\{x\}$ for the space of all formal series 
$$
\sum_{n \in \C} a_{(n)} x^{-n-1}
$$
with $a_{(n)} \in \cL(M,N)$.
\begin{Definition}
Let $V$ be a vertex operator superalgebra, and let $M, N$ and $K$ be $V$-modules.
An \emph{intertwining operator} of type $\binom{K}{M \, N}$ is a linear map $\cY: M \to \cL(N, K)\{x\}$, written
$$
\cY(a,x) = \sum_{n \in \C} a^{\cY}_{(n)} x^{-n-1},
$$
which satisfies the following properties. 
\begin{enumerate}
\item If $a \in M$ is even (resp. odd) then $a^{\cY}_{(n)}$ is even (resp. odd) for all $n \in \C$.
\item For every $a \in M$, $b \in N$ and $k \in \C$, we have $a^{\cY}_{(k+n)}b = 0$ for all sufficiently large $n \in \Z$.
\item For every $a \in M$, $\cY$ satisfies the $L_{-1}$-derivative property:
\begin{equation}
\cY(L_{-1} a, x) = \frac{d}{dx} \cY(a,x)
\end{equation}

\item For every homogeneous $a \in V$, $b \in M$, $c \in N$, every $m,n \in \Z$, and every $k \in \C$, the Borcherds(/Jacobi) identity holds:
\begin{align}
\sum_{j = 0}^\infty \binom{m}{j} \big(a^L_{(n+j)}b\big)^\cY_{(m+k-j)}&c = \sum_{j=0}^\infty (-1)^j \binom{n}{j} a_{(m+n-j)}^N b_{(k+j)}^\cY c \nonumber\\ 
&-(-1)^{p(a)p(b)} \sum_{j=0}^\infty(-1)^{j+n} \binom{n}{j} b_{(n+k-j)}^\cY a_{(m+j)}^M c. \label{eqnBorcherdsIdentityIntertwiner}
\end{align}
\end{enumerate}
We denote by $I \binom{K}{M \, N}$ the vector space of all intertwining operators of the indicated type.
\end{Definition}
The $L_{-1}$-derivative property cannot be deduced from the Borcherds identity for intertwining operators like it can be for module operators $Y^M$, and so we must include it in the definition.
Indeed, the Borcherds identity for intertwining operators is invariant under shifting every $b^\cY_{(k)}$ to $b^\cY_{(k+\alpha)}$, for any $\alpha \in \C$.
We do, however, have analogs of the other basic properties established for module operators (which are deduced immediately from the Borcherds identity).

\begin{Proposition}\label{propIntertwinerProperties}
Let $M$, $N$ and $K$ be modules over a vertex operator superalgebra $V$, and let $\cY \in I\binom{K}{M \, N}$.
\begin{enumerate}
\item $\cY$ satisfies a Borcherds product formula:
\begin{equation}
\big(a_{(n)}^M b\big)^\cY_{(k)} c= \sum_{j=0}^\infty (-1)^j \binom{n}{j} \left(a_{(n-j)}^K b_{(k+j)}^\cY - (-1)^{p(a)p(b)+n}b_{(n+k-j)}^\cY a_{(j)}^N\right)c.
\end{equation}
for every $a \in V$, $b \in M$, $c \in N$, $n \in \Z$ and $k \in \C$.
\item $\cY$ satisfies a Borcherds commutator formula:
\begin{equation}
a_{(m)}^K b_{(k)}^\cY c - (-1)^{p(a)p(b)} b_{(k)}^\cY a_{(m)}^N c = \sum_{j =0}^\infty \binom{m}{j} \big(a_{(j)}^M b\big)^\cY_{(m+k-j)}c
\end{equation}
for every $a \in V$, $b \in M$, $c \in N$, $m \in \Z$ and $k \in \C$.
\item If $b \in L$ is homogeneous of conformal weight $\Delta_b$, then $b_{(k)}^\cY N_\alpha \subset K_{\alpha - k -1 + \Delta_b}$. \label{itmIntertwinerHomogeneous} 
\end{enumerate}
\end{Proposition}

If $V$ is a unitary vertex operator superalgebra and $M$, $N$, and $K$ are unitary $V$-modules, then by Proposition \ref{propIntertwinerProperties}\eqref{itmIntertwinerHomogeneous}, every intertwining operator $\cY \in I\binom{K}{M \, N}$ can be written as a sum $\cY(a,x) = \sum_{k \in \R} a_{(k)} x^{-k-1}$ indexed by the real numbers.

\begin{Proposition}\label{propIntertwinerDescent}
Let $V$ be a simple unitary vertex operator superalgebra, and let $W$ be a unitary subalgebra.
Let $N$ and $K$ be simple $W$-submodules of a unitary $V$-module $\tilde M$, and let $M$ be a simple $W$-submodule of $V$.
Let $p_K: \tilde M \to K$ be the orthogonal projection (which maps into $K$ by Proposition \ref{propModuleProjection}).
Then there is a $\Delta \in \R$ such that the map $\cY:M \to \cL(N,K)\{x\}$ defined by
\begin{equation}\label{eqnDescentIntertwiner}
\cY(a,x) = p_K x^{\Delta} Y(a,x) |_N
\end{equation}
is an intertwining operator of type $\binom{K}{M \, N}$.
\end{Proposition}
\begin{proof}
Assume for now that $\Delta \in \R$ is arbitrary.
Since $p_K$ is even by Proposition \ref{propModuleProjection}, the parity requirement for the modes $a_{(k)}^\cY$ is satisfied.
The truncation condition $a^{\cY}_{(k+n)}b = 0$ for large $n$ is inherited from the corresponding property of $a_{(n)}$.
The Borcherds identity for $\cY$ is an immediate consequence of the Borcherds identity for $Y$, and the fact that $p_K$ commutes with the $W$-actions on $M$ and $N$ by Proposition \ref{propModuleProjection}.

It remains to check the $L_{-1}$ derivative property, which requires the correct choice of $\Delta$.
We will write $L_n^V$ and $L_n^W$ for the representations of the Virasoro algebra on $V$-modules coming from the conformal vectors $\nu^V$ and $\nu^W$, respectively.
By the Borcherds commutator formula and Proposition \ref{propUnitarySubalgebraIsUnitaryVOA}, $(L_0^V - L_0^W)|_M$ commutes with the module action $a_{(n)}^M = a_{(n)}|_M$ for $a \in W$.
Since $M$ is simple, there is a scalar $\Delta_M \in \R$ such that $L_0^V b - L_0^W b = \Delta_M b$ for all $b \in M$.
Repeating the argument for $N$ and $K$, we obtain real scalars $\Delta_N$ and $\Delta_K$ satisfying the analogous identities for $b \in N$ and $b \in K$.

Now set $\Delta = -\Delta_K + \Delta_M + \Delta_N$, and let $a \in M$ be homogeneous with $L_0^Wa = \Delta_a^W a$.
Then for $b \in N$ we obtain
$$
[L_0^W, p_K a_{(n)}]b = (\Delta^W_a-n+\Delta-1)p_K a_{(n)}b,
$$
by substituting $L_0^W|_M = L_0^V|_M - \Delta_M$, and similarly for the other modules, along  with the fact that $p_K$ commutes with $L_0^W$.
On the other hand, by the Borcherds commutator formula we have
$$
[L_0^W, p_K a_{(n)}] = p_K[L_0^W, a_{(n)}] = p_K ((L_{-1}^Wa)_{(n+1)} + \Delta_a^W a_{(n)}).
$$
Hence 
$$
p_K (L_{-1}^W a)_{(n+1)}b = -(n-\Delta+1)p_K a_{(n)}b
$$
for all $b \in N$.
Combining these formulas we obtain for homogeneous $a \in M$
$$
(L^W_{-1} a)_{(n-\Delta +1)}^\cY = p_K (L_{-1}^W a)_{(n + 1)}|_N = -(n-\Delta+1) a^\cY_{(n-\Delta)}.
$$
This extends by linearity to all $a \in M$, which establishes the $L_{-1}$-derivative property.
%
\end{proof}

\subsection{Fermi conformal nets}

In this section we will briefly outline the basic ideas of Fermi conformal nets, the $\Z/2\Z$-graded analog of local conformal nets.
The interested reader can find more detail in the original reference \cite{CaKaLo08}.

We first recall some basic terminology.
A super Hilbert space $\cH$ is a Hilbert space equipped with a $\Z/2\Z$ grading $\cH = \cH^0 \oplus \cH^1$.
The corresponding grading involution is $\Gamma = 1_{\cH^0} \oplus -1_{\cH^1}$.
Elements of $\cH^0$ (resp $\cH^1$) are called \emph{even} (resp. \emph{odd}) \emph{homogeneous vectors}, and if $\xi \in \cH^i$ we denote the parity of $\xi$ by $p(\xi) = i$.
The $\Z/2\Z$ grading on $\cH$ induces one on $\cB(\cH)$, corresponding to the involution $x \mapsto \Gamma x \Gamma$.
The supercommutator $[\, \cdot \, , \, \cdot \,]_{\pm}$ on $\cB(\cH)$ is given by $[x,y]_{\pm} = xy - (-1)^{p(x)p(y)}yx$ for homogeneous $x$ and $y$, and by extending linearly otherwise.

An \emph{interval} $I \subset S^1$ is an open, connected, non-empty, non-dense subset.
We denote by $\cI$ the set of all intervals.
If $I \in \cI$, we denote by $I^\prime$ the complementary interval $\interior{I^c}$.

We now fix notation for diffeomorphisms:
\begin{itemize}
\item For $n \in \bbZ_{\ge 1} \cup \{\infty\}$ we denote by $\Diff^{(n)}(S^1)$ the $n$-fold cover of the group of orientation preserving diffeomorphisms of the unit circle $S^1$.
The central extension $\Diff^{(\infty)}(S^1)$ of $\Diff(S^1)$ by $\bbZ$ is generated by the full rotation $r_{2\pi}$, and for finite $n$ representations of $\Diff^{(n)}(S^1)$ correspond to representations of $\Diff^{(\infty)}(S^1)$ for which $r_{2n\pi} = 1$.

\item We use the (standard, but potentially confusing) notation $\Diff(I)$ for the subgroup of $\Diff(S^1)$ consisting of diffeomorphisms which act as the identity on $I^\prime$, and embed $\Diff(I) \hookrightarrow \Diff^{(n)}(S^1)$ in the natural way.

\item Let $\Diff_c^{(n)}(S^1)$ be the central extension of $\Diff^{(n)}(S^1)$ by $U(1)$ corresponding to the central charge $c$, and let $\Diff_c(I)$ be the corresponding central extension of $\Diff(I)$. 
By a representation of $\Diff_c^{(n)}(S^1)$ we will mean a representation of the group in which the central $U(1)$ acts standardly, and similarly for $\Diff_c(I)$. 

\item Let $\Mob^{(n)}(\bbD) \subset \Diff^{(n)}(S^1)$ be the M\"{o}bius subgroup consisting of diffeomorphisms which extend to biholomorphic maps on the unit disk $\bbD$.
In fact, we have compatible inclusions $\Mob^{(n)}(\bbD) \subset \Diff_c^{(n)}(S^1)$ as well.
\end{itemize}

See e.g. \cite[\S2.2.1]{HenriquesColimits} and references therein for a discussion of the central extensions in the final two bullet points.

\begin{Definition}
A Fermi conformal net of central charge $c$ is given by the data:
\begin{enumerate}
\item A super Hilbert space $\cH = \cH^1 \oplus \cH^0$, with corresponding unitary grading involution $\Gamma$.
\item A strongly continuous projective unitary representation $U:\Diff_c^{(2)}(S^1) \to \cU(\cH)$ which restricts to an honest unitary representation of $\Mob^{(2)}(\bbD)$.
\item For every $I \in \cI$, a von Neumann algebra $\cA(I) \subset \cB(\cH)$.
\end{enumerate}
The data is required to satisfy:
\begin{enumerate}
\item The local algebras $\cA(I)$ are $\Z/2\Z$ graded. That is, $\Gamma \cA(I)\Gamma = \cA(I)$.
\item If $I,J \in \cI$ and $I \subset J$, then $\cA(I) \subset \cA(J)$.
\item If $I,J \in \cI$ and $I \cap J = \emptyset$, then $[\cA(I), \cA(J)]_{\pm} = \{0\}$.
\item $U(\gamma)\cA(I)U(\gamma)^* = \cA(\gamma(I))$ for all $\gamma \in \Diff_c^{(2)}(S^1)$, and $U(\gamma)xU(\gamma)^* = x$ when $x \in \cA(I)$ and $\gamma \in \Diff_c(I^\prime)$.
\item There is a unique (up to scalar) unit vector $\Omega \in \cH$, called the \emph{vacuum vector}, which satisfies $U(\gamma)\Omega = \Omega$ for all $\gamma \in \Mob^{(2)}(\bbD)$.
This vacuum vector is required to be cyclic for the von Neumann algebra $\cA(S^1):=\bigvee_{I \in \cI} \cA(I)$, and it must satisfy $\Gamma \Omega = \Omega$.
\item The generator $L_0$ of the one-parameter group $U(r_{\theta})$ is positive.
\end{enumerate}
\end{Definition}
A Fermi conformal net with $\cH = \cH^0$ is called a local conformal net (or just a conformal net).
If we set $\cA_b(I) = \{x \in \cA(I) : p(x) = 0\}$, then $\cA_b$ is a local conformal net on $\cH^0$.

Fermi conformal nets have many properties analogous to familiar properties of conformal nets.
We list some basic properties here:
\begin{Theorem}[\cite{CaKaLo08}]\label{thmFermiNetProps}
Let $\cA$ be a Fermi conformal net. Then we have:
\begin{enumerate}
\item (Haag duality) $\cA(I^\prime) = \{x \in \cB(\cH) : [x,y]_{\pm} = 0 \mbox{ for all } y \in \cA(I) \}$
\item (Reeh-Schlieder) $\cH = \overline{\cA(I)\Omega}$ for every $I \in \cI$.
\item $U(r_{2\pi}) = \Gamma$ and $\Gamma U(\gamma) = U(\gamma)\Gamma$ for all $\gamma \in \Diff^{(2)}_c(S^1)$.
\item $\cA(I)$ is a type III factor for every interval $I \in \cI$.
\end{enumerate}
\end{Theorem}

A family of von Neumann subalgebras $\cB(I) \subset \cA(I)$ is called a \emph{covariant subnet} if $\cB(I) \subset \cB(J)$ when $I \subset J$ and $U(\gamma)\cB(I)U(\gamma)^* = \cB(\gamma(I))$ for all $\gamma \in \Mob^{(2)}(\bbD)$.

The following theorem is proven by combining \cite[Thm. 6.2.29]{Weiner05} (for existence) and \cite[Thm 6.10]{CKLW18} (for uniqueness) in the case of local conformal nets.
It is easily adapted to the case of Fermi conformal nets by observing that every $\cB(I)$ is $\bbZ/2\bbZ$-graded by $\Gamma = U(r_{2\pi})$.
\begin{Theorem}\label{thmSubnetsAreNets}
Let $\cB$ be a covariant subnet of a Fermi conformal net $\cA$.
Then there is a unique strongly continuous projective unitary representation of $\Diff_c^{(2)}(S^1)$ making $\cB$ into a Fermi conformal net on $\cH_B :=\overline{\cB(S^1)\Omega}$.
\end{Theorem}

If $\cB$ is a covariant subnet of a Fermi conformal net, then the usual argument (given in e.g. \cite[Lem. 2]{KawahigashiLongo04}, using the Bisognano-Wichmann property \cite[Thm. 2]{CaKaLo08}) shows that for $x \in \cA(I)$, we have $x \in \cB(I)$ if and only if $x\Omega \in \cH_B$.
In particular, we have:
\begin{Proposition}\label{propSubnetTrivial}
Let $\cA$ be a Fermi conformal net on $\cH$, and let $\cB \subset \cA$ be a covariant subnet.
Then $\cB = \cA$ if and only if $\cH_B = \cH$.
\end{Proposition}

We will make frequent use of the graded tensor product $\cA_1 \grotimes \cA_2$ of a pair of Fermi conformal nets $(\cA_1, U_1)$ and $(\cA_2, U_2)$ (see \cite[\S2.6]{CaKaLo08}).
If $\cH_1$ and $\cH_2$ are super Hilbert spaces, then $\cH_1 \otimes \cH_2$ is naturally a super Hilbert space with grading $\Gamma \otimes \Gamma$.
For $x_i \in \cB(H_i)$, define $x_1 \grotimes x_2 = x_1 \Gamma^{p(x_2)} \otimes x_2 \in \cB(\cH_1 \otimes \cH_2)$ for homogeneous $x_2$, and by extending linearly otherwise.
Define $(\cA_1 \grotimes \cA_2)(I) = \{x_1 \grotimes x_2 : x_i \in \cA_i(I)\}^{\prime\prime}$, where the double commutant $S^{\prime \prime}$ is the von Neumann algebra generated by a self-adjoint set $S$.
This construction produces a Fermi conformal net $(\cA_1 \grotimes \cA_2, U_1 \otimes U_2)$ \cite[\S 2.6]{CaKaLo08}.


\begin{Definition}
A \emph{representation} of a Fermi conformal net $\cA$ is a super Hilbert space $\cH_\pi$ and a family of representations (i.e. normal, even, $*$-homomorphisms) $\pi_I:\cA(I) \to \cB(\cH_\pi)$, indexed by $I \in \cI$, which satisfy $\pi_{I}|_J = \pi_J$ when $J \subset I$.
\end{Definition}
We point out that when $\cH_{\pi}$ is separable, as it always will be in this article, the normality of $\pi_I$ is automatic.

There are obvious notions of subrepresentations, irreducible representations, and direct sums of representations.
An isomorphism of representations is an even unitary which intertwines the actions of all local algebras.
The vacuum representation of $\cA$ is given by the Hilbert space $\cH_0 := \cH$ and the defining actions of the algebras $\cA(I)$.
Given representations $\pi_i$ of $\cA_i$, there is a representation $\pi_1 \otimes \pi_2$ of $\cA_1 \otimes \cA_2$ given by $(\pi_1 \otimes \pi_2)_I(x_1 \hotimes x_2) = \pi_1(x_1) \hotimes \pi_2(x_2)$ (see \cite[\S2.6]{CaKaLo08}).

\begin{Definition}
If $\pi$ and $\lambda$ are representations of $\cA$, and $I \in \cI$, then local intertwiners $\Hom_{\cA(I)}(\cH_\pi, \cH_\lambda)$ is given by:
\begin{equation*}\label{eqn: local intertwiner def}
\Hom_{\cA(I)}(\cH_\pi,\cH_\lambda) := \{ x \in \cB(\cH_\pi,\cH_\lambda) \, : \, x\pi_I(y) = (-1)^{p(x)p(y)} \lambda_I(y)x  \mbox{ for all } y \in \cA(I)\},
\end{equation*}
and we define $\End_{\cA(I)}(\cH_\pi)$ similarly.
As usual, equations involving $p(x)$ and $p(y)$ should be interpreted by extending linearly for non-homogeneous elements.
\end{Definition}
In this notation, Haag duality says $\End_{\cA(I)}(\cH_0) = \cA(I^\prime)$.

By Haag duality, $U(\gamma) \in \cA(I)$ when $\gamma \in \Diff_c(I)$, and so given a representation of $\cA$ we obtain strongly continuous representations $\pi_I \circ U$ of every $\Diff_c(I)$.
By \cite[Thm. 11]{HenriquesColimits}, there is a unique strongly continuous representation $U^\pi$ of $\Diff_c^{(\infty)}(S^1)$ such that $U^\pi|_{\Diff_c(I)} = \pi_I \circ U$ (this was originally proven for irreducible representations in \cite{DFK04}).
Since $\Diff_c^{(\infty)}(S^1)$ is generated by the $\Diff_c(I)$ \cite[Lem. 17(ii)]{HenriquesColimits}, we may argue as in \cite[Lem. 3.1]{KawahigashiLongo04} or \cite[Prop. 12]{CaKaLo08}) to obtain
$$
 U^\pi(\gamma) \pi_I(x) U^\pi(\gamma)^* = \pi_{\gamma(I)}(U(\gamma)xU(\gamma)^*)
$$
for all $\gamma \in \Diff_c^{(\infty)}(S^1)$.
By \cite[Thm. 5.4]{MorinelliTanimotoWeiner18}, the net $\cA$ has the split property, and therefore  by \cite[Prop. 56]{KaLoMu01} if $\cH_\pi$ is separable the sector admits a direct integral decomposition 
$$
\pi = \int^\oplus \pi_x dx
$$
where $\pi_x$ is a representation of $\cA$ for almost every $x$.
It follows that 
$$
U^\pi(\gamma) = \int^{\oplus} U_x(\gamma) dx
$$
where for almost every $x$, $U_x$ is the representation obtained by piecing together the representation $\pi_{x,I} \circ U$ of $\Diff_c(I)$.
By \cite[Thm. 3.8]{Weiner06}, the generator of $U_x(r_\theta)$ is positive (after making $U_x|_{\Mob^{(\infty)}}$ into an honest representation in the canonical way).
Thus combining the results from the literature we have:

\begin{Theorem}\label{thm: reps are diff covariant}
Let $\pi$ be a representation of a Fermi conformal net $(\cA, U)$ on a separable Hilbert space.
Then there is a unique, strongly continuous, positive energy representation $U_\pi$ of $\Diff_c^{(\infty)}(S^1)$ which satisfies
$U_\pi|_{\Diff_c(I)} = \pi_I \circ U$ for all intervals $I$ and
$$
 U_\pi(\gamma) \pi_I(x) U_\pi(\gamma)^* = \pi_{\gamma(I)}(U(\gamma)xU(\gamma)^*)
$$
for all $\gamma \in \Diff_c^{(\infty)}(S^1)$.
\end{Theorem}

\newpage

\section{Systems of generalized annuli} \label{sec: generalized annuli}

The idea of \cite{GRACFT1} is to generate conformal nets from vertex operator algebras via local insertions into degenerate annuli; an overview of the approach is given in Section \ref{sec: BLVO intro}.
The precise geometry required of the degenerate annulus was influenced by technical considerations which are unimportant to the overall idea, but necessary at this stage to give rigorous proofs.
In this article we will abstract away the geometric considerations in order to simplify the statements of theorems and provide flexibility for future work.
The technical Definitions \ref{def: system of incoming annuli} to \ref{def: system of generalized annuli} below provide this abstraction.
Roughly speaking, these definitions encode the idea of a certain subset of the `semigroup of degenerate annuli,' a proposed generalization of the Segal-Neretin semigroup of annuli, which we now expand on.

The motivation starts with the semigroup of annuli \cite{SegalDef,Neretin90}, which consists of isomorphism classes of Riemann surfaces homeomorphic to an annulus $S^1 \times [0,1]$ whose boundary components have been equipped with parametrizations by the unit circle $S^1$.
The semigroup operation is that of sewing boundary components along the parametrization.

Neretin showed that every irreducible positive energy representation $L(c,h)$ of $\Vir$ yields a projective representation of the semigroup of annuli.
Every annular Riemann surface is isomorphic to some $A_r := \{r \le \abs{z} \le 1\}$. 
When $A_r$ has the standard boundary parametrizations, $A_r$ acts by $r^{L_0}$ in Neretin's representation.
Changing boundary parametrizations by a diffeomorphism $\gamma \in \Diff(S^1)$ corresponds to composition with $U_{c,h}(\gamma)$.

The operators on $\cH_{c,0}$ arising from Neretin's representation are not local in the sense of conformal nets (that is, they do not lie in local algebras $\cA_c(I)$).
To obtain such operators, we pass to `degenerate' (or `partially thin') annuli, where the incoming boundary is allowed to overlap with the outgoing boundary.
$$
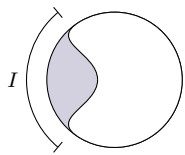
\begin{tikzpicture}[scale=.9,baseline={([yshift=-.5ex]current bounding box.center)}]
	\coordinate (a) at (120:1cm);
	\coordinate (b) at (240:1cm);
	\coordinate (c) at (180:.25cm);
	\fill[fill=red!10!blue!20!gray!30!white] (0,0) circle (1cm);
	\draw (0,0) circle (1cm);
	\fill[fill=white] (a)  .. controls ++(210:.6cm) and ++(90:.4cm) .. (c) .. controls ++(270:.4cm) and ++(150:.6cm) .. (b) -- ([shift=(240:1cm)]0,0) arc (240:480:1cm);
	\draw ([shift=(240:1cm)]0,0) arc (240:480:1cm);
	\draw (a) .. controls ++(210:.6cm) and ++(90:.4cm) .. (c);
	\draw (b) .. controls ++(150:.6cm) and ++(270:.4cm) .. (c);
	\draw (130:1.2cm) -- (130:1.4cm);
	\draw (230:1.2cm) -- (230:1.4cm);
	\draw (130:1.3cm) arc (130:230:1.3cm);
	\node at (180:1.5cm) {\scriptsize{$I$}};
\end{tikzpicture}
$$
The most extreme example of such an annulus is a totally thin one, where the incoming and outgoing boundary coincide; such an annulus corresponds to a diffeomorphism.

Neretin's representations should extend to this larger semigroup, which should produce families of operators acting on each $\cH_{c,0}$ corresponding to the image of this representation.
We can restrict this representation to degenerate annuli whose support (the `thick part') lies in an interval $I$ and such that the boundary parametrizations agree on $I^\prime$, and we should obtain a family of subsemigroups of the local algebras $\cA_c(I)$.

While we will not give a formal definition of the semigroup of degenerate annuli or show that such a semigroup would act on $\cH_{c,0}$, we will attempt to extract certain properties that this representation would have.
These properties comprise the definition of a `system of generalized annuli' in Section \ref{sec: def of system of generalized annuli}.
In Section \ref{sec: example of system of generalized annuli}, we will construct an example of a system of generalized annuli which arises from the geometric considerations of the semigroup of degenerate annuli explored above.
In future work, we hope to rigorously construct representations of the semigroup of degenerate annuli, and these would also furnish examples of systems of generalized annuli.

\subsection{Virasoro nets}\label{sec: virasoro nets}

Before arriving at the definition of a system of generalized annuli, which will be certain operators acting on the Hilbert space completions $\cH_{c,0}$ of vacuum Virasoro representations, we need to introduce the Virasoro conformal nets.
We quickly summarize basic facts about the Virasoro conformal nets.
A more detailed treatment may be found in \cite[\S2.4]{Carpi04} and references therein.

For every $c \in \{ 1 - \frac{6}{(m+2)(m+3)} : m \ge 1 \} \cup [1,\infty)$, there is a unitary, irreducible representation of the Virasoro algebra $L(c,0)$, called the vacuum representation with central charge $c$.
For each allowable value of the central charge $c$, unitary irreducible highest weight representations of $\Vir_c$ are parametrized by a non-negative real number $h$, and we denote the corresponding representation $L(c,h)$.\footnote{%
$\Vir_c$ is a Lie algebra equipped with a distinguished central copy of $\bbR$, and by a representation of $\Vir_c$ we mean a representation of the Lie algebra in which the distinguished copy of $\bbR$ acts standardly.
This is the same as a central charge $c$ representation of $\Vir$.}
When $c \ge 1$, there is an irreducible representation for every $h \ge 0$, but for the discrete series there are only finitely many irreducible representations for each $c$.

By the work of Goodman and Wallach \cite{GoWa85},  the Hilbert space completion $\cH_{c,h}$ of $L(c,h)$ carries an irreducible representation $U_{c,h}$ of $\Diff_c^{(\infty)}(S^1)$ which integrates the representation of the Virasoro algebra.
As described in \cite{Carpi04}, the assignment $\cA_c(I) = \{ U_{c,0}(\gamma) : \gamma \in \Diff_c(I)\}$ yields a conformal net on $\cH_{c,0}$.
By \cite[Prop. 2.1]{Carpi04} (along with Theorem \ref{thm: reps are diff covariant} to relax the hypotheses), every irreducible representation of $\cA_c$ is isomorphic to one on some $\cH_{c,h}$ characterized by $\pi_{c,h,I}(U_{c,0}(\gamma)) = U_{c,h}(\gamma)$ for all $\gamma \in \Diff_c(I)$.
Moreover, by \cite[Thm. 5.6]{Weiner17}, for every $L(c,h)$ the corresponding representation $\pi_{c,h}$ exists.
Thus for a unitary representation $M$ of $\Vir_c$ which decomposes as a direct sum $M=\bigoplus L(c,h_i)$, we have a corresponding representation of the diffeomorphism group given by $\bigoplus U_{c,h_i}$, and of the Virasoro net given by $\bigoplus \pi_{c,h_i}$.
In particular, we will apply this when $M$ is a unitary module of a vertex operator superalgebra.
Conversely, if $\pi$ is a representation of $\cA_c$, then by Theorem \ref{thm: reps are diff covariant} there is a corresponding positive energy representation $U^\pi$ of $\Diff_c^{(\infty)}(S^1)$.

\subsection{Generalized annuli}

As stated in the introduction to Section \ref{sec: generalized annuli}, our goal is to give an abstract presentation of what a representation of the semigroup of degenerate annuli might look like.
Operators representing degenerate annuli whose `thick part' lies inside an interval $I$, and such that the boundary parametrizations agree on $I^\prime$, should lie in $\cA_c(I)$.
More generally, every genuinely degenerate annulus should be localizable in some $\cA_c(I)$ by reparametrizing the boundary.
We give the definition here of a `generalized annulus', which is an operator on $\cH_{c,0}$ which may be localized in this way.
Such operators act in non-vacuum representations of $\cA_c$, and we will outline that process in this section.

\begin{Definition}\label{def: localizable operators}
Let $U=U_{c,0}$ be the vacuum representation of $\Diffc$ on $\cH_{c,0}$, and let $\cA_c(I) \subset \cB(\cH_{c,0})$ be the local algebra of the Virasoro net.
An operator $A \in \cB(\cH_{c,0})$ is \emph{left localizable in $I$} if there exists $\gamma \in \Diffc$ such that $U(\gamma)A \in \cA_c(I)$.
Similarly, $B \in \cB(\cH_{c,0})$ is \emph{right localizable in $I$} if there exists $\gamma \in \Diffc$ such that $AU(\gamma) \in \cA_c(I)$.
We write $\Ann^{\ell}_I(\cH_{c,0})$ for the class of operators which are left localizable in $I$ and which have dense image.
We write $\Ann^{r}_I(\cH_{c,0})$ for the class of operators which are right localizable in $I$ and injective.
We abbreviate these to $\Ann^{\ell/r}_I$ when $c$ is understood from context.
\end{Definition}

Left localizable operators are an abstract characterization of the action of a degenerate annulus contained in the unit disk whose outgoing boundary is the unit circle, parametrized by the identity.
Similarly, right localizable operators model degenerate annuli contained in the complement of the unit disk, whose incoming boundary is the unit circle parametrized by the identity.
Observe that $A \in \Ann^\ell_I$ if and only if $A^* \in \Ann^r_I$; we think of the adjoint as reflecting the degenerate annulus through the unit circle.
We will refer to operators which are right or left localizable as generalized annuli.

$$
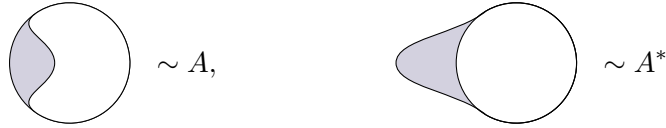
\begin{tikzpicture}[scale=.8,baseline={([yshift=-.5ex]current bounding box.center)}]
	\coordinate (a) at (120:1cm);
	\coordinate (b) at (240:1cm);
	\coordinate (c) at (180:.25cm);
	\fill[fill=red!10!blue!20!gray!30!white] (0,0) circle (1cm);
	\draw (0,0) circle (1cm);
	\fill[fill=white] (a)  .. controls ++(210:.6cm) and ++(90:.4cm) .. (c) .. controls ++(270:.4cm) and ++(150:.6cm) .. (b) -- ([shift=(240:1cm)]0,0) arc (240:480:1cm);
	\draw ([shift=(240:1cm)]0,0) arc (240:480:1cm);
	\draw (a) .. controls ++(210:.6cm) and ++(90:.4cm) .. (c);
	\draw (b) .. controls ++(150:.6cm) and ++(270:.4cm) .. (c);
\end{tikzpicture}
\,\, \sim A,
\qquad\qquad\qquad
\begin{tikzpicture}[scale=.8,baseline={([yshift=-.5ex]current bounding box.center)}]
	\coordinate (a) at (120:1cm);
	\coordinate (b) at (240:1cm);
	\coordinate (c) at (180:2cm);
	\draw (0,0) circle (1cm);
	\fill[red!10!blue!20!gray!30!white] (a)  .. controls ++(220:.6cm) and ++(90:.4cm) .. (c) .. controls ++(270:.4cm) and ++(140:.6cm) .. (b) arc (240:480:1cm);
	\filldraw[fill=white] (0,0) circle (1cm);
	\draw ([shift=(240:1cm)]0,0) arc (240:480:1cm);
	\draw (a) .. controls ++(220:.6cm) and ++(90:.4cm) .. (c);
	\draw (b) .. controls ++(140:.6cm) and ++(270:.4cm) .. (c);
\end{tikzpicture}
\,\, \sim A^*
$$

Now let $\pi$ be a representation of $\cA_c$ on a Hilbert space $\cH_\pi$, and let $U^\pi$ be the corresponding representation of $\Diff_c^{(\infty)}(S^1)$ (see Theorem \ref{thm: reps are diff covariant}).
Given $A \in \Ann^{\ell}_I$ and $\gamma \in \Diffc$ such that $U(\gamma)A \in \cA_c(I)$, we set
$$
\pi_{I,\tilde \gamma}(A) := U^\pi(\tilde \gamma)^* \pi_I(U(\gamma)A),
$$
where $\tilde \gamma$ is a lift of $\gamma$ to $\Diff_c^{(\infty)}(S^1)$.
Note that changing the lift $\tilde \gamma$ can change $\pi_{I,\tilde \gamma}(A)$ by a unitary which commutes with $\cA_c$ on $\cH_\pi$, but a priori $\pi_{I,\tilde \gamma}$ could also depend on the choice of $\gamma$.
We will show that this is not the case via two easy lemmas.

\begin{Lemma}\label{lem: local diff is local}
Let $\gamma \in \Diffc$, let $I$ be an interval, and suppose that $U(\gamma) \in \cA_c(I)$.
Then $\gamma \in \Diff_c(I)$.
\end{Lemma}
\begin{proof}
Suppose to a contradiction that $\gamma(z) \ne z$ for some $z \in I^\prime$.
Then we may find an interval $J \subset I^\prime$ such that $\gamma(J) \cap J = \emptyset$, and thus $\cA_c(J) \cap \cA_c(\gamma(J)) \subset Z(\cA_c(J)) = \bbC 1$.
However if we choose $x \in \cA_c(J)$ with $1 \ne x$, then $U(\gamma)xU(\gamma)^* = x$ since $U(\gamma) \in \cA_c(I)$, so $x \in \cA_c(J) \cap \cA_c(\gamma(J))$, a contradiction.
\end{proof}

\begin{Lemma}\label{lem: generalized ann act projectively}
Let $\pi$ be a representation of $\cA_c$, and let $U^\pi$ be the corresponding representation of $\Diffcinfty$.
Let $A \in \Ann^{\ell}_I$ and $\gamma_1,\gamma_2 \in \Diffc$ such that $U(\gamma_i)A \in \cA_c(I)$.
Choose lifts $\tilde \gamma_i$ of $\gamma_i$ to $\Diffcinfty$.
Then for some integer $m$ we have 
$$
\pi_{I,\tilde \gamma_1}(A) = e^{2 \pi i m L_0} \, \pi_{I,\tilde \gamma_2}(A).
$$
where $L_0$ is the generator of $U^\pi(r_{\theta})$.
In particular, if $\sigma(L_0) \subset h + \bbZ$, then $\pi_{I,\tilde \gamma_1}(A)$ and $\pi_{I,\tilde \gamma_2}(A)$ differ by a scalar of modulus 1.
\end{Lemma}
\begin{proof}
Let $v = U(\gamma_1)U(\gamma_2)^*$, let $x = U(\gamma_2)A$, and let $y \in \cA_c(I^\prime)$.
Observe that $x, vx \in \cA_c(I)$, so we have
$$
yvx = vxy = vyx,
$$
so that $[y,v]x = 0$.
Since $A \in \Ann^\ell_I$ it has dense image, and thus $x$ has dense image as well.
It follows that $[y,v] = 0$, and thus $U(\gamma_1 \circ \gamma_2^{-1}) = v \in \cA_c(I)$.
Hence by Lemma \ref{lem: local diff is local}, $\gamma_1 \circ \gamma_2^{-1} \in \Diff_c(I)$.
It follows that 
\begin{align*}
\pi_{I,\tilde \gamma_1}(A) &= U^{\pi}(\tilde \gamma_1)^* \pi_I(U(\gamma_1) A)\\
&= U^{\pi}(\tilde \gamma_1)^* \pi_I(U(\gamma_1 \circ \gamma_2^{-1}))\pi_I(U(\gamma_2)A)\\
&= U^{\pi}(\tilde \gamma_1)^* U^{\pi}(\gamma_1 \circ \gamma_2^{-1})\pi_I(U(\gamma_2)A)\\
\end{align*}
where we have used the canonical lift of $\gamma_1 \circ \gamma_2^{-1}$ from $\Diff_c(I)$ to $\Diffcinfty$.
The generator of the central extension $\Diffcinfty$ of $\Diff_c(S^1)$ acts by $e^{2 \pi i L_0}$, and so
$$
U^{\pi}(\tilde \gamma_1)^* U^{\pi}(\gamma_1 \circ \gamma_2^{-1}) = e^{2 \pi i mL_0} U^{\pi}(\tilde \gamma_2)^*
$$
for some integer $m$, completing the proof.
\end{proof}

\begin{Definition}\label{def: generalized annulus action}
Let $A \in \Ann^{\ell}_I$ and let $\pi$ be a positive energy representation of $\cA_c$. 
Then we define $\pi_I(A) := \pi_{I,\tilde \gamma}(A)$ for any appropriate $\gamma$, which is well-defined up to a multiple of $e^{2\pi i L_0}$.
Similarly, if $A \in \Ann^{r}_I$, then $\pi_I(A) := \pi_I(AU(\gamma)) U^\pi(\tilde \gamma)^*$.
\end{Definition}
Observe that if $A \in \cA_c(I)$, then this definition agrees with the value coming from the representation $\pi_I$ (up to the necessary ambiguity), so we think of Definition \ref{def: generalized annulus action} as extending $\pi_I$ from $\cA_c(I)$ to $\Ann^{\ell/r}_I$.
We will omit the interval $I$ and the representation $\pi$ when they are clear from context.

It is clear from the definition that if $A \in \Ann^{\ell}_I$, then $A^* \in \Ann^r_I$ and $\pi_I(A)^* = \pi_I(A^*)$.
This is potentially ambiguous when $A \in \Ann^{\ell}_I \cap \Ann^{r}_I$, as its action $\pi_I(A)$ obtained from thinking of it as an element of $\Ann^{\ell}_I$ may not coincide with $\pi_I(A)$ obtained from thinking of it as an element of $\Ann^{r}_I$. 
This is a consequence of the fact that if $A \in \cA(I) \cap \cA(J)$ with $I \cap J$ disconnected, then $\pi_I(A)$ may not coincide with $\pi_J(A)$.
However, it will always be clear from context which action we refer to, and it should be remarked that for $\gamma \in \Diffc$, we have $\pi_I(U(\gamma)) = U^\pi(\tilde \gamma)$ for all $I$ (with equality in the sense of Definition \ref{def: generalized annulus action}), regardless of whether $U(\gamma)$ is regarded as left or right localizable.

\begin{Lemma}\label{lem: diff covariance of generalized annuli}
Let $A \in \Ann^\ell_I$, $B \in \Ann^r_I$ and let $\gamma \in \Diffcinfty$.
Let $\pi$ be a representation of $\cA_c$.
Then we have:
\begin{enumerate}
\item $U(\gamma)A \in \Ann^\ell_I$ and $AU(\gamma)^* \in \Ann^\ell_{\gamma(I)}$, and moreover $\pi_I(U(\gamma)A) = U^\pi(\gamma)\pi_I(A)$ and $\pi_{\gamma(I)}(AU(\gamma)^*) = \pi_{I}(A)U^\pi(\gamma)^*$.
\item $BU(\gamma)^* \in \Ann^r_I$ and $U(\gamma)B \in \Ann^r_{\gamma(I)}$, and moreover $\pi_I(BU(\gamma)^*) = BU^\pi(\gamma)^*$ and $\pi_{\gamma(I)}(U(\gamma)B) = U^\pi(\gamma)\pi_{I}(B)$
\end{enumerate}
\end{Lemma}
\begin{proof}
The second part follows from the first by taking adjoints.
Recall that the action $\pi_I(A)$ is only well-defined up to ambiguities discussed in Definition \ref{def: generalized annulus action}.
Let $\psi \in \Diffcinfty$ be such that $U(\psi)A \in \cA_c(I)$.
Then $U(\gamma)U(\psi)AU(\gamma)^* \in \cA_c(\gamma(I))$ and so
\begin{align*}
\pi_{\gamma(I)}(AU(\gamma)^*) &= 
U^\pi(\psi)^*U^{\pi}(\gamma)^*\pi_{\gamma(I)}(U(\gamma)U(\psi)AU(\gamma)^*)\\
&=U^\pi(\psi)^*\pi_{I}(U(\psi)A)U^\pi(\gamma)^*\\
&= \pi_I(A)U^\pi(\gamma)^*.
\end{align*}
The other assertions are similar.
\end{proof}

As a final technical note, we will need the notion of tensor product of generalized annuli.
This generalizes the fact that we may take the tensor product of representations of the semigroup of degenerate annuli with central charges $c_1$ and $c_2$ to obtain one of central charge $c_1+c_2$.

\begin{Definition}\label{def: tensor product of generalized annuli}
If $A_i \in \Ann^{\ell/r}_I(\cH_{c_i,0})$, we say that the tensor product of $A_1$ and $A_2$ exists if there exists a $A \in \Ann^{\ell/r}_I(\cH_{c_1+c_2,0})$ such that $\pi^{c_1 \otimes c_2}_I(A) = A_1 \otimes A_2$, where $\pi^{c_1 \otimes c_2}$ refers to the action of $\Ann^{\ell}_I(\cH_{c_1+c_2,0})$ on $\cH_{c_1,0} \otimes \cH_{c_2,0}$.
\end{Definition}
Such an $A$ must be unique, and by abuse of notation we refer to $A$ as the generalized annulus $A_1 \otimes A_2$.
With this notation, if $\cH_1$ and $\cH_2$ are a pair of representations of $\cA_{c_i}$, then $\pi_I(A_1 \otimes A_2) = \pi_I(A_1) \otimes \pi_I(A_2)$.

\begin{Definition}\label{def: annulus pair}
Let $\Ann_I$ be the collection of pairs $(B,A)$ with $A \in \Ann^{\ell}_I$ and $B \in \Ann^{r}_I$ for which there exists a common $\gamma \in \Diffc$ such that $U(\gamma)A, BU(\gamma)^* \in \cA_c(I)$.
\end{Definition}

Typical examples of elements of $\Ann_I$ include ones of the form $(A^*,A)$ for $A \in \Ann^{\ell}_I$ as well as ones of the form $(U(\gamma), A)$. 
Indeed, Definition \ref{def: annulus pair} is introduced for the purpose of simultaneously describing these two cases.


\begin{Remark}\label{rmk: localized annuli commute}
Observe that if $(B,A) \in \Ann_I$ and $(\tilde B, \tilde A) \in \Ann_J$, with $I$ and $J$ disjoint intervals, then in any positive energy representation $\pi_I(B)\pi_I(A)$ commutes with $\pi_J(\tilde B) \pi_J(\tilde A)$.
Indeed, $BA \in \cA_c(I)$ and $\tilde B \tilde A \in \cA_c(J)$, and we have $\pi_I(B)\pi_I(A) = \pi_I(BA)$, and $\pi_I(\tilde B)\pi_I(\tilde A) = \pi_I(\tilde B \tilde A)$.
\end{Remark}

\subsection{Definition of systems of generalized annuli}

\label{sec: def of system of generalized annuli}

As described in the beginning of Section \ref{sec: generalized annuli}, a system of generalized annuli provides an abstract characterization of the kinds of operators which could correspond to the action of degenerate annuli on representations of the Virasoro algebra.
Alternatively, it can be interpreted as a list of properties satisfied by the example constructed in Section \ref{sec: example of system of generalized annuli}, as that is the only example we are interested in in this article.
However, we would like to leave the option to generalize our framework to include the full semigroup of degenerate annuli, which will also provide an example of a system of generalized annuli.

\begin{Definition}[System of incoming/outgoing generalized annuli with central charge $c$]
\label{def: system of incoming annuli}
A system of incoming generalized annuli of central charge $c$ is a family of subsets $\scA^{in}_I \subset \Ann^{\ell}_I(\cH_{c,0})$ for all intervals $I$, such that:
\begin{enumerate}
\item $1 \in \scA^{in}_I$ for all $I$
\item If $I \subset J$ then $\scA^{in}_I \subset \scA^{in}_J$ 
\item If $A \in \scA^{in}_I$ then $U(r_{\theta})AU(\gamma)^* \in \scA^{in}_{\gamma(I)}$ for all $\gamma \in \Diffc$ and all $\theta \in \bbR$.
\end{enumerate}
The associated system of generalized outgoing annuli is given by the adjoints 
$$
{\scA_I^{out}}:=\{A^* : A \in \scA_I^{in}\}.
$$

A system of generalized incoming annuli must come equipped with an interior function ``$\Int$'' which assigns to every $A \in \scA_I^{in}$ an open set $\Int(A)$ which is the interior of a degenerate annulus not containing 0.
That is, $\Int(A) = \interior{\bbD} \setminus D$ where $\bbD$ is the closed unit disk and $D \subset \bbD$ is the closure of a Jordan domain with $C^\infty$ boundary with $0 \in \interior{D}$.
We further require that if $U(\gamma)A \in \cA_c(I)$, then $\gamma^{-1}(I^\prime) \subseteq \partial D \cap S^1$.
We do not rule out the possibility that $\Int(A) = \emptyset$.
This function must be compatible with reparametrization:
$$
\Int(U(r_{\theta})A U(\gamma)^*) = r_{\theta}(\Int(A)).
$$
If $B \in \scA_I^{out}$ then we define $\Int(B) = \overline{\Int(A^*)}^{-1} \subset \bbD^c$.
\end{Definition}

The reader should keep in mind that elements in $\scA_I^{in}$ capture the idea of an operator assigned to an annulus whose outer boundary is the unit circle, parametrized by the identity, and whose thin interval is parametrized by an interval containing $I^\prime$.
Elements in $\scA_I^{out}$ are reflections of incoming annuli through the unit circle.
A system of generalized annuli consists of pairs $(B,A) \in \scA^{out} \times \scA^{in}$, which corresponds to a degenerate annulus containing the unit circle, divided in half via the unit circle:

$$
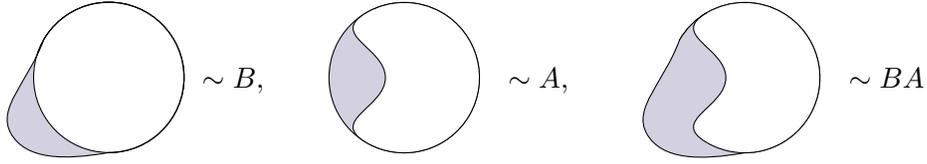
\begin{tikzpicture}[baseline={([yshift=-.5ex]current bounding box.center)}]
	\coordinate (a) at (150:1cm);
	\coordinate (b) at (270:1cm);
	\coordinate (c) at (210:1.5cm);
	\draw (0,0) circle (1cm);
	\fill[red!10!blue!20!gray!30!white] (a)  .. controls ++(250:.6cm) and ++(120:.4cm) .. (c) .. controls ++(300:.4cm) and ++(190:.6cm) .. (b) arc (270:510:1cm);
	\filldraw[fill=white] (0,0) circle (1cm);
	\draw ([shift=(270:1cm)]0,0) arc (270:510:1cm);
	\draw (a) .. controls ++(250:.6cm) and ++(120:.4cm) .. (c);
	\draw (b) .. controls ++(190:.6cm) and ++(300:.4cm) .. (c);
\end{tikzpicture}
\,\, \sim B,
\qquad
\begin{tikzpicture}[baseline={([yshift=-.5ex]current bounding box.center)}]
	\coordinate (a) at (120:1cm);
	\coordinate (b) at (240:1cm);
	\coordinate (c) at (180:.25cm);
	\fill[fill=red!10!blue!20!gray!30!white] (0,0) circle (1cm);
	\draw (0,0) circle (1cm);
	\fill[fill=white] (a)  .. controls ++(210:.6cm) and ++(90:.4cm) .. (c) .. controls ++(270:.4cm) and ++(150:.6cm) .. (b) -- ([shift=(240:1cm)]0,0) arc (240:480:1cm);
	\draw ([shift=(240:1cm)]0,0) arc (240:480:1cm);
	\draw (a) .. controls ++(210:.6cm) and ++(90:.4cm) .. (c);
	\draw (b) .. controls ++(150:.6cm) and ++(270:.4cm) .. (c);
\end{tikzpicture}
\,\, \sim A,
\qquad
\begin{tikzpicture}[baseline={([yshift=-.5ex]current bounding box.center)}]
	\coordinate (a1) at (150:1cm);
	\coordinate (b1) at (270:1cm);
	\coordinate (c1) at (210:1.5cm);
	\fill[red!10!blue!20!gray!30!white] (a1)  .. controls ++(250:.6cm) and ++(120:.4cm) .. (c1) .. controls ++(300:.4cm) and ++(190:.6cm) .. (b1) arc (270:510:1cm);
	\draw ([shift=(270:1cm)]0,0) arc (270:510:1cm);
	\coordinate (a) at (120:1cm);
	\coordinate (b) at (240:1cm);
	\coordinate (c) at (180:.25cm);
	\fill[fill=red!10!blue!20!gray!30!white] (0,0) circle (1cm);
	\fill[fill=white] (a)  .. controls ++(210:.6cm) and ++(90:.4cm) .. (c) .. controls ++(270:.4cm) and ++(150:.6cm) .. (b) -- ([shift=(240:1cm)]0,0) arc (240:480:1cm);
	\draw ([shift=(240:1cm)]0,0) arc (240:480:1cm);
	\draw (a) .. controls ++(210:.6cm) and ++(90:.4cm) .. (c);
	\draw (b) .. controls ++(150:.6cm) and ++(270:.4cm) .. (c);
	\draw (a) arc (120:150:1cm) (a1) .. controls ++(250:.6cm) and ++(120:.4cm) .. (c1);
	\draw (b1) .. controls ++(190:.6cm) and ++(300:.4cm) .. (c1);
\end{tikzpicture}
\,\, \sim
BA
$$

\begin{Definition}[System of generalized annuli of central charge $c$]
\label{def: system of annuli central charge c}
A system of generalized annuli of central charge $c$ is a system of incoming and outgoing generalized annuli $\scA^{in/out}_I$, and a family of subsets $\scA_I \subset \Ann_I \, \cap \,\, {\scA^{out}_I} \times \scA_I^{in}$ such that
\begin{enumerate}
\item if $I \subset J$ then $\scA_I \subset \scA_J$
\item if $(B,A) \in \scA_I$ then $(A^*,B^*) \in \scA_I$
\item if $\gamma \in \Diff_c(I)$, then $(1,U(\gamma)) \in \scA_I$
\item if $\gamma \in \Diffc$, $r_\theta \in \Rot(S^1)$, $(B,A) \in \scA_I$, then $(U(\gamma)BU(r_{\theta})^*, U(r_\theta)AU(\gamma)^*) \in \scA_{\gamma(I)}$.
\item if $r_\theta \in \Rot(S^1)$, $I$ is an interval containing both the intervals $J$ and $r_\theta(J)$, and $(B,A) \in \scA_J$, then $(B, U(r_{\theta})AU(r_\theta)^*) \in \scA_I$
\item for all $I$ there exists an $A \in \scA_I^{in}$ such that $(1,A) \in \scA_I$ and $\Int(A) \ne \emptyset$
\item if $(1,A), (B,1) \in \scA_I$ then $(B,A) \in \scA_I$.
\end{enumerate}
For $(B,A) \in \scA_I$, we define 
$$
\Int(B,A) = {(\cl(\Int(A)) \cup \cl(\Int(B))}\interior{\,\,}.
$$ 
\end{Definition}

Now a system of generalized annuli consists of systems of generalized annuli of every allowable central charge $c$ which are compatible with tensor product.

\begin{Definition}[System of generalized annuli]\label{def: system of generalized annuli}
A system of generalized annuli consists of systems of generalized annuli $\scA_{I,c} \subset \scA_{I,c}^{out} \times \scA_{I,c}^{in}$ for every allowable central charge $c$, such that for every $c_1,c_2$ and $A \in \scA^{in}_{I,c_1+c_2}$ there exist $A^\prime \in \scA^{in}_{I,c_1}$ and $A^{\prime\prime} \in \scA^{in}_{I,c_2}$ such that $A = A^\prime \otimes A^{\prime\prime}$ in the sense of Definition \ref{def: tensor product of generalized annuli}.
Moreover, we require that every $A^\prime \in \scA^{in}_{I,c_1}$ and $A^{\prime\prime} \in \scA^{in}_{I,c_2}$ can be obtained from such a decomposition.
For all $A$, we require $\Int(A) = \Int(A^\prime) = \Int(A^{\prime\prime})$.
Observe that elements of $\scA^{out}_{I,c}$ inherit the same properties by taking adjoints.

Finally we require that if $A = A^\prime \otimes A^{\prime\prime}$ and $B = B^\prime \otimes B^{\prime\prime} \in \scA_{I,c_1+c_2}^{out}$, then $(B,A) \in \scA_{I,c_1+c_2}$ if and only if $(B^\prime,A^\prime) \in \scA_{I,c_1}$ and $(B^{\prime\prime},A^{\prime\prime}) \in \scA_{I,c_2}$.
\end{Definition}

Observe that if $A = A^\prime \otimes A^{\prime\prime}$, then $A^\prime$ and $A^{\prime\prime}$ are determined up to a scalar (recall elements of $\Ann_I^{\ell/r}$ must be non-zero, since they have dense image or are injective).
We will often suppress this ambiguity for clarity, as described below.

\begin{Convention}\label{conv: annuli}
In general, we fix a system $\{\scA_{I,c}\}_I$ of generalized annuli which is in the background of all definitions and theorems.
The particular system for which we can construct examples is described in Section \ref{sec: example of system of generalized annuli}.
Given $A \in \scA_{I,c}^{in}$ and a representation $\pi$ of $\cA_c$, we have an action of $A$ given by $\pi_I(A) \in \cB(\cH_\pi)$.
In general we will omit the explicit reference to the central charge, as well as the notation $\pi_I$, and simply allow $A \in \scA_I$ to act on representations of $\cA_c$.
If $A$ is acting on a tensor product of representations $\cH_\pi \otimes \cH_\lambda$, then we will denote the operators $A^\prime$ and $A^{\prime\prime}$ of Definition \ref{def: system of generalized annuli} again by $A$.
This produces notation analogous to group representations: $A(\xi \otimes \eta) = A\xi \otimes A\eta$.
Recall that these are only well-defined up to a multiple of $e^{2 \pi i L_0}$, which will often be a scalar for us.


Sometimes we will need to emphasize the space that $A$ is acting on.
If $\cH$ is a positive energy representation of $\Diff_c^{(\infty)}$ with finite energy vectors $M$, we will sometimes write $A^M$ or $A^{\cH}$.
Alternatively, if $W$ is a non-conformal unitary subalgebra of $V$, then a unitary $V$-module $M$ naturally carries two different positive energy representations.
We will write $A^W$ and $A^V$ to distinguish between the two possible actions of $A$ on $\cH_M$.
\end{Convention}

%
%

\subsection{Example of a system of generalized annuli}\label{sec: example of system of generalized annuli}

In this section we will construct a system of generalized annuli (a version of which appeared implicitly in \cite{GRACFT1}).
These are directly inspired by degenerate annuli; see the introduction to Section \ref{sec: generalized annuli} or Section \ref{sec: BLVO intro} for more detail.

We briefly recall facts about semigroups of univalent functions from \cite[\S2.4]{GRACFT1} (a univalent function is an injective holomorphic map).
We denote the closed unit disk by $\bbD$, and its interior by $\interior{\bbD}$.
The following originates from \cite{Koenigs1884}, and a modern textbook treatment may be found in \cite[\S6.1]{Shapiro93}.

\begin{Proposition}[Koenigs]
\label{prop: Koenigs}
Let $\varphi:\interior{\bbD} \to \interior{\bbD}$ be a holomorphic map with $\varphi(0) = 0$, $\varphi \not \equiv 0$ and $f(\interior{\bbD}) \ne \interior{\bbD}$.
There there exists a unique holomorphic function $\sigma$ on $\interior{\bbD}$ such that $\sigma(0) = 0$, $\sigma^\prime(0) = 1$, and $\sigma(\varphi(z)) = \varphi^\prime(0) \sigma(z)$.
If $\varphi$ is univalent, then so is $\sigma$.
\end{Proposition}

The function $\sigma$ of Proposition \ref{prop: Koenigs} is called the \emph{Koenigs map} associated to $\varphi$.

\begin{Definition}\label{def: good univalent}
Let $\scU$ be the collection of univalent functions $\varphi:\interior{\bbD} \to \interior{\bbD}$ which satisfy
\begin{enumerate}
\item $\varphi(0) = 0$ 
\item $\varphi^\prime(0) \in \bbR_{> 0}$
\item $\varphi(\interior{\bbD})$ is a Jordan domain with $C^\infty$ boundary
\item $\sigma(\interior{\bbD})$ is a Jordan domain with $C^\infty$ boundary, where $\sigma$ is the Koenigs map associated to $\varphi$.
\end{enumerate}
\end{Definition}

Recall that by the smooth Riemann mapping theorem (see e.g. \cite[\S8]{Bell92}), if $U$ and $V$ are Jordan domains with $C^\infty$ boundary and $\varphi$ is a univalent map with $\varphi(U) = V$, then $V$ extends smoothly to the closures of $U$ and $V$, and induces a diffeomorphism of $\partial U$ and $\partial V$.
In light of this fact, the third item of Definition \ref{def: good univalent} is redundant in the presence of the fourth item, but we leave it for clarity.
We will implicitly extend maps $\varphi \in \scU$ to smooth maps $\varphi:\bbD \to \bbD$, and use the induced diffeomorphism $\varphi|_{S^1}:S^1 \to \partial \varphi(\bbD)$.

We now consider continuous semigroups of univalent self-maps $(\varphi_t)_{t \ge 0}$ of $\interior{\bbD}$.
That is, we consider families with $\varphi_0 = \operatorname{id}$, $\varphi_t \circ \varphi_s = \varphi_{t + s}$, and with the property that $(t,z) \mapsto \varphi_t(z)$ is continuous on $\bbR_{\ge 0} \times \interior{\bbD}$.
If $\varphi_t(\interior{\bbD}) \ne \interior{\bbD}$ and $\varphi_t(0) = 0$ for some (or, equivalently, all) $t > 0$, then all $(\varphi_t)_{t > 0}$ share a common Koenigs map $\sigma$, and $\varphi_t$ is determined by $\sigma$ up to reparametrizing $(\varphi_t) \mapsto (\varphi_{\alpha t})$ for some $\alpha > 0$.
See \cite[Prop. 2.35]{GRACFT1} and preceding discussion for more detail.

\begin{Definition}\label{def: good semigroup}
Let $\scG$ be the collection of non-trivial continuous semigroups $(\varphi_t)_{t \ge 0} \subset \scU$ which satisfy the normalization condition
$
\lim_{t \searrow 0} \frac{\varphi_t(z) - z}{t} = - \frac{\sigma(z)}{\sigma^\prime(z)}
$
for all $z \in \interior{\bbD}$, where $\sigma$ is the Koenigs map of the $\varphi_t$.
The (modified) vector field $\rho$ associated to $(\varphi_t)_{t \ge 0}$ is $\rho(z) = \frac{\sigma(z)}{z \sigma^\prime(z)}$.
\end{Definition}
The normalization condition serves to choose a canonical representative of the orbit of $(\varphi_t)$ under reparametrization, and will simplify formulas at times.
The vector field which generates $\varphi_t$ is $-\frac{\sigma(z)}{\sigma^\prime(z)}$, but the modified function $\rho$ will serve our purposes better.
We will omit the qualifier ``modified'' for brevity.
The generator $\rho$ satisfies (see \cite{BerksonPorta78})
\begin{equation}\label{eqn: univalent semigroup generator}
\dot{\varphi}_t(z) = -\varphi_t(z) \rho(\varphi_t(z)).
\end{equation}

We consider (potentially degenerate) annuli which are obtained by removing $\varphi_t(\interior{\bbD})$ from the unit disk.
We allow arbitrary smooth parametrization of the incoming boundary $\varphi_t(S^1)$.
Since $\varphi_t$ provides a canonical parametrization, the data of a parametrization of $\varphi_t(S^1)$ is the same as the data of an element of $\Diff(S^1)$.
On the other hand, we only allow the outgoing boundary $S^1$ to be parametrized by a M\"{o}bius transformation of the disk $\psi$.
Equivalently, we can consider the annulus $\bbD \setminus \psi(\varphi_t(\interior{\bbD}))$ with outgoing boundary parametrized by the identity.
We want the image of the operator assigned to this annulus to contain eigenvectors of $L_0$, and thus we must require $0 \in \psi(\varphi_t(\interior{\bbD}))$.
Putting this together, we obtain the following definition of a ``nice'' degenerate annulus:

\begin{Definition}[Geometric annuli]\label{def: geometric annuli}
Let $\GAnn^{in}$ be the collection of tuples $X:=(\psi,\rho,t,\gamma)$, where
\begin{enumerate}
\item $\rho$ is the generator of a semigroup $\varphi_t \in \scG$
\item $t \in \bbR_{\ge 0}$
\item $\psi \in \Mob(\bbD)$
\item $0 \in \psi(\varphi_t(\interior{\bbD}))$
\item $\gamma \in \Diff(S^1)$
\end{enumerate}
We define $\Int(X) := \psi(\interior{\bbD} \setminus \varphi_t(\bbD))$.
Let $\GAnn^{in}_I$ be the subcollection with $\varphi_t(\gamma^{-1}(I^\prime)) \subset S^1$.

Let $\GAnn_I$ be the collection of pairs $(\tilde X,X)$ where $X,\tilde X \in \GAnn^{in}_I$ and 
\begin{equation}\label{eqn: agree on thin}
(\psi \circ \varphi_t \circ \gamma^{-1})|_{I^\prime} = (\tilde \psi \circ \tilde \varphi_{\tilde t} \circ \tilde \gamma^{-1})|_{I^\prime}.
\end{equation}
where $\tilde X = (\tilde \psi,\tilde \rho ,\tilde t, \tilde \gamma)$.

Given a choice of central charge $c$ and a lift of $\gamma$ to $\Diffc$, we define an operator $\pi(X)$ on $\cH_{c,0}$ by
$$
\pi(X) = U(\psi)e^{-t L(\rho)} U(\gamma)^*.
$$
Here $U=U_{c,0}$ is the representation of $\Diff_c(S^1)$ on $\cH_{c,0}$ and $L(\rho) = \sum_{n \ge 0} \hat \rho(n) L_n$, where $\hat \rho$ are the Fourier coefficients of $\rho$.
\end{Definition}

The annulus $X \in \GAnn^{in}$ corresponds to $\bbD \setminus \psi(\varphi_t(\interior{\bbD}))$ as described above.
The condition \eqref{eqn: agree on thin} says the parametrizations of the incoming and outgoing boundary agree on $I$, which is to say that $X$ should be thought of as being localized in $I$.
If $\gamma \in \Diff_I(S^1)$, then $(\id,\rho,0,\gamma) \in \GAnn^{in}_I$ (for arbitrary $\rho$, which plays no role since $t=0$), and $\pi(\gamma) = U(\gamma)$.

To interpret the exponential $e^{-tL(\rho)}$, observe that the subspace of $L(c,0)$ spanned by $L_0$-eigenvectors with eigenvalue at most $N$ is invariant under $L(\rho)$, and therefore $e^{-tL(\rho)}$ defines an operator on this subspace.
Hence $e^{-tL(\rho)}$ defines an endomorphism of $L(c,0)$.
Using the quantum energy inequality of Fewster and Hollands \cite[\S4-5]{FewsterHollands}, we can see that $e^{-tL(\rho)}$ in fact defines a bounded operator on $\cH_{c,0}$, and thus we have:

\begin{Lemma}[{\cite[Prop. 3.13]{GRACFT1}}]
Let $X \in \GAnn^{in}$.
Then $\pi(X) \in \cB(\cH_{c,0})$
\end{Lemma}

\begin{Definition}[Choice of generalized annuli]
\label{def: choice of generalized annuli}
Let
$$
\scA_{I,c}^{in} = \{ \pi(X) : X \in \GAnn_I\},
$$
where we include all possible lifts of $\gamma$ to $\Diffc$.
Let
$$
\scA_{I,c} = \{(\pi(\tilde X)^*, \pi(X)) \,\, : \,\, (\tilde X, X) \in \GAnn_I \}
$$
We define the interiors
$$
\Int(\pi(X)) := \Int(X).
$$
\end{Definition}

Note that a priori the interior of $\pi(X)$ could depend on a particular choice of $X$, and not just on $\pi(X)$.
While it is possible for $\pi(X) = \pi(Y)$ in a non-trivial way (by adjusting $\psi$, $\rho$ and $\gamma$ by a common rotation), in fact the interior does not change.
Verifying the details of this is tedious and will be omitted as it does not affect the remainder of the paper.
We can simply pick a single $X$ to represent each generalized annulus $\pi(X)$\footnote{
More precisely, we pick a single $X$ for each orbit $U(r_{\theta})\pi(X)U(\gamma)^*$ under reparametrization, and use this to define the interior for the entire orbit.
When possible, we choose $X$ with non-empty interior.
}.

We now must verify that the given family $\scA_I$ is actually a system of generalized annuli.
First, we need need to verify that the operators $\pi(X)$ defined above are localizable.

\begin{Lemma}\label{lem: localizable}
Let $X:=(\psi,\rho,t,\gamma) \in \GAnn_I^{in}$.
Then there exists $\hat \gamma \in \Diff(S^1)$ such that $\hat\gamma \circ \psi \circ \varphi_t \circ \gamma^{-1}|_{I^\prime} = \operatorname{id}$.
For any such $\hat \gamma$, we have $U(\hat \gamma)\pi(X) \in \cA_c(I)$, for any lift of $\hat \gamma$ to $\Diffc$.
\end{Lemma}
\begin{proof}
We may assume without loss of generality that $\psi = \operatorname{id}$, by replacing $\hat \gamma$ by $\hat \gamma \circ \psi$.
We may assume without loss of generality that $\gamma = \operatorname{id}$ (by replacing $I$ with $\gamma(I)$ and $\hat \gamma$ with $\gamma \circ \hat \gamma$).
In this case, $\pi(X) = e^{-tL(\rho)}$.
It suffices to show that $U(\hat \gamma) e^{-tL(\rho)} \in \cA_c(I)$ for one such $\hat \gamma$, as any other differs by an element of $\Diff_c(I)$.
Moreover the result is preserved if $U(\hat \gamma)$ is replaced by a scalar multiple, so we need only verify the result for one lift of $\hat \gamma$ to $\Diffc$.

We first produce such a $\hat \gamma$.
Decompose $\rho(z) = f(z) + i g(z)$, and let $(\gamma_s)_{s \in \R} \subset \Diff(S^1)$ be the one parameter group obtained by integrating the real vector field $-izg(z)\frac{d}{dz} = g(e^{i \theta}) \frac{d}{d\theta}$ on $S^1$.
Since $(\id,\rho,t,\id) \in \GAnn_I^{in}$, we have $J:= \varphi_t(I^\prime) \subset S^1$.
For all $s$ we have $\varphi_s(\interior{\bbD}) \subset \interior{\bbD}$, so for $0 \le s \le t$ we must have $\varphi_s(I^\prime) \subset S^1$.
Since $\dot{\varphi_s}(z) = -\varphi_s(z) \rho(\varphi_s(z))$ by \eqref{eqn: univalent semigroup generator}, we must have that $\rho|_{\varphi_s(I^\prime)}$ is purely imaginary for all $0 \le s \le t$.
Setting $\tilde J = \bigcup_{0 \le s \le t} \varphi_s(I^\prime)$, we have $f|_{\tilde  J} = 0$, and so for $z \in I^\prime$ and $0 \le s \le t$, we have
$\dot{\varphi_s}(z) = i \varphi_s(z)g(\varphi_s(z))$ and $\varphi_0(z) = z$.
Thus $\varphi_s(z)$ satisfies the initial value problem characterizing $\gamma_{-s}(z)$, and we have $\gamma_s(\varphi_s(z)) = z$ for all $z \in I^\prime$ and $0 \le s \le t$.
In particular, we have $\gamma_t \circ \varphi_t|_{I^\prime} = \id|_{I^\prime}$.
Thus $\gamma_t$ will play the role of $\hat \gamma$ in the following, and we will show that $U(\gamma_t)\pi(X) \in \cA_c(I)$.

Since $U$ is obtained by integrating the positive energy representation of $\Vir_c$ on $\cH_{c,0}$, we have $U(\gamma_s) = e^{isL(g)}$ (see \cite[\S3.2]{Weiner05}).
We must verify that $e^{itL(g)}e^{-tL(\rho)} \in \cA_c(I)$.
We will do so via the Trotter product formula applied to $-tL(f)-itL(g)$ (see e.g.\cite[Cor. 3.5.5]{Pazy}).
Observe that $-L(f)$, $-iL(g)$, and $-L(\rho)$ share a common core of smooth vectors for $L_0$, and on this core $-L(\rho) = -L(f)-i L(g)$.
Therefore to verify the hypotheses of the Trotter product formula, it suffices to show that for each operator we may add a real constant and have the resulting operator generate a contraction semigroup.
This is clear for $iL(g)$, as it generates a unitary group.
For the self-adjoint operator $-L(f)$ it is also clear, as its spectrum is bounded above by \cite{FewsterHollands} and functional calculus provides the desired semigroup.
Finally, in \cite[Prop. 3.13]{GRACFT1} we verified directly that $-L(\rho)$ may be shifted to obtain a contraction semigroup generator (again using \cite{FewsterHollands} to invoke the Lumer-Phillips theorem \cite[Cor. 1.4.4]{Pazy}).
Thus we have 
\begin{equation}\label{eqn: Trotter expansion}
e^{-t L(\rho)} = \lim_{N \to \infty}  (e^{-itL(g)/N}e^{-tL(f)/N} )^N
= \lim_{N \to \infty} (U(\gamma_{t/N})^*e^{-tL(f)/N} )^N,
\end{equation}
in the strong operator topology.
Fix $N$, and for $0 \le j \le N$ set $f_j = \gamma_{-tj/N}^\prime \cdot (f \circ \gamma_{-tj/N})$.
By \cite[Lem. 6.2.22]{Weiner05}, we have for $0 \le j < N$
\begin{equation}\label{eqn: stress energy projective covariance}
U(\gamma_{t/N}) e^{-tL(f_j)/N} U(\gamma_{t/N})^* = \alpha_j e^{-tL(f_{j+1})}
\end{equation}
for scalars $\alpha_j$.
Hence applying \eqref{eqn: stress energy projective covariance} to \eqref{eqn: Trotter expansion} $N$ times, we obtain scalars $\alpha_N$ such that
\begin{align}
\label{eqn: local generalized annuli trotter calculation}
U(\gamma_t) e^{-t L(\rho)} 
&= U(\gamma_t) \lim_{N \to \infty} \big(U(\gamma_{t/N})^*e^{-tL(f)/N} \big)^N \\
&=  \lim_{N \to \infty} U(\gamma_{t/N})^N\big(U(\gamma_{t/N})^*e^{-tL(f)/N} \big)^N \nonumber\\
&=  \lim_{N \to \infty}U(\gamma_{t/N})^{N-1}e^{-tL(f)/N} \big(U(\gamma_{t/N})^*e^{-tL(f)/N} \big)^{N-1} \nonumber\\
&=  \lim_{N \to \infty} \alpha^\prime e^{-tL(f_{N-1})/N}U(\gamma_{t/N})^{N-1} \big(U(\gamma_{t/N})^*e^{-tL(f)/N} \big)^{N-1} \nonumber\\
&= \lim_{N \to \infty} \alpha_N \,\, e^{-tL(f_{N-1})/N} \, e^{-tL(f_{N-2})/N} \cdots e^{-tL(f)/N}, \nonumber
\end{align}
with the limit converging in the strong operator topology.
The precise value of $\alpha^\prime$, which is not important for us, is $\alpha_0 \cdots \alpha_{N-2}$.
What is important, especially for Corollary \ref{cor: action of annuli on non-vacuum} is that $\alpha^\prime$ and $\alpha_N$ are determined by the $\alpha_j$ of \eqref{eqn: stress energy projective covariance}.

We showed above that for all $0 \le s \le t$, $f$ vanishes on $\varphi_s(I^\prime)$, and that $\gamma_s(\varphi_s(z)) = z$ for all $z \in I^\prime$.
Hence for all $0 \le j < N$, we have that $f_j$ vanishes on $I^\prime$, and so $L(f_j)$ is affiliated with $\cA_c(I)$ (see e.g. the more general statement \cite[Prop. 5.1.6]{Weiner05}).
Thus the calculation \eqref{eqn: local generalized annuli trotter calculation} exhibits $U(\gamma_t) e^{-t L(\rho)}$ as a strong limit of operators in $\cA_c(I)$, and we obtain the desired conclusion that $U(\gamma_t) e^{-t L(\rho)} \in \cA_c(I)$ as well.
\end{proof}

\begin{Corollary}\label{cor: generalized annuli localizable}
If $X \in \GAnn_I^{in}$ then  $\pi(X) \in \Ann^{\ell}_I(\cH_{c,0})$.
If $(\tilde X, X) \in \GAnn_I$, then $(\pi(\tilde X)^*, \pi(X)) \in \Ann_I(\cH_{c,0})$.
\end{Corollary}
\begin{proof}
To show that $\pi(X) \in \Ann_I^\ell(\cH_{c,0})$, by definition we must produce a diffeomorphism $\hat \gamma$ such that $U(\hat \gamma)\pi(X) \in \cA_c(I)$, and this is demonstrated by Lemma \ref{lem: localizable}.
Since $e^{-tL(\rho)}$ maps $L(c,0)$ onto itself, it has dense image.
Moreover, if $X = (\psi,\rho,t,\gamma)$ and $\tilde X = (\tilde \psi, \tilde \rho, \tilde t, \tilde \gamma)$ lie in $\GAnn_I$, then by definition 
$$
(\psi \circ \phi_t \circ \gamma^{-1})|_{I^\prime} = (\tilde \psi \circ \tilde \phi_{\tilde t} \circ \tilde \gamma^{-1})|_{I^\prime}.
$$
Thus we may choose a common $\hat \gamma$ in Lemma \ref{lem: localizable} such that $U(\hat \gamma)\pi(X), U(\hat \gamma) \pi(\tilde X) \in \cA_c(I)$.
Hence $(\pi(\tilde X)^*, \pi(X)) \in \Ann_I(\cH_{c,0})$ by definition.
\end{proof}

\begin{Corollary}\label{cor: existence of geometric ann}
For all intervals $I$, there exists a $(\tilde X, X) \in \GAnn_I$ such that $\Int(\tilde X, X) \ne \emptyset$.
Moreover, $\tilde X$ may be taken to be a diffeomorphism, and we may take the M\"{o}bius transformation $\psi$ of $X=(\psi,\rho,t,\gamma)$ to be the identity.
There exists an interval $I$ such that $\tilde X$ may be taken to be the identity $(\id, 0, 0, \id)$.
\end{Corollary}
\begin{proof}
In \cite[\S2.4]{GRACFT1} we demonstrated the existence of a $X=(\id,\rho,t,\gamma) \in \GAnn^{in}_I$ for some interval $I$, with the property that $0 \in \varphi_t(\interior{\bbD})$.
If $\tilde \gamma$ is a diffeomorphism from Lemma \ref{lem: localizable}, and we identify $\tilde \gamma$ with $(\id, \rho, 0, \tilde \gamma^{-1}) \in \GAnn_I$, then $(\tilde \gamma, X) \in \GAnn_I$.
For any diffeomorphism $\hat \gamma$, we have $X \circ \hat \gamma^{-1} := (\id,\rho,t, \hat \gamma \circ \gamma) \in \GAnn_{\hat \gamma(I)}$, and $(\tilde \gamma \circ \hat \gamma^{-1}, X \circ \hat \gamma^{-1}) \in \GAnn_{\hat \gamma(I)}$.
Choosing $\hat \gamma = \tilde \gamma$ provides an example with $\tilde X = \id$.
Alternatively, for any interval $J$ we may choose $\hat \gamma$ so that $\hat \gamma(I)= J$.
\end{proof}

In the process of proving Lemma \ref{lem: localizable}, we also computed the action of elements of $\scA_I^{in}$ on (non-vacuum) positive energy representations.
We record this observation now.
\begin{Corollary}\label{cor: action of annuli on non-vacuum}
Let $(\id,\rho,t,\id) \in \GAnn_I$, and decompose $\rho(z) = f(z) + i g(z)$. 
Let $\gamma_s = \operatorname{Exp}(i s z g(z) \frac{d}{dz})$, regarded as a subgroup of $\Diff^{(\infty)}(S^1)$.
Equip $\gamma_s$ with a lift to $\Diff^{(\infty)}_c(S^1)$.
Then $U(\gamma_t)e^{-tL(\rho)} \in \cA_c(I)$ and
$$
\pi_{c,h,I}(U(\gamma_t)e^{-tL(\rho)}) = U_{c,h}(\gamma_t)e^{-tL(\rho)}.
$$
\end{Corollary}
\begin{proof}
We showed in Lemma \ref{lem: localizable} that $U(\gamma_t)e^{-tL(\rho)} \in \cA_c(I)$.
Moreover, we showed that
$$
U(\gamma_t) e^{-t L(\rho)} = \lim_{N \to \infty} \alpha_N \,\, e^{-tL(f_{N-1})/N} \, e^{-tL(f_{N-2})/N} \cdots e^{-tL(f)/N}
$$
where $f_j$ is as in the proof of the lemma, and the scalars $\alpha_N$ are determined by the $\alpha_j$ of \eqref{eqn: stress energy projective covariance}.
In fact, this argument applies in any positive energy representation with central charge $c$, and the scalars $\alpha_j$ only depend on $f$, $j$, and $\gamma_s$, but not on the particular representation (see \cite[Lem. 6.2.21]{Weiner05}, and originally \cite{GoWa85}).
Hence
$$
U_{c,h}(\gamma_t) e^{-t L(\rho)} = \lim_{N \to \infty} \alpha_N \,\, e^{-tL(f_{N-1})/N} \, e^{-tL(f_{N-2})/N} \cdots e^{-tL(f)/N}
$$
for the same scalars $\alpha_N$.
Since $\pi_{c,h,I}(e^{-tL(f_j)/N}) = e^{-tL(f_j)/N}$ by \cite[Lem. 5.5]{Weiner17} and $\pi_I$ is strongly continuous, the desired result follows.
\end{proof}

With this preliminary work done, we can verify the necessary axioms to show that $\scA_I$ is a system of generalized annuli.

\begin{Lemma}
The collections $\scA_I \subset \scA_{I}^{out} \times \scA_{I}^{in}$ are a system of generalized annuli.
\end{Lemma}
\begin{proof}
First, in Corollary \ref{cor: generalized annuli localizable} above we verified that $\scA_I^{in} \subset \Ann^{\ell}_I(\cH_{c,0})$ and that $\scA_I \subset \Ann_I(\cH_{c,0})$.
While there are many other axioms to check in the definition of a system of generalized annuli, almost all of them hold for $\scA_I$ by definition and we will omit the routine verification.
The only property of a system of generalized annuli of central charge $c$ that needs to be verified is number 6, that for every $I$ there is a $(1,A) \in \scA_I$ such that $\Int(A) \ne \emptyset$.

If $(1,A) \in \scA_I$ then $(1,r_\theta A r_{-\theta}) \in \scA_{r_\theta(I)}$, so it suffices to show that we may find such $A$ for intervals $I$ of arbitrarily short length.
In Corollary \ref{cor: existence of geometric ann} we showed that for some interval $I$, there exists a $X=(\operatorname{id},\rho,t,\gamma) \in \GAnn_I^{in}$ such that $(\id, X) \in \GAnn_I$ and $\Int(X) \ne \emptyset$.
Let $\varphi_t$ be the semigroup of univalent maps generated by $\rho$.
Observe that if we can find a M\"{o}bius transformation $\psi$ such that $0 \in \psi(\varphi_t(\interior{\bbD}))$, then $X_\psi:=(\psi,\rho,t,\psi \circ \gamma)$ satisfies
$$
(1,\pi(X_\psi)) = (U(\psi)U(\psi)^*,U(\psi)\pi(X)U(\psi)^*) \in \scA_{\psi(I)}.
$$

Hence we must find a $\psi$ with $0 \in \psi(\varphi_t(\interior{\bbD}))$ and such that $\psi(I)$ has arbitrarily short length.
Let $J \subset \varphi_t(S^1) \cap S^1$ be an interval.
We may find a disk $D \subset \varphi_t(\bbD)$ which is tangent to $S^1$ at a point $\zeta \in J$.
Let $\psi_s \subset \Mob$ be the one-parameter group which fixes $\zeta$ and its antipodal point $\tilde \zeta$, parametrized so that as $s \to + \infty$, $\tilde \zeta$ is an attractive fixed point.
Then for sufficiently large $s$, it is clear that $0 \in \psi_s(D) \subset \psi_s(\varphi_t(\bbD))$, and also that $\lim_{s \to \infty} \abs{I} = 0$.
Thus choosing arbitrarily large $s$, we get the desired $(1,X_{\psi_s}) \in \scA_{\psi_{s}(I)}$.

Finally, to check that the systems of generalized annuli of central charge $c$ form a single system of generalized annuli, we must check compatibility with tensor products.
This is provided by the explicit formula in Corollary \ref{cor: action of annuli on non-vacuum}.
\end{proof}

\newpage

\section{Bounded localized vertex operators}\label{sec: BLVO}

In Section \ref{sec: generalized annuli}, we introduced the notion of a system of generalized annuli $\scA_I \subset \scA_I^{out} \times \scA_I^{in}$.
Generalized annuli $A \in \scA_I^{in}$ and $B \in \scA_I^{out}$ act on $\cH_\pi$ for any representation of a Virasoro net $\scA_c$, and following Convention \ref{conv: annuli} we denote the corresponding operators again by $A$ and $B$.
In particular, we will be interested in these actions when $\cH_\pi$ is the Hilbert space completion of a unitary module $M$ of a VOA, which naturally carries a representation of $\cA_c$.
Generalized annuli $(B,A) \in \scA_I$ have an interior $\Int(B,A)$, and for $z \in \Int(B,A)$ we will consider the insertion operators $BY(a,z)A$ for $a \in V$.
More generally, we can consider the analogous insertion operators induced by modules and intertwining operators.

In Section \ref{sec: def BLVO} we will describe when these insertion operators generate a conformal net, and in Sections \ref{sec: BLVO for subalgbras} and \ref{sec: fermion again} we will discuss examples of this property.
In particular, we will describe how to adapt the results of \cite{GRACFT1} to show that the free fermion conformal net can be generated by point insertions as described above.
First, however, we describe general properties of operators of the form $BY(a,z)A$, which at first glance are not even unbounded operators, but only sesquillinear forms.

\subsection{Holomorphic families of sesquilinear forms}

Let $\cH_1$ and $\cH_2$ be Hilbert spaces, and let $D_i \subset \cH_i$ be dense subspaces.
A (densely defined, sesquilinear) pairing with domain $D_1 \times D_2$ is a linear map $a:D_1 \otimes_{alg} D_2^\dagger \to \bbC$, where $D_2^\dagger$ is complex conjugate vector space.
We denote the pairing 
$$
\ip{av_1,v_2}:=a(v_1 \otimes v_2)
$$
for $v_i \in D_i.$
We denote by $\Pair(D_1,D_2)$ the space of all such pairings.

Let 
$$
D(a) = \{v_1 \in D_1 \, : \, \sup_{\substack{v_2 \in D_2\\ \norm{v_2} \le 1}} \abs{\ip{av_1,v_2}} < \infty\}.
$$ 
For $v_1 \in D(a)$, $av_1$ defines an element of the continuous dual of $D_2^\dagger$, which is $\cH_2$, and thus $a$ yields an unbounded operator $\cH_1 \to \cH_2$, with domain $D(a)$;
we denote this operator again by $a$.
For $v_1 \in D(a)$ and $v_2 \in D_2$, the expression $\ip{av_1,v_2}$ defined above coincides with the inner product $\ip{av_1,v_2}_{\cH_2}$.
We say that the pairing $a$ is closable (resp. bounded) if the unbounded operator $(a,D(a))$ is closable (resp. bounded), and if $a$ is bounded we extend it to a bounded linear map $a:\cH_1 \to \cH_2$.

Let $U \subset \bbC$ be an open subset, and let $f:U \to \Pair(D_1,D_2)$ be a function.
We define the domain of boundedness of $f$ as
$$
\dom(f) := \Interior{\{z \in U \, : \, f(z) \mbox{ is bounded}\}},
$$
and say that $f$ is pointwise bounded if $\dom(f) = U$.
We say that $f$ is very weakly holomorphic on $U$ if for all $v_i \in D_i$, the function $z \mapsto \ip{f(z)v_1,v_2}$ is holomorphic.
A function $g:U \to \bbC$ is called locally bounded (on $U$) if for every $z \in U$ there is a neighborhood $V$ of $z$ such that $\{\abs{g(z)} \, : \, z \in V\}$ is a bounded set.

\begin{Lemma}\label{lem: bounded implies holomorphic}
Let $f: U \to \Pair(D_1,D_2)$ be very weakly holomorphic and pointwise bounded, and assume that for all $\xi_i \in \cH_i$ the function $z \mapsto \ip{f(z)\xi_1,\xi_2}$ is locally bounded on $U$.
Then $f$ defines a (strongly) holomorphic function $U \to \cB(\cH_1,\cH_2)$.
\end{Lemma}
\begin{proof}
By the Banach-Steinhaus theorem, $\norm{f(z)\xi}$ is locally bounded, and indeed then so is $\norm{f(z)}$.
Hence if $\xi_i \in \cH_i$ we can choose sequences $v_i^{(n)} \in D_i$ which converge to $\xi_i$, and we obtain locally uniform convergence $\bip{f(z)v_1^{(n)},v_2^{(n)}} \to \ip{f(z)\xi_1,\xi_2}$.
Hence $\ip{f(z)\xi_1,\xi_2}$ is a holomorphic function, making $f:U \to \cB(\cH_1,\cH_2)$ a weakly holomorphic function.
By standard results (e.g. \cite[Ch. 5]{TaylorLay}), this means that $f$ is a (strongly) holomorphic function.
\end{proof}

There are several other equivalent hypotheses one can use, such as replacing pointwise boundedness and local boundedness of matrix coefficients by uniform local boundedness of the matrix coefficients. That is, if one assumes that for all $z \in U$ there exists a neighborhood $V$ of $z$ such that
$$
\sup \{ \, \abs{\ip{f(z)v_1,v_2}} \, : \, z \in V, \, v_i \in D_i, \, \norm{v_i} \le 1\} < \infty
$$
then $f$ is pointwise bounded and $\norm{f(z)}$ is locally bounded, and one obtains the same result.

\subsection{Holomorphic operator-valued functions from intertwining operators}
\label{sec: holomorphic forms from intertwining operators}
Let $V$ be a unitary vertex operator algebra, let $M$, $N$ and $K$ be unitary $V$-modules, with Hilbert space completions $\cH_M$, $\cH_N$, and $\cH_K$, respectively.
Let $\cY \in I \binom{K}{M N}$.
For $B \in \cB(\cH_K)$ and $A \in \cB(\cH_N)$, we would like to define the insertion operator $B\cY(a,z)A$.
We think of $A$ and $B$ as being two halves of a degenerate annulus, and $\cY(a,z)$ inserting the state $a$ at a point $z$ inside the annulus.
In order to define $B\cY(a,z)A$, we need to be able to make sense of matrix coefficients for that expression on a large enough domain.

\begin{Definition}\label{def: incoming and outgoing annuli}
Let $M$ be a unitary representation of the Virasoro algebra on which $L_0$ is diagonalizable, and let $\cH$ be the Hilbert space completion of $M$.
An \emph{incoming generalized annulus} on $\cH$ is an operator $A \in \cB(\cH)$ such that $M \subset \Ran(A)$ and $A^{-1}(M)$ is dense in $\cH$.
An operator $B \in \cB(\cH)$ is an \emph{outgoing generalized annulus} if $A^*$ is an incoming generalized annulus.
We write $\Ann^{in}(\cH)$ and $\Ann^{out}(\cH)$ for the set of incoming and outgoing generalized annuli, respectively.
We write $D_{A}$ for $A^{-1}(M)$, and $D_{B}$ for $(B^*)^{-1}(M)$.
\end{Definition}

An incoming generalized annulus models the operator assigned to a (potentially degenerate) annulus embedded in $\bbD$ whose outer boundary is the unit circle, parametrized by the identity (or a rotation).
Similarly, an outgoing generalized annulus models annuli whose incoming boundary is the unit circle parametrized in the same way.

Let us now return to $\cY \in I \binom{K}{M N}$, and let $B \in \Ann^{out}(\cH_K)$ and $A \in \Ann^{in}(\cH_N)$.
For $z \in \bbC \setminus \{0\}$, equipped with a choice of $\log z$, we have  canonical pairing $B\cY(-, z)A \in \Pair(N \otimes D_A, D_B)$ given by
$$
\ip{B\cY(v_1,z)A\xi_2, \xi_3} := \ip{\cY(v_1,z)A\xi_2, B^*\xi_3}_{\cH_K}.
$$
Thus we have defined a multi-valued holomorphic function on $\bbC \setminus \{0\}$ valued in $\Pair(M \otimes D_A, D_B)$, which becomes single-valued on domains $U$ equipped with a holomorphic branch of $\log z$.
In the following, if $U \subset \bbC \setminus \{0\}$ then we will use the term `multi-valued holomorphic function on $U$' to mean a holomorphic function of $\log z$.


We will be interested in the regularized versions $B\cY(s^{L_0}-,z)A$, which will define bounded maps $\cH_M \otimes \cH_N \to \cH_K$ for appropriate choices of $A,B,s,z$.
\begin{Definition}\label{def: interior}
Let $\cY \in I \binom{K}{M N}$, $A \in \Ann^{in}(\cH_N)$, and $B \in \Ann^{out}(\cH_K)$.
For $s \ge 0$, define
$$
\Int_{\cY,s}(B,A) = \bigcup_{r > s}  \{ z \in \bbC \setminus \{0\} : B\cY(r^{L_0}-,z)A \mbox{ is bounded} \}.
$$
We simply write $\Int_{\cY}(B,A)$ for $\Int_{\cY,0}(B,A)$.
\end{Definition}
Observe that the boundedness of $B\cY(s^{L_0}-,z)A$ does not depend on a choice of $\log z$.

The usefulness of Definition \ref{def: interior} is that it provides natural domains on which to define $B\cY(s^{L_0}-,z)A$ as a holomorphic operator valued function.

\begin{Proposition}\label{prop: interior open}
Let $\cY \in I \binom{K}{M N}$, $A \in \Ann^{in}(\cH_N)$, and $B \in \Ann^{out}(\cH_K)$.
Then we have
\begin{enumerate}
\item For all $s \ge 0$, $\Int_{\cY,s}(A,B)$ is open.
\item If $s > 0$, $B\cY(s^{L_0}-,z)A$ defines a multi-valued holomorphic function $\Int_{\cY,s}(B,A) \to \cB(\cH_M \otimes \cH_N, \cH_K)$.
\item If $v \in M$, then $B\cY(v,z)A$ defines a multi-valued holomorphic function $\Int_{\cY}(B,A) \to \cB(\cH_N, \cH_K)$.
\end{enumerate}
\end{Proposition}
As indicated above, the above multi-valued holomorphic functions become single valued on subsets $U$ of the domain equipped with a holomorphic branch of $\log z$.

We will require a few lemmas to prove Proposition \ref{prop: interior open} (after the proof of Lemma \ref{lem: local norm boundedness}), but first point out an easy consequence.
\begin{Corollary}
If $F \subset \Int_\cY(A,B)$ is compact, then there exists an $s > 0$ such that $B\cY(s^{L_0}-,z)A$ defines a multi-valued holomorphic function on a neighborhood of $F$.
In particular, $\sup_{z \in K} \norm{B\cY(s^{L_0}-,z)A} < \infty$.
\end{Corollary}
\begin{proof}
Since $\Int_{\cY}(A,B) = \cup_{s > 0} \Int_{\cY,s}(B,A)$ and $F$ is compact, we must have $F \subset \Int_{\cY,s}(B,A)$ for some $s$.
The rest follows immediately from Proposition \ref{prop: interior open}.
\end{proof}

We now accumulate the necessary results to prove Proposition \ref{prop: interior open}.
\begin{Lemma}\label{lem: dilation translation bounded}
Let $M$ be a unitary representation of $\Vir$ on which $L_0$ is diagonalizable and $L_0 \ge h \ge 0$.
Let $\cH$ be the Hilbert space completion of $M$.
Let $z \in \bbC$ and $r \in \bbR_+$ with $\abs{z} + r < 1$.
Then $e^{zL_{-1}} r^{L_0}$ defines a bounded operator on $\cH$ with
$$
\norm{e^{zL_{-1}} r^{L_0}} \le \frac{(\abs{z}+r)^h}{r(1-(\abs{z}+r)^2)^{\frac12}}.
$$.
\end{Lemma}
\begin{proof}
Let $\frm = \Span \{L_{0}, L_{\pm 1}\} \subset \Vir$.
Then $M$ decomposes as a direct sum of irreducible positive energy representations $V_t$ of $\frm$ with lowest weight $t$ (see \cite[App. A]{Weiner05} for more detail).
Thus it suffices to verify the desired estimate on each $V_t$ with $t \ge h$.
If $h = t = 0$ then $e^{zL_{-1}}r^{L_0} = 1$ and the desired estimate holds.
Thus we have reduced to the case of establishing our estimate on $V_t$ with $t \ge h > 0$.

By \cite[Prop. A.2.11]{Weiner05}, such $V_t$ are the algebraic span of an orthonormal set $\{\xi_n\}_{n \ge 0}$ with $L_0 \xi_n = (t+n)\xi_n$ and $L_{-1} \xi_n = \sqrt{(n+2t)(n+1)}\xi_{n+1}$.
We will show by direct computation that $e^{zL_{-1}}r^{L_0}$ defines a bounded operator. 
Consider a finite sum $\xi = \sum_{n \ge 0} \alpha_n \xi_n \in V_t$.
Then
\begin{align*}
\ip{e^{zL_{-1}}r^{L_0}\xi,\xi_n} &= \sum_{k=0}^n \frac{1}{k!} z^k r^{n-k+t} \alpha_{n-k} \ip{L_{-1}^{k}\xi_{n-k}, \xi_n}\\
&= \sum_{k=0}^n \frac{1}{k!} z^k r^{n-k+t} \alpha_{n-k} \left(\prod_{j=0}^{k-1} \sqrt{(n-k+2t+j)(n-k+1+j)}\right)\\
\end{align*}
Hence 
\begin{align*}
\babs{\ip{e^{zL_{-1}}r^{L_0}\xi,\xi_n}} &\le 
\sum_{k=0}^n \frac{1}{k!} \abs{z}^k r^{n-k+t} \abs{\alpha_{n-k}} \left(\prod_{j=0}^{k-1}(n-k+t+j+1)\right)\\
&= r^{n+t}\sum_{k=0}^n \abs{\alpha_{n-k}}\tfrac{\abs{z}^k}{r^k} \binom{n+t}{k}\\
&\le r^{n+t}\left(\sum_{k=0}^n \abs{\alpha_{n-k}}^2 \right)^{\tfrac12} \left(\sum_{k=0}^n \frac{\abs{z}^{2k}}{r^{2k}} \binom{n+t}{k}^2\right)^{\tfrac12}\\
&\le r^{n+t}\norm{\xi} \sum_{k=0}^n \frac{\abs{z}^{k}}{r^{k}} \binom{n+t}{k}.
\end{align*}
The first step is the triangle inequality and inequality of arithmetic and geometric means, the second is the definition of the binomial coefficient, the third is Cauchy-Schwarz, and the fourth is the fact that $\sqrt{a + b} \le \sqrt{a} + \sqrt{b}$.
Let $m \in \bbZ$ be the smallest integer with $m \ge t$, and observe that for $0 \le k \le n$ we have
$\binom{n+t}{k} \le \binom{n+m}{k}$, and thus we have 
$$
\sum_{k=0}^n \frac{\abs{z}^{k}}{r^{k}} \binom{n+t}{k} \le \sum_{k=0}^{n+m} \frac{\abs{z}^k}{r^k} \binom{n+m}{k} = r^{-n-m}(\abs{z}+r)^{n+m}.
$$
Hence combining with the previous calculation we get 
$$
\babs{\ip{e^{zL_{-1}}r^{L_0}\xi,\xi_n}} \le r^{-(m-t)}\norm{\xi}(\abs{z}+r)^{n+m},
$$
and thus
$$
\norm{e^{zL_{-1}}r^{L_0}\xi} = \left(\sum_{n \ge 0} \babs{\ip{e^{zL_{-1}}r^{L_0}\xi,\xi_n}}^2 \right)^{\frac12} \le \frac{(\abs{z}+r)^t}{r} \norm{\xi} (1-(\abs{z}+r)^2)^{-\frac12}.
$$
We used that $m-t \le 1$ and $t \le m$.
Thus $e^{zL_{-1}}r^{L_0}$ extends to a bounded operator on the Hilbert space completion of $V_t$ with norm at most $\frac{(\abs{z}+r)^t}{r(1-(\abs{z}+r)^2)^{\frac12}}$. 

Returning to our positive energy representation $\cH$, we have seen that for every irreducible component $V_t$ of the action of $\frm$ we have
$$
\norm{e^{zL_{-1}}r^{L_0}}_{V_t} \le \frac{(\abs{z}+r)^t}{r(1-(\abs{z}+r)^2)^{\frac12}} \le \frac{(\abs{z}+r)^h}{r(1-(\abs{z}+r)^2)^{\frac12}}
$$
since $t \ge h$, and thus the desired norm bound holds on all of $\cH$.
\end{proof}

\begin{Remark}
The estimated obtained in Lemma \ref{lem: dilation translation bounded} is extremely bad for small $r$ and fixed $\abs{z}$. In fact, if $e^{zL_{-1}}s^{L_0}$ is bounded and $r < s$, then $\norm{e^{zL_{-1}}r^{L_0}} \le (r/s)^h \norm{e^{zL_{-1}}s^{L_0}}$ which gives the expected geometric decay in norm in the small $r$ limit.
\end{Remark}

\begin{Lemma}\label{lem: local norm boundedness}
Let $\cY \in I \binom{K}{M N}$ and let $A$ (resp. $B$) be an incoming (resp. outgoing) generalized annulus on $\cH_{N}$ (resp. $\cH_K$).
Suppose that $B\cY(r^{L_0}-,z)A$ is bounded and that $\abs{z} > r$.
Then for any $s < r$ and any $w \in z + (r-s)\bbD$ the map $B\cY(s^{L_0}-,w)A$ is bounded, with $\norm{B\cY(s^{L_0}-,w)A}$ uniformly bounded on compact subsets of $z + (r-s)\interior{\bbD}$.
\end{Lemma}
\begin{proof}
Let $s < r$, and let $w \in z + (r-s)\interior{\bbD}$.
Recall that if $v_1 \in M$, $v_2 \in N$, $v_3 \in K$ then by \cite[Eq. (5.4.21)]{FHL93}
$$
\bip{\cY(e^{(w-z)L_{-1}}v_1,z)v_2, v_3} = \ip{\cY(v_1,w)v_2,v_3},
$$
where $\log w$ is determined by analytically continuing a choice of $\log z$ to $z + (r-s)\bbD$, which does not contain $0$ by assumption.
Thus setting $X = r^{-L_0} e^{(w-z)L_{-1}} s^{L_0}$, we have
$$
\bip{B\cY(s^{L_0}v_1,w)A\xi_2,\xi_3} = \bip{B\cY(r^{L_0}Xv_1,z)A\xi_2,\xi_3}
$$
for $\xi_2 \in D_A$ and $\xi_3 \in D_B$.
Observe that $r^{-L_0} e^{(z-w) L_{-1}} = e^{\frac{z-w}{r} L_{-1}} r^{-L_0}$ so that $X = e^{\frac{w-z}{r} L_{-1}} (s/r)^{L_0}$.
Since $w \in z + (r-s)\interior{\bbD}$, we have $\frac{\abs{w-z}}{r} + \frac{s}{r} < 1$.
Hence $X$ is bounded by Lemma \ref{lem: dilation translation bounded}, with norm bounded uniformly on compact subsets of $z + (r-s)\interior{\bbD}$.
Thus we have the same for $B\cY(s^{L_0}-,w)A =(B\cY(r^{L_0}-,z)A)(X \otimes 1)$.
\end{proof}

\begin{proof}[Proof of Proposition \ref{prop: interior open}]
We first consider part 1.
If $z \in \Int_{\cY,s}(A,B)$, then there exists an $r > s$ such that $B\cY(r^{L_0}-,z)A$ is bounded.
By Lemma \ref{lem: local norm boundedness}, for $t$ with $r > t > s$, there is a neighborhood $U$ of $z$ on which $B\cY(t^{L_0}-,z)A$ is bounded, uniformly on compact subsets of $U$.
In particular, $U \subset \Int_{\cY,s}(A,B)$, completing the proof of part 1.

For part 2, observe that $B\cY(t^{L_0}-,z)A$ defines a holomorphic function 
$$
B\cY(t^{L_0}-,z): U \to \Pair(M \otimes D_A, D_B)
$$ with norms locally bounded, and thus yields a holomorphic function $U \to \cB(\cH_M \otimes \cH_N, \cH_K)$ by Lemma \ref{lem: bounded implies holomorphic}.
The same then immediately follows for $B\cY(s^{L_0}-,z)A$.

For part 3, observe that if $z \in \Int_{\cY}(B,A)$, then $B\cY(r^{L_0}-,z)A$ is bounded for some $r > 0$, and by the above work we may choose $r$ so that $B\cY(r^{L_0}-,z)A$ defines a multi-valued holomorphic function on a neighborhood $U$ of $z$.
Since $v^\prime := r^{-L_0}v \in M$, it follows that $B\cY(v,z)A$ is holomorphic on $U$ as well.
\end{proof}

We close with a simple application.
Let $V$ be a simple unitary vertex operator superalgebra, let $M$, $N$, and $K$ be unitary $V$-modules, and let $\cY \in I \binom{K}{M N}$.
Then for any $\xi \in \cH_M$ and $b \in N$ and $k \in \bbR_{\ge 0}$, the expression $p_k\cY(\xi,z)b$ expands to an infinite series in $p_k M$, where $p_k$ is the projection onto the $k$-eigenspace of $L_0$ and some value of $\log z$ has been chosen.
If each such series converges in $p_k M$, then $\cY(\xi,z)b$ defines an element of the algebraic completion $\hat K$.
It is possible that $\cY(\xi,z)b \in \cH_K$, and this occurs exactly when the homogeneous components of this vector have square summable norm.
In the following, the assertion $\cY(\xi,z)b \in \cH_K$ implicitly assumes the convergence of the series defining the homogeneous components.
Observe that this occurs for $\cY(s^{L_0}\xi,z)b$ if there exists some $A \in \Ann^{in}(\cH_N)$ such that $z \in \Int_{\cY,s}(1,A)$.
%

\begin{Lemma}\label{lem: components in closure}
Let $V$ be a simple unitary vertex operator superalgebra, let $M$, $N$, and $K$ be unitary $V$-modules, and let $\cY \in I \binom{K}{M N}$.
Suppose that there is some  $A \in \Ann^{in}(\cH_N)$ such that $\Int_{\cY}(1,A) \ne \emptyset$, and let $z \in \Int_{\cY}(1,A)$.
Then for a homogeneous vector $b \in N$, let
$$
K_z :=  \{\cY(a,z)b : a \in M\} \subset \cH_K.
$$
Then for all $k \in \bbR$ we have $a_{(k)}b \in \overline{K_z}$.
\end{Lemma}
\begin{proof}
By the definition of $\Ann^{in}(\cH_N)$, there is some $\eta \in \cH_N$ such that $A\eta = b$.
Hence by the definition of $\Int_{\cY}(1,A)$, there is an $r > 0$ such that the map $\xi \mapsto \cY(r^{L_0}\xi,z)b = \cY(r^{L_0}\xi,z)A\eta$ is bounded.
It follows that $\overline{K_z}$ contains $\cY(r^{L_0}\xi,z)b$ for all $\xi \in \cH_M$.
As in the proof of Lemma \ref{lem: local norm boundedness}, we may apply Lemma \ref{lem: dilation translation bounded} to obtain a number $s>0$ with the property that if $\abs{z-w}$ is sufficiently small, then $X=r^{-L_0}e^{(w-z)L_{-1}}s^{L_0}$ is bounded.
In this case, for homogeneous $a \in M$ we have $\cY(a,w)b = s^{-\Delta_a} \cY(r^{L_0}Xa,z)b \in \overline{K_z}$.
Extending linearly, we have $\cY(a,w)b \in \overline{K_z}$ for all $a \in M$, provided $\abs{z-w}$ is sufficiently small.
Thus if $a \in M$, we have $e^{i \theta L_0} \cY(a,z)b = e^{i\theta \Delta_b} \cY(e^{i \theta L_0}a,e^{i \theta} z)b \in \overline{K_z}$ provided $\abs{\theta}$ is sufficiently small.
Thus for sufficiently small $\theta$, $e^{i \theta L_0}$ maps $K_z$ into $\overline{K_z}$.
It follows that for such $\theta$, $\overline{K_z}$ is invariant under $e^{i \theta L_0}$.
But then it follows that $\overline{K_z}$ is invariant under $e^{i \theta L_0}$ for all $\theta$, and thus $\overline{K_z}$ decomposes as a direct sum of eigenspaces for $L_0$.
For $a$ homogeneous, since $\cY(a,z)b \in K_z$, it follows that $a_{(k)}b \in \overline{K_z}$.
We now extend linearly to non-homogeneous $a$, completing the proof.
\end{proof}

\subsection{Definition of bounded localized vertex operators}\label{sec: def BLVO}

Fix a system $\scA$ of generalized annuli, and a unitary vertex operator superalgebra $V$ of central charge $c$.
Let $(B,A) \in \scA_{c,I}$ (which we abbreviate to $\scA_I$ in the following).
In particular, this implies that $A \in \Ann^{\ell}_I(\cH_{c,0})$, and similarly $B \in \Ann^r_I(\cH_{c,0})$.
As $V$ decomposes as a direct sum of irreducible unitary highest weight representations of $\Vir_c$, there is a representation $\pi$ of the Virasoro net $\cA_c$ on the Hilbert space completion $\cH_V$.
Thus we have actions $\pi_I(A)$ and $\pi_I(B)$ on the positive energy representation $\cH_V$ which are well-defined up to sign since the conformal weights of $V$ are half integers (see Definition \ref{def: generalized annulus action}). 
Following Convention \ref{conv: annuli} we omit the $\pi_I$.

\begin{Definition}[Bounded insertions] \label{def: bounded insertions}
We say that $V$ has bounded insertions for $\scA$ if for every interval $I \in \cI$ and every $(B,A) \in \scA_I$ we have $B \in \Ann^{out}(\cH_V)$, $A \in \Ann^{in}(\cH_V)$ and $\Int(B,A) \subset \Int_Y(B,A)$.
\end{Definition}

The requirement that $\Int(B,A) \subset \Int_Y(B,A)$ ensures that for every $z \in \Int(B,A)$ there exists an $s > 0$ such that $BY(s^{L_0}-,z)A$ extends to a bounded operator in $\cB(\cH_V \otimes \cH_V,\cH_V)$.

\begin{Definition}[Local algebras of insertions]
If $V$ has bounded insertions for $\scA$, we define the local algebras corresponding to intervals $I \subset S^1$ by
\begin{equation}\label{eqn: local algebra def}
\cA_V(I) = \{ BY(a,z)A : (B,A) \in \scA_I, z \in \Int(B,A), a \in V \}^{\prime\prime} \vee \{BA : (B,A) \in \scA_I\}^{\prime\prime}.
\end{equation}
\end{Definition}
$$
\begin{tikzpicture}[baseline={([yshift=-.5ex]current bounding box.center)}]
	\coordinate (a1) at (150:1cm);
	\coordinate (b1) at (270:1cm);
	\coordinate (c1) at (210:1.5cm);
	\fill[red!10!blue!20!gray!30!white] (a1)  .. controls ++(250:.6cm) and ++(120:.4cm) .. (c1) .. controls ++(300:.4cm) and ++(190:.6cm) .. (b1) arc (270:510:1cm);
	\draw ([shift=(270:1cm)]0,0) arc (270:510:1cm);
	\coordinate (a) at (120:1cm);
	\coordinate (b) at (240:1cm);
	\coordinate (c) at (180:.25cm);
	\fill[fill=red!10!blue!20!gray!30!white] (0,0) circle (1cm);
	\fill[fill=white] (a)  .. controls ++(210:.6cm) and ++(90:.4cm) .. (c) .. controls ++(270:.4cm) and ++(150:.6cm) .. (b) -- ([shift=(240:1cm)]0,0) arc (240:480:1cm);
	\draw ([shift=(240:1cm)]0,0) arc (240:480:1cm);
	\draw (a) .. controls ++(210:.6cm) and ++(90:.4cm) .. (c);
	\draw (b) .. controls ++(150:.6cm) and ++(270:.4cm) .. (c);
	\draw (a) arc (120:150:1cm) (a1) .. controls ++(250:.6cm) and ++(120:.4cm) .. (c1);
	\draw (b1) .. controls ++(190:.6cm) and ++(300:.4cm) .. (c1);
	
%
	\node at (197:.9cm) {\textbullet};
	\node at (204:.7cm) {$z$};
	\node at (183:.9cm) {$v$};
%
%
%
\end{tikzpicture}
$$
Note that we are required to include the generators $BA$ separately in the case when $(B,A)$ has empty interior.
Observe that $\Int(B^*,A^*) = \{ \overline{z}^{-1} : z \in \Int(A,B)\}$ and $(BY(a,z)A)^* = A^* Y(\tilde a,\overline{z}^{-1}) B^*$, where 
$$
\tilde a = e^{\overline{z} L_1} (-\overline{z}^{-2})^{L_0} \kappa \theta a,
$$
from the definition of a unitary $V$-module.
Hence the generating sets in \eqref{eqn: local algebra def} are self-adjoint, and $\cA_V(I)$ is a von Neumann algebra.
It becomes a superalgebra with respect to the involution $x \mapsto \Gamma_V x \Gamma_V$.
Finally, observe that if $z \in \Int(B,A) \cap \Int_{Y,s}(B,A)$, then for all $\xi \in \cH_V$ we have $BY(s^{L_0}\xi,z)A \in \cA_V(I)$, as it is a norm limit of elements of the given generating set, and moreover the map $\xi \mapsto BY(s^{L_0}\xi,z)A$ is norm continuous in $\xi$ and holomorphic in $z$.

\begin{Definition}[Bounded localized vertex operators]\label{def: bounded localized vertex operators}
We say that $V$ has bounded localized vertex operators if $\cA_V(I)$ and $\cA_V(J)$ super-commute whenever $I$ and $J$ are disjoint intervals.
\end{Definition}

As a relatively easy consequence one obtains that $\cA_V$ is a conformal net.

\begin{Proposition}
Let $V$ be a simple unitary vertex operator superalgebra with bounded localized vertex operators.
Then $\cA_V$ is a Fermi conformal net which is diffeomorphism covariant with respect to the representation $U^V:\Diff^{(2)}_c(S^1) \to U(\cH_V)$ induced by the conformal vector of $V$.
\end{Proposition}
\begin{proof}
The proof is essentially identical to \cite[Prop. 4.8]{GRACFT1}.
The hypothesis in the cited result that $\cA_V(I)$ commute with localized diffeomorphisms is now redundant, as localized diffeomorphisms automatically lie in $\cA_V(I)$ under the present definitions.
To obtain cyclicity of the vacuum, simply note that by the definition of a system of generalized annuli, we may take some $(1,A) \in \scA_I$ with non-empty interior, and then apply the argument of \cite{GRACFT1}.
\end{proof}

\subsection{Bounded localized vertex operators for subalgebras}\label{sec: BLVO for subalgbras}

The difficulty, however, is in producing examples of vertex operator superalgebras which have bounded localized vertex operators.
This was accomplished in \cite{GRACFT1} by first verifying the property for the free fermion, and then showing that the property is preserved under tensor products and under taking subtheories.
We take a similar approach with our modified definition.
In Section \ref{sec: fermion again} we verify that the free fermion has bounded localized vertex operators in the new regime.
We repeat here the results of \cite{GRACFT1}, sketching proofs for the convenience of the reader.
We fix a system of generalized annuli $\scA$.

In order to verify that a tensor product is an incoming annulus if and only if the tensor factors are, we will require the following exercise in Hilbert space operator theory.
\begin{Lemma}\label{lem: split simple tensors}
Let $x_i \in \cB(\cH_i,\cK_i)$ for $i=1,2$, and suppose $\varphi_1 \otimes \varphi_2 \in \Ran(x_1 \otimes x_2)$.
Let $\xi \in \ker(x_1 \otimes x_2)^\perp$ be the vector such that $x\xi = \varphi_1 \otimes \varphi_2$.
Then we may decompose $\xi = \xi_1 \otimes \xi_2$ for $\xi_i \in \cH_i$.
In particular, if $\varphi_1 \otimes \varphi_2 \ne 0$ then $\varphi_i \in \Ran(x_i)$.
\end{Lemma}
\begin{proof}
To show that $\xi$ is a simple tensor, it suffices to verify that for all $\psi_i,\eta_i \in \cH_i$ we have
\begin{equation}\label{eqn: ip cocycle}
\ip{\xi, \psi_1 \otimes \psi_2}\ip{\xi, \eta_1 \otimes \eta_2}
=
\ip{\xi, \psi_1 \otimes \eta_2}\ip{\xi, \eta_1 \otimes \psi_2}.
\end{equation}
Observe that $\ker(x_1 \otimes x_2)^\perp$ is the closed span of vectors of the form $\psi_1 \otimes \psi_2$ with $\psi_i \in \ker(x_i)^\perp = \overline{\Ran(x_i^*)}$.
It therefore suffices to verify \eqref{eqn: ip cocycle} when all $\eta_i,\psi_i \in \Ran(x_i^*)$.
This may be directly checked for such vectors using the fact that $x\xi$ is a simple tensor.
Hence we indeed have $\xi = \xi_1 \otimes \xi_2$ for some $\xi_i \in \cH_i$.
It follows that $x\xi_i$ is a scalar multiple of $\varphi_i$.
If $\varphi_1 \otimes \varphi_2 \ne 0$ then these scalars must be non-zero, and thus $\phi_i \in \Ran(x_i)$.
\end{proof}

\begin{Proposition}\label{prop: BLVO tensor product}
Let $V_1$ and $V_2$ be simple unitary vertex operator superalgebras, and let $V = V_1 \otimes V_2$.
Then $V$ has bounded localized vertex operators if and only if $V_1$ and $V_2$ do.
In this case, $\cA_{V_1} \otimes \cA_{V_2} = \cA_{V}$.
\end{Proposition}
\begin{proof}
The detailed proof is \cite[Prop. 4.11]{GRACFT1}.
We first consider the equivalence of bounded insertions for $V$ versus $V_1$ and $V_2$.
Let $c_i$ be the central charge of $V_i$, and $c=c_1+c_2$ be the central charge of $V$.
Let $A \in \scA^{in}_{I,c}$.
For the moment, we denote the action of $A$ on $\cH_V$ by $A^V$, and equip $A^V$ with a splitting $A^{V} = A^{V_1} \otimes A^{V_2}$ on $\cH_V = \cH_{V_1} \otimes \cH_{V_2}$ (which is possible by the definition of system of generalized annuli).
Observe that by Lemma \ref{lem: split simple tensors}, $A^V \in \Ann^{in}(\cH_V)$ if and only if both $A^{V_i} \in \Ann^{in}(\cH_{V_i})$.
We can make a similar argument if $B \in \scA^{out}_I$.
Moreover, if $(B,A) \in \scA_{I,c}$ and $a_i \in V_i$, then we have
$$
B^V Y^V(a_1 \otimes a_2, z) A^V = (B^{V_1} Y^{V_1}(a_1,z) A^{V_1}) \otimes (B^{V_2} Y^{V_2}(a_2,z) A^{V_2}).
$$
Since $\Int(B,A) = \Int(B_i,A_i)$ by definition, we can conclude that if $V_1$ and $V_2$ have bounded insertions, then so does $V$.
Moreover, since every $(B, A) \in \scA_{I,c_i}$ arises as $(B^{V_i}, A^{V_i})$ from such a splitting, we may reverse the above argument and conclude that $V_1$ and $V_2$ have bounded insertions if $V$ does.

We now turn to localized vertex operators.
It is straightforward to check that if $V_1$ and $V_2$ have localized vertex operators, then so does $V_1 \otimes V_2$, and that $\cA_{V_1} \otimes \cA_{V_2} = \cA_{V_1 \otimes V_2}$.
Consider the reverse direction, and let $(B,A) \in \scA_I$ and $(\tilde B, \tilde A) \in \scA_J$ with $I$ and $J$ disjoint.
Following Convention \ref{conv: annuli}, we will omit reference to the central charge as well as the spaces that generalized annuli act on.
Let $z \in \Int(B,A)$, $\tilde z \in \Int(\tilde B, \tilde A)$ and $a,\tilde a \in V_1$.
By assumption $B Y^V(a \otimes \Omega,z) A$ supercommutes with $B Y^V(\tilde a \otimes \Omega,\tilde z) A$.
Observe that 
$$
B Y^V(a \otimes \Omega,z)A = B Y^{V_1}(a,z)A \otimes BA
$$
and similarly 
$$
\tilde B Y^V(\tilde a \otimes \Omega, \tilde z) \tilde A = \tilde B Y^{V_1}(\tilde a,\tilde z)\tilde A \otimes \tilde B \tilde A.
$$
We know that $BA$ and $\tilde B \tilde A$ are even and commute, as they are generalized annuli localized in disjoint intervals (see Remark \ref{rmk: localized annuli commute}).
Since $BA \otimes BA$ lies in the factor $\cA_V(I)$ and $\tilde B \tilde A \otimes \tilde B \tilde A$ commutes with $\cA_V(I)$ we have
$$
(BA \otimes BA)(\tilde B \tilde A \otimes \tilde B \tilde A) \ne 0.
$$
In particular, $(BA)(\tilde B \tilde A) \ne 0$.
Thus we may deduce that $B Y^{V_1}(a,z)A$ and $\tilde B Y^{V_1}(\tilde a,\tilde z)\tilde A$ supercommute, which shows that $V_1$ has bounded localized vertex operators.
One can address $V_2$ in a similar fashion.
\end{proof}

\begin{Proposition}\label{prop: BLVO subalgebra}
Let $V$ be a vertex operator superalgebra with bounded localized vertex operators, and let $W \subset V$ be a unitary subalgebra.
Then $W$ has bounded localized vertex operators, and $\cA_V$ has a subnet $\cB$ with $\cH_B = \cH_W$ and $\cB|_{\cH_W} = \cA_W$.
\end{Proposition}
\begin{proof}
We sketch the proof from \cite[Thm. 4.12]{GRACFT1}.
By Proposition \ref{prop: BLVO tensor product} it suffices to consider the case of a conformal subalgebra, as $W \otimes W^c$ is isomorphic to a conformal subalgebra, where $W^c$ is the coset.
Bounded insertions is clearly inherited by conformal subalgebras, so we just need to check locality.
Let
$$
\cB(I) = \{ BY^V(a,z)A : (B,A) \in \scA_I, z \in \Int(B,A), a \in W \}^{\prime\prime} \vee \{BA : (B,A) \in \scA_I\}^{\prime\prime} \subset \cA_V(I)
$$
be the subalgebra with insertions only from $W$.
It is straightforward to check that $\cB$ is a covariant subnet of $\cA_V$, that $\cH_B = \cH_W$, and that 
$$
BY^V(a,z)A|_{\cH_W} = BY^W(a,z)A.
$$
Thus $\cA_W(I) = \cB(I)e_W$ defines a conformal net, and, in particular, is local.
\end{proof}

\subsection{The free fermion revisited}
\label{sec: fermion again}

In \cite{GRACFT1}, we introduced a notion called `bounded localized vertex operators,' and demonstrated that the free fermion VOA had this property. 
Definition \ref{def: localizable operators} is a slightly more general definition of bounded localized vertex operators, depending on a choice of system of generalized annuli (Definition \ref{def: system of generalized annuli}).
The arguments of \cite{GRACFT1} show that the free fermion VOA has bounded localized vertex operators in the sense of Definition \ref{def: localizable operators}, with respect to the system of generalized annuli given in Definition \ref{def: choice of generalized annuli}.
For the convenience of the reader, we will now review the structure of the argument of \cite{GRACFT1}, and outline the minor modifications necessary to adapt them to the context of this article.

We start from the beginning to fix notation, but a more detailed version of this discussion may be found in \cite[\S2.1.1]{GRACFT1}.
Let $H$ be a Hilbert space, and $\Lambda H = \bigoplus_{j=0}^\infty \Lambda^j H$ be the associated exterior Hilbert space.
For $\xi \in H$, let $a(\xi)$ be the operator on $\Lambda H$ given by exterior multiplication by $\xi$ on the left.
Then $a(\xi)$ is bounded with $\norm{a(\xi)} = \norm{\xi}$, and $a(\xi)^*$ is contraction with $\xi$.
The $C^*$-algebra generated by the $a(\xi)$ on $\Lambda H$ is the canonical anticommutation relations algebra $\CAR(H)$.
For $p$ a projection on $H$, we set $H_p = (pH)^* \oplus (1-p)H$, and $\cF_{p} = \Lambda H_p$.
There is an irreducible representation $\pi_p:\CAR(H) \to \cB(\cF_p)$ given by $\pi_p(a(\xi)) = a((p\xi)^*)^* + a((1-p)\xi)$.
There are natural isomorphisms $\CAR(H \oplus K) \cong \CAR(H) \hotimes \CAR(K)$ (where $\hotimes$ indicates the graded tensor product of algebras), and if $q$ is a projection on $K$ we have a natural isomorphism $\cF_{p \oplus q} \cong \cF_p \otimes \cF_q$ of $\CAR(H \oplus K)$ representations.
We will be interested in the case when $H$ is a direct sum of copies of $L^2(S^1)$, and $p$ is a direct sum of projections onto the Hardy space $H^2(\bbD)$.

We use the term \emph{generalized disk} to refer to the closure of the bounded connected component of a $C^\infty$ Jordan curve.
Let $D_{0,1}, \ldots D_{0,k}, D_1$ be a family of generalized disks, and assume that the $D_{0,i}$ are pairwise disjoint and all contained in $D_1$.
Set $D_0 = D_{0,1} \sqcup \cdots \sqcup D_{0,k}$, and let $\Sigma = D_1 \setminus \interior{D_0}$.
Its outgoing boundary $\partial_1 \Sigma$ is $\partial D_1$, and its incoming boundary $\partial_0 \Sigma$ is $\partial D_0$.
We give $\partial_1 \Sigma$ the positive orientation about $D_1$, and we give $\partial_0 \Sigma$ the negative orientation about $D_0$.

\begin{Definition}[Degenerate genus zero Riemann surfaces]
A \emph{degenerate (genus zero) Riemann surface} is a subset $\Sigma \subset \bbC$ arising from the above construction.
Let $\cDR$ be the collection of such $\Sigma$ equipped with orientation preserving diffeomorphisms $\beta_i: \bigsqcup_{j \in \pi_0(\partial_i\Sigma)} S^1 \to \partial_i \Sigma$ and a choice of smooth square roots ${\beta_i}^\prime(z)^{1/2}$ where $$\beta_i^\prime = \frac{1}{iz} \frac{d}{d\theta} \beta_i(e^{i \theta})_{e^{i \theta}=z}.$$
Let $\cDA$ be the subcollection for which $\partial_0\Sigma$ consists of exactly one circle.

A holomorphic function on $\Sigma \in \cDR$ is by definition a continuous function $F$ on $\Sigma$ such that $F|_{\partial_i\Sigma}$ is smooth and $F$ is holomorphic on the interior of $\Sigma$, and we write $\cO(\Sigma)$ for the algebra of holomorphic functions.
The boundary Hilbert spaces are $H_i = \bigoplus_{j \in \pi_0(\partial_i \Sigma)} L^2(S^1)$, and the polarizations are projections $p_i$ with  $p_i H_i = \bigoplus_{j \in \pi_0(\partial_i \Sigma)} pH$, where $pH = H^2(\bbD)$ is the Hardy space of the disk.

If $\beta_i$ are the above parametrizations and $F \in \cO(\Sigma)$, then the pullback is given by $\beta_i^*F(z) = \beta_i^\prime(z)^{1/2}F(\beta_i(z)) \in H_i$.
The Hardy space $H^2(\Sigma) \subset H_1 \oplus H_0$ is given by
$$
H^2(\Sigma) = \operatorname{cl} \Big( \{ ({\beta_1}^* F,{\beta_0}^* F) : F \in \cO(\Sigma) \} \Big)
$$
where the closure is taken in $H_1 \oplus H_0$.

The incoming and outgoing Fock spaces are $\cF_i = \cF_{p_i}$, and we identify $\CAR(H_i)$ representations $\cF_i \cong \bigotimes_{j \in \pi^0(\partial_i)} \cF$, where $\cF = \cF_p$ is the Fock space associated to the standard polarization of $L^2(S^1)$.
An even operator $T \in \cB(\cF_0,\cF_1)$ satisfies the Segal commutation relations for $\Sigma$ if
$$
a(f_1)T = Ta(f_0), \qquad a(\overline{zf_1})^*T = T a(\overline{zf_0})^*
$$
for every $(f_1,f_0) \in H^2(\Sigma) \subset H_1 \oplus H_0$.
We write $E(\Sigma)$ for the space of operators satisfying the Segal commutation relations for $\Sigma$.
\end{Definition}

In \cite{GRACFT1}, $\cO(\Sigma)$ was defined as functions holomorphic in a neighborhood of $\Sigma$, but by Mergelyan's theorem \cite[\S20]{BigRudin}, any function holomorphic on $\Sigma$ in our sense is a uniform limit of functions holomorphic in a neighborhood of $\Sigma$, and so the resulting Hardy spaces coincide.

As in \cite[Prop. 3.8]{GRACFT1}, one has:

\begin{Proposition}
Let $\Sigma \in \cDR$.
Then $\dim(E(\Sigma)) \le 1$.
\end{Proposition}

The argument of \cite[Prop. 3.8]{GRACFT1} produces a densely defined operator $T$ with the property that $\dim(E(\Sigma)) = 1$ if and only if $T$ is bounded, in which case $E(\Sigma) = \Span T$.

\begin{Definition}
Let $\Sigma \in \cDR$.
We say that $\Sigma$ is \emph{bounded} if $\dim(E(\Sigma)) = 1$.
\end{Definition}

When $\Sigma$ is genuinely a manifold, the analysis in \cite{Ten16} (based on \cite{SegalDef}) shows that $\Sigma$ is bounded, and by \cite[Prop. 3.9]{GRACFT1} $E(\Sigma)$ defined here coincides with the Segal CFT operator constructed in \cite{Ten16}.
In \cite{GRACFT1}, we showed that some degenerate $\Sigma$ are bounded as well.
This property is unchanged by reparametrizing the boundary of $\Sigma$ \cite[Prop. 3.6]{GRACFT1}.

\begin{Proposition}\label{prop: ESigma injective}
Let $T \in E(\Sigma)$ and assume $T \ne 0$.
Then $T$ is injective.
If $\partial_0 \Sigma \ne \emptyset$, then $T$ has dense image as well.
\end{Proposition}
\begin{proof}
One may show that $\ker T$ is invariant under $\CAR(H_0)$ just as in \cite[Prop. 4.11]{Ten16}, and by the irreducibility of the representation on $\cF_0$ the injectivity follows.
The same argument shows that if $\partial_0 \ne \emptyset$, then $\operatorname{coker} T$ is invariant under $\CAR(H_1)$.
\end{proof}

We use this to show compatibility with composition.

\begin{Proposition}\label{prop: segal cft composition}
Let $\Sigma,\tilde \Sigma \in \cDP$, and assume that $\partial_1 \Sigma = \partial \tilde D_{0,j}$, where $\tilde D_{0,j}$ is one of the incoming disks of $\tilde \Sigma$.
Assume that this boundary component has the same parametrization for $\partial_1 \Sigma$ and $\partial_0 \tilde \Sigma$.
Let $\Sigma_+ = \Sigma \cup \tilde \Sigma$.
If $\Sigma$ and $\tilde \Sigma$ are bounded, then so is $\Sigma_+$. 
Moreover, $E(\Sigma_+) = E(\tilde \Sigma) \circ_j E(\Sigma)$, where $\circ_j$ means inputting the output of elements of $E(\Sigma)$ into the $j$th tensor factor of the domain of elements of $E(\tilde \Sigma)$.
\end{Proposition}
\begin{proof}
If we choose non-zero $T \in E(\Sigma)$ and $\tilde T \in E(\tilde \Sigma)$, then $T_+ := \tilde T \circ_j T$ is non-zero by Proposition \ref{prop: ESigma injective}.
Then by direct argument $T_+$ satisfies the Segal commutation relations for $\Sigma_+$, and thus $T_+ \in E(\Sigma_+)$.
\end{proof}

\begin{Definition}
Let $\cDA_{0} \subset \cDA$ be the subcollection consisting of $\Sigma = D_1 \setminus D_0$ such that $0 \in \interior{D_0}$.
If $\Sigma \in \cDA_0$ with boundary parametrizations $\beta_i$, let $\Sigma^*$ be the degenerate annulus $D_0^* \setminus D_1^*$, where $D_i^*$ is the image of $D_i$ under the map $r(z) = \overline{z}^{-1}$.
We equip $\Sigma^*$ with boundary parametrizations $\hat \beta_i:=r \circ \beta_{1-i}$, with a choice of square root $\hat \beta_i^\prime(z)^{1/2} = \frac{\overline{\beta_{1-i}^\prime(z)^{1/2}}}{z {\overline{\beta_{1-i}(z)}}}$.
\end{Definition}

\begin{Proposition}\label{prop: Segal CFT adjoint}
Let $\Sigma \in \cDA_0$.
Then $E(\Sigma^*) = E(\Sigma)^*$, with adjoints taken elementwise.
\end{Proposition}
\begin{proof}
From the definition of the Segal commutation relations, it suffices to prove that $(f_1, f_0) \in H^2(\Sigma)$ if and only if $(\overline{zf_0}, \overline{zf_1}) \in H^2(\Sigma^*)$.
Since the map $(f_1,f_0) \mapsto (\overline{zf_0}, \overline{zf_1})$ is unitary and involutive, and ${\Sigma^*}^* = \Sigma$, it suffices to prove that if $f_i = \beta_i^* F$ then $(\overline{zf_0}, \overline{zf_1}) \in H^2(\Sigma^*)$.
And indeed, if we set $G(z) = z \overline{F(\overline{z}^{-1})} \in \cO(\Sigma^*)$, then we have $\hat \beta_i^*G = \overline{z \beta_{1-i}^* F}$.
\end{proof}

There is a natural notion of equivalence of $\Sigma,\tilde \Sigma^\prime \in \cDP$, which is to say an orientation preserving bijection $F$ such that $F \in \cO(\Sigma)$, $F^{-1} \in \cO(\Sigma^\prime)$, and such that $F$ is compatible with boundary parametrizations.
If $\Sigma$ and $\tilde \Sigma$ are equivalent, then $H^2(\Sigma) = H^2(\tilde \Sigma)$, and so $E(\Sigma) = E(\tilde \Sigma)$.
To be careful, we would need $F$ to be a spin-isomorphism, and for $F$ to preserve the square roots of the derivatives of the boundary parametrizations.
However, in the following we will only be interested in the boundedness of various $\Sigma$, which is independent of the choice of square root.

Let $X = (\psi,\rho,t,\gamma) \in \GAnn^{in}$, so that $\pi(X) = U(\psi)e^{-tL(\rho)}U(\gamma)^*$, well-defined up to scalar (Definition \ref{def: geometric annuli}).
Associated to $\rho$, we have a semigroup $\varphi_t \in \scG$ (Definition \ref{def: good semigroup}).
From $X$, we may construct $\Sigma_X \in \cDA_0$ given by $\bbD \setminus \varphi_t(\interior{\bbD})$, with boundary parametrizations given by $\beta_1 = \psi$ and $\beta_0 = \varphi_t \circ \gamma^{-1}$, along with some choice of square root (which is irrelevant for our analysis here).
Alternatively, $\Sigma_X$ is equivalent to the degenerate annulus $\hat\Sigma_X$ given by $\bbD \setminus \psi(\varphi_t(\interior{\bbD}))$, with $\beta_1 = \id$ and $\beta_0 = \psi \circ \varphi_t \circ \gamma^{-1}$.

\begin{Lemma}[]\label{lem: calculation of annuli}
Let $X \in \GAnn^{in}$.
Then $\Sigma_X$ is bounded and $E(\Sigma_X) = \Span \pi(X)$.
\end{Lemma}
\begin{proof}
The case when $\psi = \id$ is \cite[Prop. 3.16]{GRACFT1}.
The general case follows from the reparametrization formula \cite[Prop. 3.6]{GRACFT1}.
\end{proof}

The Fock space $\cF$ of $L^2(S^1)$, taken with respect to the standard polarization $H^2(\bbD)$, carries a positive energy representation of $\Vir_1$, and finite energy vectors $\cF^0$ have the structure of a vertex operator superalgebra, called the free (Dirac) fermion; we denote the state-field correspondence by $Y$ for the remainder of the section.
We refer the reader to \cite[\S5.1]{Kac98} or \cite[Ex. 2.22]{GRACFT1} for more details.

\begin{Lemma}\label{lem: segal cft annuli are incoming}
Let $\Sigma \in \cDA_0$. 
Assume that $\partial_1\Sigma = S^1$ and that it is parametrized by the identity, and assume that $\Sigma$ is bounded.
If $T \in E(\Sigma)$ is non-zero, then $T \in \Ann^{in}(\cF)$.
\end{Lemma}
\begin{proof}
We must show that $T^{-1}(\cF^0)$ is dense and that $T(\cF)$ contains $\cF^0$.
By definition of $\cDA_0$, we have $\Sigma = D_1 \setminus \interior{D_0}$ for generalized disks $D_i$ such that $0 \in \interior{D_0}$.
By assumption, $\partial_1\Sigma = S^1$ so $D_1 = \bbD$.

Choose $r > 0$ is sufficiently small that $r\bbD \subset \interior{D_0}$, and let $\Sigma_r = \bbD \setminus r\interior{\bbD}$, with the natural parametrizations.
Let $\tilde \Sigma = D_0 \setminus r\interior{\bbD}$, with outgoing boundary parametrized the same as the incoming boundary of $\Sigma$.
Since $\tilde \Sigma$ is non-degenerate (as $r$ was chosen sufficiently small), we have that $\tilde \Sigma$ is bounded by \cite[Thm. 4.9]{Ten16}.
Similarly, by \cite[Prop. 5.2]{Ten16}, $E(\Sigma_r)$ is spanned by $r^{L_0}$.
Thus by Lemma \ref{prop: segal cft composition} and Lemma \ref{lem: calculation of annuli} we may choose $\tilde T \in E(\tilde \Sigma)$ such that $r^{L_0} = T\tilde T$.
It follows immediately that $T(\tilde T \cF^0) = \cF^0$, and in particular $\cF^0 \subset T\cF$.
Moreover, since $\tilde T$ has dense image $\tilde T \cF^0$ is dense, and thus $T^{-1} \cF^0$ is dense as well.
\end{proof}

\begin{Corollary}\label{cor: GAnn are incoming}
Let $X \in \GAnn^{in}$, and let $\pi(X)$ be its action on $\cF$.
Then $\pi(X) \in \Ann^{in}(\cF)$.
\end{Corollary}
\begin{proof}
Write $X = (\psi, \rho,t,\gamma)$, and consider the degenerate annulus $\hat \Sigma_X$.
As a subset of $\bbC$, we have $\hat \Sigma_X = \bbD \setminus \psi(\varphi_t(\interior{\bbD}))$, and by definition of $\GAnn^{in}$ we have $0 \in \psi(\varphi_t(\interior{\bbD}))$.
Hence $\hat \Sigma_X \in \cDA_0$.
Recall that $\hat \Sigma_X$ is equivalent to $\Sigma_X$, so that $\hat \Sigma_X$ is bounded and $\pi(X)$ spans $E(\hat \Sigma_X)$ by Lemma \ref{lem: calculation of annuli}.
By definition, $\hat \Sigma_X$ has outgoing boundary equal to the unit circle, with identity parametrization, and so we may apply Lemma \ref{lem: segal cft annuli are incoming} to $\pi(X)$ to obtain the desired conclusion.
\end{proof}

Let $X,\tilde X \in \GAnn^{in}$, and let $\Sigma = \hat \Sigma_X$, and $\tilde \Sigma = \hat \Sigma_{\tilde X}^*$.
By Lemma \ref{lem: calculation of annuli}, $E(\Sigma)$ is spanned by $\pi(X)$ and additionally invoking Proposition \ref{prop: Segal CFT adjoint} we have that $E(\tilde \Sigma)$ is spanned by $\pi(\tilde X)^*$.
By Lemma \ref{lem: segal cft annuli are incoming}, $\pi(\tilde X)^* Y(s^{L_0}-,z)\pi(X)$ defines a holomorphic family of sesquilinear forms (see Section \ref{sec: holomorphic forms from intertwining operators}).
The technical tools developed in \cite{GRACFT1} provide a way of showing that these forms are bounded for appropriate values of $s$ and $z$, and at the same time identifying the resulting bounded maps with operators assigned to degenerate Riemann surfaces.
The following lemma was proven as \cite[Thm. 3.21]{GRACFT1} in the case where $\tilde T = 1$ and $\psi = \id$, but the generalization to the stated case is straightforward\footnote{%
The M\"{o}bius transformation $\psi$ does not appear in \cite[Thm. 3.21]{GRACFT1}, but the argument only requires that $\Sigma$ be a degenerate annulus inside the unit disk with outgoing boundary parametrized by the identity.
This is why it is important to work with $\hat \Sigma_X$ instead of $\Sigma_X$ (see the discussion before Lemma \ref{lem: calculation of annuli}).
We are using crucially here that $\psi$ is a M\"{o}bius transformation and not just an arbitrary diffeomorphism.
}.

\begin{Lemma}\label{lem: segal cft insertions are vertex operators}
Let $\Sigma, \tilde \Sigma \in \cDA_0$ with $\partial_1 \Sigma = \partial_0 \tilde \Sigma = S^1$ and both of these boundary components parametrized by the identity.
Assume that $\Sigma$ and $\tilde \Sigma$ are bounded, and let $T \in E(\Sigma)$ and $\tilde T \in E(\tilde \Sigma)$ be non-zero.
Let $\Sigma_+ = \Sigma \cup \tilde \Sigma$.
Suppose $z + s\bbD \subset \interior{\Sigma_+}$, and let $\Sigma_{z,s} = \Sigma_+ \setminus (z + s \interior{\bbD})$, with the new boundary component parametrized by $w \mapsto z + sw$.
Then $z \in \Int_{Y,s}(\tilde T, T)$ and $E(\Sigma_{z,s}) = \Span \tilde T Y(s^{L_0}-,z)T$.
\end{Lemma}

We now prove the main result of this section, updating the result \cite[Thm. 4.4]{GRACFT1} that the free fermion has bounded localized vertex operators to the definition being used in this article (Definition \ref{def: bounded localized vertex operators}).
At the same time, we will prove that the resulting net $\cA_{\cF}$ agrees with the free fermion Fermi conformal net $\cM(I) := \{ a(f), a(f)^* : f \in L^2(I) \}$ (see \cite[\S15]{Wa98}).

\begin{Theorem}\label{thm: FF BLVO}
The free fermion vertex operator superalgebra has bounded localized vertex operators for the system of generalized annuli $\scA$ of Definition \ref{def: choice of generalized annuli}.
The Fermi conformal net $\cA_{\cF}$ is equal to $\cM$.
\end{Theorem}
\begin{proof}
We first consider bounded insertions (Definition \ref{def: bounded insertions}).
Let $(\tilde X, X) \in \GAnn_I$.
That is $X = (\psi,\rho,t,\gamma) \in \GAnn^{in}_I$ and $\tilde X = (\tilde \psi, \tilde \rho, \tilde t, \tilde \gamma) \in \GAnn^{in}_I$, and 
$$
(\psi \circ \phi_t \circ \gamma^{-1})|_{I^\prime} = (\tilde \psi \circ \tilde \phi_{\tilde t} \circ \tilde \gamma^{-1})|_{I^\prime},
$$
and the image of $I^\prime$ under these maps lies inside $S^1$.
Then $(\pi(\tilde X)^*, \pi(X)) \in \scA_I$ is a typical element of our system of generalized annuli.
By Lemma \ref{lem: segal cft annuli are incoming}, $\pi(X) \in \Ann^{in}(\cF)$ and $\pi(\tilde X)^* \in \Ann^{out}(\cF)$.

We now have to check that $\Int(\pi(\tilde X)^*, \pi(X)) \subset \Int_{Y}(\pi(\tilde X)^*, \pi(X))$.
By Lemma \ref{lem: segal cft insertions are vertex operators}, $\interior{\Sigma_+} \subset \Int_{Y}(\pi(\tilde X)^*, \pi(X))$.
We now write $\Sigma(X)$, $\tilde \Sigma(\tilde X)$ and $\Sigma_+(\tilde X, X)$ to emphasize the dependence on $X$ and $\tilde X$.
Unpacking the definition of $\Int(\pi(X))$ (Definition \ref{def: choice of generalized annuli}) and $\Int(\pi(\tilde X)^*, \pi(X))$ (Definition \ref{def: system of generalized annuli}), we have 
$$
\Int(\pi(\tilde X)^*, \pi(X)) = \bigcup \interior{\Sigma_+}(\tilde Y, Y)
$$
where the union is taken over $(\tilde Y, Y) \in \GAnn_I$ with $\pi(Y) = \pi(X)$ and $\pi(\tilde Y) = \pi(\tilde X)$.
Repeating the application of Lemma \ref{lem: segal cft insertions are vertex operators}, we obtain 
$$
\interior{\Sigma_+}(\tilde Y, Y) \subset \Int_{Y}(\pi(\tilde Y)^*, \pi(Y)) = \Int_{Y}(\pi(\tilde X)^*, \pi(X))
$$
for all such $(\tilde Y, Y)$. 
Hence $\Int(\pi(\tilde X)^*, \pi(X)) \subset \Int_{Y}(\pi(\tilde X)^*, \pi(X))$, which completes the proof of bounded insertions.

Now the proof that $\cF$ has bounded localized vertex operators (that is, that $\cA_{\cF}(I)$ and $\cA_{\cF}(J)$ commute when $I$ and $J$ are disjoint), is essentially identical to \cite[Thm. 4.4]{GRACFT1}.
It suffices to show that $\cA_{\cF}(I) \subset \cM(I) = \cM(I^\prime)^\prime$, where $\cM(I^\prime)$ is the free fermion net introduced above.
Let $(\tilde X, X) \in \GAnn_I$ and let $z \in \Int_{Y,s}(\pi(\tilde X)^*, \pi(X))$.
Let $\Sigma_{z,s} = (\hat \Sigma_X \cup \hat \Sigma_{\tilde X}^*) \setminus (z + s \interior{\bbD})$ as before.
Then $\pi(\tilde X)^*Y(s^{L_0}-,z)\pi(X) \in E(\Sigma_{z,s})$ by Lemma \ref{lem: segal cft insertions are vertex operators}.
Given $f \in C_c^\infty(I^\prime)$, we extend it to a holomorphic function on $\Sigma_{z,s}$ by setting it to be $0$ on $\Sigma_{z,s} \setminus I^\prime$.
Since
$$
(\psi \circ \phi_t \circ \gamma^{-1})|_{I^\prime} = (\tilde \psi \circ \tilde \phi_{\tilde t} \circ \tilde \gamma^{-1})|_{I^\prime},
$$
the boundary parametrization of $\partial_1\Sigma_{z,s}$ agrees with the parametrization of $\partial_0\Sigma \subset \partial_0\Sigma_{z,s}$ on $I^\prime$.
Hence by the Segal commutation relations, elements of $E(\Sigma_{z,s})$ commute with $a(f)$ and $a(\overline{zf})^*$.
But as $f$ ranges over $C_c^\infty(I^\prime)$, such elements generate $L^2(I^\prime)$
Hence $E(\Sigma_{z,s}) \subset \cM(I^\prime)^\prime = \cM(I)$.
\end{proof}

\newpage

\section{Geometric realization of VOA modules}\label{sec: BLVO modules}

\subsection{Definition and basic properties}\label{sec: def of module}

We fix a system of generalized annuli $\scA_I \subset \scA_I^{out} \times \scA_I^{in}$ to use throughout the section.
Let $V$ be a simple unitary vertex operator superalgebra of central charge $c$ with bounded localized vertex operators, and let $M$ be a unitary $V$-module.
The local algebras $\cA_V(I)$ are generated by operators of the form $BY(a,z)A$ where $(B,A) \in \scA_I$.
The Hilbert space completion $\cH_M$ carries a representation of the Virasoro net $\cA_c$ arising from the positive energy representation corresponding to the modes of $Y^M(\nu,x)$, and thus the generalized annuli $A$ and $B$ act on $\cH_M$.
The representation $\pi^M$ of $\cA_V$ corresponding to $M$ should satisfy 
\begin{equation}\label{eqn: piM formula initial}
\pi^M(BY(a,z)A) = BY^M(a,z)A,
\end{equation}
where following Convention \ref{conv: annuli} we have allowed $A$ and $B$ to act on $\cH_M$ without explicitly indicating this in the notation.

There are several potential issues with this expression.
The first is that we require $A \in \Ann^{in}(\cH_M)$ and $B \in \Ann^{out}(\cH_M)$ in order for the sesquilinear form $BY^M(a,z)A$ to be densely defined (see Definition \ref{def: incoming and outgoing annuli}).
Second, the assertion \eqref{eqn: piM formula initial} should be interpreted as including the assumption that the right-hand side defines a bounded operator.

Another issue is that the actions of $A$ and $B$ on $\cH_V$ and $\cH_M$ in \eqref{eqn: piM formula initial} are only well-defined up to multiples of $e^{2\pi i L_0}$.
Recall from Definition \ref{def: generalized annulus action} that if $A \in \Ann^{\ell}_I(\cH_{c,0})$, then the action of $A$ on a representation $\cH_\pi$ of the Virasoro net $\cA_c$ is given by 
$$
\pi_I(A) = U^\pi(\tilde \gamma)^*\pi_I(U(\gamma)A)
$$
where $\gamma \in \Diff_c(S^1)$ satisfies $U(\gamma)A \in \cA_c(I)$, and $\tilde \gamma$ is a lift of $\gamma$ to $\Diffcinfty$.
However, if $(B,A) \in \scA_I$ then by definition $(B,A) \in \Ann_I$, so that we may choose a common $\gamma$ such that $BU(\gamma)^*, U(\gamma)A \in \cA_c(I)$, and if we choose a common lift $\tilde \gamma$, the meaning of $BY(a,z)A$ is
\begin{equation}\label{eqn: unpacking module def}
BY(a,z)A = \pi_I(BU(\gamma)^*)U^V(\tilde \gamma) Y(a,z) U^V(\tilde \gamma)^* \pi_I(U(\gamma)A).
\end{equation}
Observe that the value of this expression is unchanged when replacing $\tilde \gamma$ by a different lift, and thus $BY(a,z)A$ is unambiguously defined.
Similarly, $BY^M(a,z)A$ is unambiguously defined for a unitary module $M$.

\begin{Definition}\label{def: piM}
Let $V$ be a simple unitary vertex operator superalgebra with bounded localized vertex operators, and let $M$ be a unitary $V$-module.
Suppose that for every interval $I$ and every $(B,A) \in \scA_I$ we have $A \in \Ann^{in}(\cH_M)$, $B \in \Ann^{out}(\cH_M)$.
Then the $\cA_V$-representation corresponding to $M$, if it exists, is the representation $\pi^M$ on the super Hilbert space $\cH_M$ such that for every $I$ and $(B,A)$ as above, and every $z \in \Int(B,A)$ and $a \in V$ we have
\begin{equation}\label{eqn: piM formulas}
\pi^M_I(BY(a,z)A) = BY^M(a,z)A.
\end{equation}
\end{Definition}

As described above, the actions of $A$ and $B$ on $\cH_M$ which arise in Definition \ref{def: piM} are obtained by integrating the representation of the Virasoro algebra corresponding to $Y^M(\nu,x)$ to a representation of the appropriate Virasoro net.

\begin{Remark}\label{rem: piM on empty interior annuli}
If $(B,A) \in \scA_I$ and $\Int(B,A) \ne \emptyset$, then it follows immediately from the definitions that $\pi^M_I(BA) = BA$ by evaluating at $a=\Omega$.
In fact, this relation holds automatically even when $\Int(B,A) = \emptyset$.
Since $(B,A) \in \Ann_I(\cH_V)$ we have $BA \in \cA_c$, and so the action of $BA$ on $\cH_M$ is precisely given by the action of the Virasoro net (see also Remark \ref{rmk: localized annuli commute} and preceding).
\end{Remark}

As a result of Remark \ref{rem: piM on empty interior annuli}, a representation corresponding to $M$ is determined on the generating set of $\cA_V(I)$ given in \eqref{eqn: local algebra def}, and thus $\pi^M$ is unique if it exists.

\begin{Proposition}\label{prop: joint continuity of module action}
Let $V$ be a simple unitary vertex operator superalgebra with bounded localized vertex operators, let $M$ be a unitary $V$-module, and suppose that $\pi^M$ exists.
Let $(B,A) \in \scA_I$ for some interval $I$ and let $z \in \Int(B,A)$.
Then $z \in \Int_{Y^M}(B,A)$.
That is, for some $s >0 $ the map $V \otimes M \to \cH_M$ given by
$$
a \otimes b \mapsto BY^M(s^{L_0}a,z)Ab
$$
extends to a bounded map $\cH_V \otimes \cH_M \to \cH_M$.
\end{Proposition}
\begin{proof}
Since $\cA_V(I)$ is a type III factor (see Proposition \ref{thmFermiNetProps}) and $\cH_V$ and $\cH_M$ are separable Hilbert spaces, there exists a unitary $U: \cH_V \to \cH_M$ such that $\pi^M_I(x) = UxU^*$ for all $x \in \cA_V(I)$.
Hence for $a \in V$, we have $BY^M(s^{L_0}a,z)A = UBY(s^{L_0}a,z)AU^*$.
Since $z \in \Int(B,A)$ and $V$ has bounded insertions, for $s$ sufficiently small the right-hand side defines a bounded map $\cH_V \otimes \cH_M \to \cH_M$, and so the left-hand side does as well.
\end{proof}

We now establish basic properties of the correspondence $M \leftrightarrow \pi^M$.

\begin{Proposition}\label{prop: unitary equivalence of piM}
Let $V$ be a simple unitary vertex operator superalgebra with bounded localized vertex operators, let $M$ be a unitary $V$-module, and suppose that $\pi^M$ exists. 
Let $U$ be the positive energy representation on $\cH_V$ corresponding to Virasoro field $Y(\nu,x)$ of $V$, and let $U^M$ be the positive energy representation corresponding to $Y^M(\nu,x)$.
Then we have the following.
\begin{enumerate}
\item $\pi_I^M(U(\gamma)) = U^M(\gamma)$ for all $\gamma \in \Diff_c(I)$, and so $\pi^M$ is diffeomorphism covariant with respect to $U^M$.
\item If $\tilde M$ is another unitary $V$-module then $M$ is unitarily isomorphic to $\tilde M$ if and only if $\pi^{\tilde M}$ exists and is unitarily equivalent to $\pi^M$.
\end{enumerate}
\end{Proposition}
\begin{proof}
First we show (1).
If $\gamma \in \Diff_c(I)$, then $(1,U(\gamma)) \in \scA_I$ and so by Remark \ref{rem: piM on empty interior annuli} we have $\pi^M_I(U(\gamma)) = U^M(\gamma)$.
By Theorem \ref{thm: reps are diff covariant} it now follows that $\cA_V$ is diffeomorphism covariant with respect to $U^M$.

We now prove (2), first establishing the easy direction.
Suppose that $\tilde M$ is unitarily equivalent to $M$ via a unitary $u:\cH_{\tilde M} \to \cH_{M}$ which takes $M$ onto $\tilde M$.
Hence $uY^{\tilde M}(a,x)u^* = Y^{M}(a,x)$ for all $a \in V$, and in particular $u L^{\tilde M}_n u^* = L^M_n$ as endomorphisms of $M$.
It follows by a standard argument that $uU^{\tilde M}(\gamma)u^* = U^M(\gamma)$ for all $\gamma \in \Diff^{(\infty)}_c(S^1)$.
Thus $\cH_M$ and $\cH_{\tilde M}$ are equivalent representations of the Virasoro net $\cA_c$.

Let $(B,A) \in \scA_I$ and $z \in \Int(B,A)$.
Write $A^M$ and $A^{\tilde M}$ to distinguish the actions of $A$ on these two Hilbert spaces, and similarly for $B$.
Since $u$ intertwines the actions of the Virasoro nets on these spaces, we have $uA^{\tilde M}u^* = A^M$ and similarly for $B$.
Hence
$$
u \pi^{\tilde M}_I(BY(a,z)A) u^* = uB^{\tilde M}Y(a,z)A^{\tilde M} u^*
B^{M}Y^M(a,z)A^{M} u^* =
\pi^{M}_I(BY(a,z)A).
$$

Conversely, suppose that $\pi^{\tilde M}$ exists and we have an even unitary $u:\cH_{\tilde M} \to \cH_{M}$ exhibiting an equivalence of $\pi^M$ and $\pi^{\tilde M}$.
Then by part (1), we have 
\begin{equation}\label{eqnDiffRepsEquivalent}
u U^{\tilde M}(\gamma) u^* = U^M(\gamma)
\end{equation} 
for any $\gamma$ lying in some $\Diff_c(I)$, and thus for any $\gamma \in \Diff_c^{(\infty)}(S^1)$ (here we use that $\Diff_c^{(\infty)}(S^1)$ is generated by the local diffeomorphism groups \cite[Lem. 17(ii)]{HenriquesColimits}). 
In particular $u U^{\tilde M}(r_\theta) u^* = U^M(r_\theta)$, and so as self-adjoint operators we have $u L_0^{\tilde M} u^* = L_0^M$.
Thus $u$ takes $\tilde M$ onto $M$ (and similarly for the even/odd subspaces $\tilde M^i$ and $M^i$).

Now let $(B,A) \in \scA_I$ and $z \in \Int(B,A)$.
Since $u$ is an equivalence of positive energy representations we have $uA^{\tilde M} u^* = A^M$ just as above, and similarly for $B$.
It follows that $A^{\tilde M} \in \Ann^{in}(\cH_{\tilde M})$ and $B^{\tilde M} \in \Ann^{out}(\cH_{\tilde M})$.
By assumptions, $uB^{\tilde M} Y^{\tilde M}(a,z)A^{\tilde M}u^* = B^{M} Y^{M}(a,z)A^{M}$, and so $B^{M}uY^{\tilde M}(a,z)u^*A^{M} = B^{M} Y^M(a,z) A^M$. 
Hence $uY^{\tilde M}(a,z)u^* = Y^M(a,z)$ as sesquilinear forms on $M \times M$.
For $b \in M$ with $a$ and $b$ homogeneous, we may project onto homogeneous subspaces of $M$ to obtain $ua^{\tilde M}_{(k)}u^*b = a^M_{(k)}b$ for all $k \in \bbZ$.
Hence $u$ is an equivalence of VOA modules.
\end{proof}

\subsection{Submodules and direct sums}

When $\pi^M$ exists, we have a bijective correspondence between subrepresentations of $\pi^M$ and submodules of $M$.

\begin{Proposition}\label{prop: reps from sums of modules}
Let $V$ be a simple unitary vertex operator superalgebra with bounded localized vertex operators.
\begin{enumerate}
\item Let $M$ and $\{M_i\}_{i \in S}$ be unitary $V$-modules with $S$ a countable set, and suppose that $M = \bigoplus_{i \in S} M_i$.
Then $\pi^M$ exists if and only if every $\pi^{M_i}$ does, and if these hold then $\pi^{M} = \bigoplus_{i \in S} \pi^{M_i}$.
\item If $N$ is a submodule of $M$ and $\pi^M$ exists, then so does $\pi^N$ and $\pi^M|_{\cH_N} = \pi^N$.
\item If $\pi^M$ exists and $\cK$ is a subrepresentation of $\pi^M$ then there exists a $V$-submodule $N$ of $M$ such that $\cK = \cH_N$.
\item If $\pi^M$ exists, then $\pi^M$ is irreducible if and only if $M$ is.
\end{enumerate}
\end{Proposition}
\begin{proof}
First (1).
Let $(B,A) \in \scA_I$.
We write $A^M$ for the action of $A$ on $\cH_M$, which by definition can be factored $A^M = U^M(\gamma)\pi_I(x)$ for some $\gamma \in \Diff^{(\infty)}_c(S^1)$ and $x$ in the Virasoro subnet $\cA_c$ of $\cA_V(I)$.
As a representation of $\cA_c$, $\cH_M$ decomposes as a direct sum $\bigoplus \cH_{M_i}$, and thus $A^M = \bigoplus A^{M_i}$, where $A^{M_i}$ is the action of $A$ on $\cH_{M_i}$ coming from the representation of $\cA_c$.
Similarly, $B^M = \bigoplus B^{M_i}$, and we may decompose $BY^M(a,z)A = \bigoplus_{i \in S} B Y^{M_i}(a,z)A$.
Moreover $A^M \in \Ann^{in}(\cH_M)$ if and only if every $A^{M_i} \in \Ann^{in}(\cH_{M_i})$.
Hence if all $\pi^{M_i}$ exist, then $\pi^M = \bigoplus_{i \in S} \pi^{M_i}$ (and so in particular, $\pi^M$ exists).
Conversely, if $\pi^M$ exists then the compression to each $\cH_{M_i}$ yields $\pi^{M_i}$.

Next, (2). Let $N \subset M$ be a submodule, and let $p_N$ be the orthogonal projection of $\cH_M$ onto $\cH_N$.
Since $M$ and $N$ share a common grading from $L_0$, one may apply a simpler version of the argument from Proposition \ref{propModuleProjection} to show that $pM = N$, and therefore $(1-p)M \subset M$.
Thus we may decompose $M = N \oplus N^\perp$, and $\pi^{M}|_{\cH_N} = \pi^N$ by the first part of the proposition.

On to (3).
If $\cK$ is a subrepresentation of $\pi^M$, then in particular it is invariant under the Virasoro subnet of $\cA_V$.
Hence $\cK$ is invariant under the full $\Diff^{(\infty)}_c(S^1)$, and in particular under the rotation subgroup.
It follows that $\cK = \cH_N$, where $N$ is the algebraic span of $L_0$-eigenvectors in $\cK$.

Choose some $(1,A) \in \scA_I$ with $\Int(A) \ne \emptyset$, and let $z \in \Int(1,A)$.
Let $a \in V$ and $b \in N$ be homogeneous.
By assumption $\cK$ is invariant under $Y^M(a,z)A$.
Since $A \in \Ann^{in}(\cH_M)$, there is some $\eta \in \cH_M$ such that $A\eta = b$.
Since $A$ lies in the von Neumann algebra generated by $\cA_V$, $\cK$ and $\cK^\perp$ are invariant under $A$ and so $\eta \in \cK$.
Hence $Y^M(a,z)b \in \cK$, from which it follows that $a_{(k)}b \in N$ for all $k \in \bbZ$.
We conclude that $N$ is a submodule of $M$.

Finally, (4) is an immediate consequence of (2) and (3).
\end{proof}

\subsection{Restriction to subalgebras}
\label{sec: restriction to subalgebras}


We now state our main theorem on construction of conformal net representations corresponding to modules over vertex operator superalgebras via descent from a larger algebra.

\begin{Theorem}\label{thmConstructionOfNetRepresentations}
Let $V$ be a simple unitary vertex operator superalgebra with bounded localized vertex operators, let $W$ be a unitary subalgebra of $V$.
Let $N$ be a unitary $V$-module such that $\pi^N$ exists, and let $M$ be a $W$-submodule of $N$.
Then $\pi^M$ exists.
\end{Theorem}

We will give the proof of Theorem \ref{thmConstructionOfNetRepresentations} at the end of Section \ref{sec: restriction to subalgebras} after some preliminary results.
The case when $W$ is a conformal subalgebra is straightforward.

\begin{Lemma}\label{lem: module existence conformal subcase}
Theorem \ref{thmConstructionOfNetRepresentations} holds when $W$ is a conformal subalgebra. 
\end{Lemma}
\begin{proof}
Since $N$ is both a unitary $V$-module and a unitary $W$-module, we can separately consider $\pi^{N;V}$, the representation of $\cA_V$ corresponding to $N$, and $\pi^{N;W}$, the representation of $\cA_W$ corresponding to $N$.
It is clear that $\pi^{N;V}|_{\cA_W} = \pi^{N;W}$, so in particular the latter exists.
Since $M$ is a $W$-submodule of the $W$-module $N$,  the desired conclusion follows from Proposition \ref{prop: reps from sums of modules}.
\end{proof}

When $W$ is not a conformal subalgebra of $V$, we will establish Theorem \ref{thmConstructionOfNetRepresentations} by embedding $W$ in the conformal subalgebra of $V$ of the form $W \otimes W^c$.
Thus we will need to understand how representations $\pi^M$ behave under tensor product.
The challenge is to show that if $\pi^{M_1 \otimes M_2}$ exists, then so do $\pi^{M_i}$ and $\pi^{M_1 \otimes M_2} = \pi^{M_1} \otimes \pi^{M_2}$.
This follows quite easily once one knows that 
\begin{equation}\label{eqn: splitting}
\pi^{M_1 \otimes M_2}(\cA_{V_1}(I) \otimes 1) \subset \cB(\cH_{M_1}) \otimes 1
\end{equation}
 and similarly for the second net.
However, this by itself is quite tricky to prove, because the definition of $\pi^{M_1 \otimes M_2}$ does not directly give information about $\pi^{M_1 \otimes M_2}|_{\cA_{V_1}}$.
Our strategy is to verify \eqref{eqn: splitting} when the $V_i$ are Virasoro VOAs, which is Lemma \ref{lem: Vir modules tensor split} below.
Before that lemma, we will need a few preparatory results (Lemmas \ref{lem: Virasoro is right} - \ref{lem: Virasoro net splitting}).
Once we establish Lemma  \ref{lem: Vir modules tensor split}, our main result about tensor products of modules and representations is Proposition \ref{propTensorProductModuleExists}, which is followed by the proof of Theorem \ref{thmConstructionOfNetRepresentations}.

\begin{Lemma}\label{lem: Virasoro is right}
Let $V = \Vir_c$ be the Virasoro vertex operator algebra at a unitary central charge $c$, and assume that $V$ has bounded localized vertex operators.
Let $\cA_c$ be the Virasoro net with central charge $c$.
Then $\cA_{V} = \cA_c$.
\end{Lemma}
\begin{proof}
By definition, if $\gamma \in \Diff_c(I)$ then $(1,U(\gamma)) \in \scA_I$, where $U=U_{c,0}$.
Hence $U(\gamma) \in \cA_V(I)$.
Thus $\cA_c(I) \subset \cA_V(I)$, and so $\cA_c$ is a subnet of $\cA_V$.
By Proposition \ref{propSubnetTrivial}, we have $\cA_c = \cA_V$.
\end{proof}

\begin{Lemma}\label{lem: Virasoro net splitting}
Let $\pi$ be an irreducible representation of $\cA_{c_1} \otimes \cA_{c_2}$.
Then $\pi$ is equivalent to $\pi_{(c_1,h_1)} \otimes \pi_{(c_2,h_2)}$, where $\pi_{(c_i,h_i)}$ is the representation of $\cA_{c_i}$ on the irreducible positive energy representation with lowest weight $h_i$.
\end{Lemma}
\begin{proof}
Let $\pi^i_I = \pi|_{\cA_{c_i}(I)}$, where we have identified $\cA_{c_i}(I)$ with its tensor factor in $\cA_{c_1}(I) \otimes \cA_{c_2}(I)$.
Then $\pi^i$ is a representation of $\cA_{c_i}$, and so there are unique strongly continuous representations $U^i$ of $\Diff^{(\infty)}_c(S^1)$ on $\cH_\pi$ such that $U^i(\gamma) = \pi_I(U^0(\gamma))$ for all $\gamma \in \Diff_c(I)$, and that representation has positive energy (Theorem \ref{thm: reps are diff covariant}).
Let $\cA_i = \bigvee_I \pi^i_I(\cA_{c_i}(I))$.
If $r_{2\pi}$ is a lift of $2\pi$ rotation to $\Diff^{(\infty)}_c(S^1)$, then $U^i(r_{2\pi})$ commutes with $U^i(\Diff^{(\infty)}_c(S^1))$, and thus with $\cA_i$.
On the other hand $U^1(r_{2\pi})$ lies in $\cA_1$, and therefore commutes with $\cA_2$, and similarly $U^2(r_{2\pi})$.
Thus $U^i(r_{2\pi})$ commutes with $\cA_1 \vee \cA_2$, which is $\cB(\cH_\pi)$ by the irreducibility of $\pi$.
Hence each $U^i(r_{2\pi})$ is a scalar, so by \cite[Prop. 2.2]{Carpi04} $\pi^i$ is a direct sum of irreducible sectors.
It follows that each $\cA_i$ is a type I factor.
Now applying the argument of \cite[Lem. 27]{KaLoMu01}, we have that $\pi$ is unitarily equivalent to a tensor product of irreducible representations, which by \cite[Prop. 2.1]{Carpi04} are of the desired form.
\end{proof}

\begin{Lemma}\label{lem: Vir modules tensor split}
Let $V_1 = \Vir_{c_1}$ and $V_2 = \Vir_{c_2}$ for unitary values of $c_i$ and assume that $V_i$ have bounded localized vertex operators.
Let $M_i$ be  unitary $V_i$-modules, let $M = M_1 \otimes M_2$, and suppose that $\pi^{M}$ exists.
Then for all intervals $I$
$$
\pi_I^{M}(\cA_{V_1}(I) \otimes 1) \subset \cB(\cH_{M_1}) \otimes 1 
\quad \mbox{ and } \quad
\pi_I^{M}(1 \otimes \cA_{V_2}(I)) \subset 1 \otimes \cB(\cH_{M_2}).
$$
\end{Lemma}
\begin{proof}
By Proposition \ref{prop: reps from sums of modules}, it suffices to prove the lemma in the case where the $M_i$ are irreducible, so we assume $M_i = L(c_i,h_i)$.
For brevity, we will write $\cH_i$ instead of $\cH_{M_i}$.
By Lemma \ref{lem: Virasoro is right}, $\cA_{V_i} = \cA_{c_i}$.
It follows by Proposition \ref{prop: reps from sums of modules} that $\pi^M$ is irreducible, so by Lemma \ref{lem: Virasoro net splitting} we have $\pi^{M} \cong \pi_{(c_1,h_1^\prime)} \otimes \pi_{(c_2,h_2^\prime)}$.
Let $M_i^\prime = L(c_i,h_i^\prime)$ and $\cH_i^\prime = \cH_{M_i^\prime}$, and let $u: \cH_{1} \otimes \cH_2 \to \cH_{1}^\prime \otimes \cH_{2}^\prime$ be an isomorphism of $\cA_{c_1} \otimes \cA_{c_2}$ representations.
The desired result will follow if we show that $u$ may be factored $u = u_1 \otimes u_2$, for $u_i: \cH_{i} \to \cH_{i}^\prime$.

By Lemma \ref{prop: unitary equivalence of piM}, we have that $u$ is an equivalence of representations $U^M$ and $U^{h_1^\prime} \otimes U^{h_2^\prime}$ of $\Diff_{c_1+c_2}^{(\infty)}(S^1)$.
The lowest weight spaces of $M_1 \otimes M_2$ and $M_1^\prime \otimes M_2^\prime$ are one-dimensional.
Thus if $w_i \in M_i$ are lowest weight unit vectors, we may choose unit vectors $w_i^\prime \in M_i^\prime$ such that $u(w_1 \otimes w_2) = w_1^\prime \otimes w_2^\prime$.

Our strategy is now as follows.
Suppose that we have found non-zero $L_0$-graded subspaces $K_i \subset M_i$ such that $u|_{K_1 \otimes K_2}$ splits into a tensor product $u_1 \otimes u_2$ (e.g. $K_i = \Span w_i$).
We will show that the same is true when $K_i$ is replaced by 
$$
K_i^+:=\Span \{ b_{(n)}^{M_i} a \,:\, n \in \bbZ, \, b \in V, \, a \in K_i\}.
$$
By the irreducibility of $M_i$, we will then be done.

So suppose that we have $K_i$ as above, and let $u_i:K_i \to \cH_i^\prime$ be isometries such that $u|_{K_1 \otimes K_2} = u_1 \otimes u_2$.
Let $a_i \in K_i$ be non-zero homogeneous vectors, and let $a_i^\prime = u_i a_i$.
Since $u$ intertwines the diagonal action of $\Diff_{c_1+c_2}^{(\infty)}(S^1)$, it follows that the $a_i^\prime$ are also homogeneous vectors.

Fix an interval $I$ and $(1,A) \in \scA_I$ with $\Int(A) \ne \emptyset$.
By the definition of system of generalized annuli, $\pi^M(A) = A^{M_1} \otimes A^{M_2}$ and similarly for the action on $\cH_1^\prime \otimes \cH_2^\prime$.
Since $\pi^M$ exists, we may find a vector $\xi \in \cH_1 \otimes \cH_2$ such that $(A^{M_1} \otimes A^{M_2})\xi = a_1 \otimes a_2$.
By Lemma \ref{lem: split simple tensors}, if we choose $\xi \in \ker(A^{M_1} \otimes A^{M_2})^\perp$ then $\xi = \xi_1 \otimes \xi_2$ for vectors $\xi_i$ such that $A\xi_i = a_i$.
Since 
$$
u(A^{M_1} \otimes A^{M_2}) = (A^{M_1^\prime} \otimes A^{M_2^\prime})u
$$
we must have 
$$
u\ker(A^{M_1} \otimes A^{M_2})^\perp = \ker (A^{M_1^\prime} \otimes A^{M_2^\prime})^\perp.
$$
In particular, $u(\xi_1 \otimes \xi_2)  \in \ker (A^{M_1^\prime} \otimes A^{M_2^\prime})^\perp$.
But we also have
$$
(A^{M_1^\prime} \otimes A^{M_2^\prime})u(\xi_1 \otimes \xi_2)
=
u(A^{M_1} \otimes A^{M_2})(\xi_1 \otimes \xi_2)
=
u(a_1 \otimes a_2)
=
a_1^\prime \otimes a_2^\prime.
$$
By Lemma \ref{lem: split simple tensors} there are $\xi_i^\prime \in \cH_i^\prime$ such that $u(\xi_1 \otimes \xi_2) = \xi_1^\prime \otimes \xi_2^\prime$, and $A\xi_i^\prime = a_i^\prime$.
Moreover $\xi_i^\prime \in \ker(A^{M_i^\prime})^\perp$.
Let $\tilde A_i$ be the inverse of the bijective map
$$
A^{M_i^\prime}: \ker(A^{M_i^\prime})^\perp \to \Ran(A^{M_i^\prime}).
$$
Then $\xi_i^\prime = \tilde A_i a_i^\prime$.
For any $z \in \Int(A)$ and $b_i \in V$ we have
\begin{align*}
u(Y^{M_1}(b_1,z)a_1 &\otimes Y^{M_2}(b_2,z)a_2)
=
u\,\pi^M_I(Y^{V_1}(b_1,z)A \otimes Y^{V_2}(b_2,z)A)(\xi_1 \otimes \xi_2)\\
& =
\pi_{(c_1,h_1^\prime),I}(Y^{V_1}(b_1,z)A)\xi_1^\prime \otimes \pi_{(c_2,h_2^\prime),I}(Y^{V_2}(b_2,z)A)\xi_2^\prime\\
&= 
\pi_{(c_1,h_1^\prime),I}(Y^{V_1}(b_1,z)A)\tilde A_1 u_1 a_1 \otimes \pi_{(c_2,h_2^\prime),I}(Y^{V_2}(b_2,z)A)\tilde A_2 u_2 a_2.
\end{align*}
Let 
$$
\cK_i^+ = \overline{ \Span \{ Y^{M_i}(b,z)a \, : \, b \in V, \, a \in K_i\}}.
$$
Since $u$ is unitary, the maps $u_i^+:\cK_i^+ \to \cH_i^\prime$ given by 
$$
u_i^+Y^{M_i}(b,z)a = \pi_{(c_i,h_i^\prime),I}(Y^{V_i}(b,z)A)\tilde A_i u_i a
$$
are scalar multiples of isometries.
Setting $b = \Omega$, we see that $u_i^+|_{K_i} = u_i$, and thus $u_i^+$ is an isometry (since $u_i$ was).

As before, set 
$$
K_i^+ := \Span \{b_{(n)}a \, : \,  a \in K_i, \, b \in V, \, n \in \bbZ\}. 
$$
We would like to apply Lemma \ref{lem: components in closure} to conclude that $K_i^+ \subset \cK_i^+$, and so we quickly check the hypotheses.
We first observe that $A \in \Ann^{in}(\cH_i)$ by Lemma \ref{lem: split simple tensors}.
Second, since $z \in \Int(A)$ we have that $Y^{V_i}(s^{L_0}-,z)A$ is bounded for some $s > 0$.
Then since $\pi^M_I$ is implemented by a unitary,
$$
Y^{M_1}(s^{L_0}-,z)A \otimes A = \pi^M_I(Y^{V_1}(s^{L_0}-,z)A \otimes A)
$$
is bounded as a map $\cH_{V_1} \otimes \cH_M \to \cH_M$.
Since $A \ne 0$ it follows that $Y^{M_1}(s^{L_0}-,z)A$ is bounded, and similarly for $Y^{M_2}(s^{L_0}-,z)A$.
Thus we may indeed apply Lemma \ref{lem: components in closure} to conclude that $K_i^+ \subset \cK_i^+$.
By restricting, we have that $u|_{K_1^+ \otimes K_2^+} = u_1^+ \otimes u_2^+$.

We now iterate the above procedure, starting with $K_i = \Span w_i$.
Since the $M_i$ are irreducible, it follows that $u|_{M_1 \otimes M_2} = u_1 \otimes u_2$ for isometries $u_i$.
Hence $u = u_1 \otimes u_2$ for unitaries $u_i$.
It follows that $\pi_M(\cA_{V_1}(I) \otimes 1) \subset \cB(\cH_{1}) \otimes 1$ and similarly that $\pi_M(1 \otimes \cA_{V_2}(I)) \subset 1 \otimes \cB(\cH_{M_2})$, and the proof is complete.
\end{proof}

We are nearly ready to prove Proposition \ref{propTensorProductModuleExists} giving the tensor product splitting of $\pi^{M_1 \otimes M_2}$.
We need one final observation.

\begin{Lemma}\label{lem: piM up to scalar}
Let $V$ be a simple unitary vertex operator superalgebra and let $M$ be a unitary $V$-module.
Suppose that $\pi$ is a representation of $\cA_V$ such that for all intervals $I$, all $(B,A) \in \scA_I$, all $z \in \Int (B,A)$, and all $a \in V$, we have $A \in \Ann^{in}(\cH_M)$, $B \in \Ann^{out}(\cH_M)$, and there exists a scalar $\lambda \in \bbC^\times$ such that
\begin{equation}\label{eqn: piM up to scalar}
\pi_I(BY^V(a,z)A) = \lambda \,\, B Y^M(a,z)A.
\end{equation}
Then $\pi = \pi^M$.
\end{Lemma}
\begin{proof}
Without loss of generality, assume that $M$ is non-zero.
If $a \in V$ and $B Y^M(a,z)A = 0$, then since $B^*,A \in \Ann^{in}(\cH_M)$ we have $\ip{Y^M(a,z)w_1,w_2} = 0$ for all $w_i \in M$, and thus $a^M_{(k)} = 0$ for all $k \in \bbZ$.
Since $V$ is simple and $M$ is non-zero, we must have $a = 0$.
Thus $B Y^M(a,z)A$ is non-zero when $a \ne 0$.
By the same argument, $B Y^V(a,z) A$ is non-zero when $a \ne 0$.

Fix $I,B,A,z$ as above, and for non-zero $a \in A$ let $\lambda_a$ be the scalar such that \eqref{eqn: piM up to scalar} holds.
We will show that $\lambda_a$ is independent of $a$.
Clearly if $a$ and $b$ are non-zero and linearly dependent, $\lambda_a = \lambda_b$.
So let $a,b \in V$ be linearly independent.
Using the fact that $\pi_I(BY^V(v,z)A)$ is linear in $v$, we obtain
$$
\lambda_{a+b} BY^M(a+b,z)A = \lambda_{a} BY^M(a,z)A + \lambda_{b}BY^M(b,z)A
$$
and therefore
$$
(\lambda_{a+b} - \lambda_a)BY^M(a,z)A + (\lambda_{a+b} - \lambda_b)BY^M(b,z)A = 0.
$$
Hence
$$
B Y^M((\lambda_{a+b} - \lambda_a)a + (\lambda_{a+b} - \lambda_b)b,z)A = 0.
$$
Since $a$ and $b$ are linearly independent, we obtain $\lambda_a = \lambda_{a+b} = \lambda_b$.

Finally, to find the constant $\lambda = \lambda_a$, we evaluate with $a = \Omega$, and obtain $\pi_I(BA) = \lambda BA$.
However, the definition of the action $BA$ on $\cH_M$ is $\pi_I(BA)$, and so $\lambda = 1$, as desired.
\end{proof}

We are now ready to prove our tensor product splitting result.

\begin{Proposition}\label{propTensorProductModuleExists}
Let $V_1$ and $V_2$ be simple unitary vertex operator superalgebras with bounded localized vertex operators, and let $M_i$ be unitary $V_i$-modules.
Let $V = V_1 \otimes V_2$ and $M = M_1 \otimes M_2$.
Then $\pi^{M}$ exists if and only if $\pi^{M_1}$ and $\pi^{M_2}$ do, in which case $\pi^{M} = \pi^{M_1} \otimes \pi^{M_2}$.
\end{Proposition}
\begin{proof}
Let $(B,A) \in \scA_{I}$.
Then $A^M$ decomposes as $A^M = A^{M_1} \otimes A^{M_2}$ with $A^{M_i}$ canonical up to scalar, and similarly $B^M = B^{M_1} \otimes B^{M_2}$ (see Convention \ref{conv: annuli}).
Observe that by Lemma \ref{lem: split simple tensors}, $A^M \in \Ann^{in}(\cH_M)$ if and only if both $A^{M_i} \in \Ann^{in}(\cH_{M_i})$, and similarly for the $B$'s.

First assume that both $\pi^{M_i}$ exist, and we will show that $\pi^M = \pi^{M_1} \otimes \pi^{M_2}$.
Let $z \in \Int(B,A)$.
Then for $a_i \in V_i$ we have
\begin{align*}
(\pi^{M_1} \otimes \pi^{M_2})_I(BY^M(a_1 \otimes a_2,z)A) &=
\big( \pi^{M_1}_I(BY^{V_1}(a_1,z)A)\big) \hotimes \, \big( \pi^{V_2}_I(BY^{M_1}(a_2,z)A)\big)\\
&= \big( BY^{M_1}(a_1,z)A\big) \hotimes \, \big( BY^{M_1}(a_2,z)A\big)\\
&= BY^{M}(a_1 \otimes a_2,z)A.
\end{align*}
Thus $\pi^{M_1} \hotimes \pi^{M_2} = \pi^{M}$, as desired.

We now consider the converse, and assume that $\pi^M$ exists.
We first consider the special case when the $V_i$ are Virasoro VOAs.
Let $(B,A) \in \scA_I$ and let $z \in \Int(B,A)$.
Let $a_i \in V_i$.
Then
\begin{equation}\label{eqn: Vir substep}
\pi^M(BY^{V_1}(a_1,z)A \otimes BY^{V_2}(a_2,z)A) = BY^{M_1}(a_1,z)A \otimes BY^{M_2}(a_2,z)A.
\end{equation}
By Lemma \ref{lem: Vir modules tensor split} we have $\pi^M(\cA_{V_1}(I) \otimes 1) \subset \cB(\cH_{M_1}) \otimes 1$, and similarly for $V_2$.
Thus 
$$
\pi^M(BY^{V_1}(a_1,z)A \otimes 1) = \lambda BY^{M_1}(a_1,z)A \otimes 1
$$
and thus by restriction $\pi^M$ we obtain a representation of $\cA_{V_1}$ such that 
$$
\pi(BY^{V_1}(a_1,z)A) = \lambda BY^{M_1}(a_1,z)A.
$$
By Lemma \ref{lem: piM up to scalar}, we see that $\pi = \pi^{M_1}$, and so $\pi^{M_1}$ exists.
Similarly, we may conclude that $\pi^{M_2}$ exists, and examining \eqref{eqn: Vir substep} we see that $\pi^M = \pi^{M_1} \otimes \pi^{M_2}$.

We now return to the general case, where $V_i$ are not necessarily Virasoro VOAs.
Let $\tilde V_i \subset V_i$ be the Virasoro subVOA, let $\tilde M_2$ be an irreducible $W_2$-submodule of $M_2$, and let $\tilde M = M_1 \otimes \tilde M_2$.
Then $\pi^{\tilde M}$ exists by Lemma \ref{lem: module existence conformal subcase}.
By the preceding case we have 
\begin{equation}\label{eqn: restriction splits}
\pi^{\tilde M}|_{\cA_{\tilde V_1} \otimes \cA_{\tilde V_2}} = (\pi^{M_1} \otimes \pi^{\tilde M_2})|_{\cA_{\tilde V_1} \otimes \cA_{\tilde V_2}}.
\end{equation}
Let $\cA_1 = \bigvee_I \pi^{\tilde M}(\cA_{V_1}(I) \otimes 1)$ and 
$\tilde \cA_2 = \bigvee_I \pi^{\tilde M}(1 \otimes \cA_{\tilde V_2}(I))$.
By \eqref{eqn: restriction splits}, we have $\tilde \cA_2 = 1 \otimes \cB(\cH_{\tilde M_2})$.
Since $\cA_1$ commutes with $\tilde \cA_2$, we have $\cA_1 \subset \cB(\cH_{M_1}) \otimes 1$.
Now we may repeat the argument from the Virasoro case to conclude that $\pi^{\tilde M}|_{\cA_{V_1} \otimes 1} = \pi^{M_1} \otimes 1$, and in particular that $\pi^{M_1}$ exists.

Arguing in a similar fashion, one may show that $\pi^{M_2}$ exists, and from there it is immediate that $\pi^{M_1 \otimes M_2} = \pi^{M_1} \otimes \pi^{M_2}$.
\end{proof}

We can now assemble the results for conformal subalgebras and tensor products to give a short proof of Theorem \ref{thmConstructionOfNetRepresentations}:

\begin{proof}[Proof of Theorem \ref{thmConstructionOfNetRepresentations}]
Recall that $V$ is a simple unitary vertex operator superalgebra with bounded localized vertex operators, $W$ is a unitary subalgebra of $V$, $N$ is a unitary $V$-module such that $\pi^N$ exists, and that $M$ is a $W$-submodule of $V$.
We must show that $\pi^M$ exists.
By Proposition \ref{prop: reps from sums of modules}, it suffices to consider when $M$ is irreducible.

By Proposition \ref{propTensorDecomposition}, $N$ is unitarily equivalent to a direct sum $\bigoplus M_i \otimes K_i$ as a $W \otimes W^c$ module, where some $M_i$ is isomorphic to $M$. 
Each $\pi^{M_i \otimes K_i}$ exists by Proposition \ref{prop: reps from sums of modules}, and thus each $\pi^{M_i}$ exists by Proposition \ref{propTensorProductModuleExists}.
\end{proof}

When $V$ has bounded localized vertex operators and $W$ is a unitary subalgebra, we have a subnet $\cB \subset \cA_V$ with $\cH_B = \cH_W$ and $\cB|_{\cH_W} = \cA_W$ (see Proposition \ref{prop: BLVO subalgebra}).
In the process of proving Theorem \ref{thmConstructionOfNetRepresentations}, we computed explicitly what $\cB$ is, or in other words we compute how $\cA_W$ acts in $\cA_V$ modules $\pi^N$.
When $W$ is not a conformal subalgebra, in a $V$-module $N$ there are two positive energy representations on $\cH_N$, one coming from the Virasoro field $Y^N(\nu_V,x)$ and one from $Y^N(\nu_W,x)$.
If $A \in \scA^{in}$, we therefore have two distinct actions on $\cH_N$, which we denote $A^V$ and $A^W$, and similarly for $B \in \scA^{out}$.

\begin{Corollary}\label{cor: action of subnet}
Let $V$ be a simple unitary vertex operator superalgebra with bounded localized vertex operators, let $W$ be a unitary subalgebra.
Let $N$ be a unitary $V$-module such that $\pi^N$ exists.
Let $(B,A) \in \scA_I$ and let $z \in \Int(B,A)$.
Then for $a \in W$, 
$$
\pi^N_I|_{\cA_W}(BY^W(a,z)A) = B^W Y^N(a,z) A^W.
$$
In particular, the subnet $\cB$ of $\cA_V$ which restricts to $\cA_W$ on $\cH_W$ has local algebras generated by operators $B^W Y^V(a,z) A^W$, with $B$, $A$, and $z$ as above.
\end{Corollary}
\begin{proof}
As before, we consider a conformal inclusion $W \otimes W^c$, and replace $W$ by $W \otimes 1$.
As in the proof of Theorem \ref{thmConstructionOfNetRepresentations}, we may decompose $N = \bigoplus M_i \otimes K_i$.
It follows that $A^W = \bigoplus A \otimes 1$, and similarly for $B$.
Moreover, for $a \in W$ we have $Y^N(a,z) = \bigoplus Y^{M_i}(a,z) \otimes 1$, so $B^WY^N(a,z)A^W = \bigoplus BY^{M_i}(a,z)A \otimes 1$.
On the other hand, by Proposition \ref{prop: reps from sums of modules} and Proposition \ref{propTensorProductModuleExists}, we have 
$$
\pi^N(BY^W(a,z)A \otimes 1) = \bigoplus B Y^{M_i}(a,z)A \otimes 1.
$$
and so the desired conclusion follows.
The formula for $\cB$ is immediate from applying the conclusion to $N=V$.
\end{proof}

\subsection{Localized intertwining operators}

Let $V$ be a unitary vertex operator superalgebra with bounded localized vertex operators, let $M$, $N$ and $K$ be unitary $V$-modules and assume that $\pi^N$ and $\pi^K$ exist.
In particular, if $(B,A) \in \scA_I$, then $A \in \Ann^{in}(\cH_N)$ and $B \in \Ann^{out}(\cH_K)$.
The expression $B\cY(a,z)A$ is therefore a densely defined sesquilinear form, subject to certain ambiguities.
First, one must choose a value of $\log z$ to account for fractional powers in the series $\cY(a,z)$.
Next, recall that the actions of $A$ and $B$ on $\cH_N$ and $\cH_K$ are only well-defined up to a factor of $e^{2\pi i L_0}$.
Even if one chooses values for these actions which are compatible (as in the beginning of Section \ref{sec: def of module}), the expression $B \cY(a,z)A$ is only  canonically defined if the $L_0$-eigenvalues in $N$ and $K$ differ by integers.
However, many properties of this operator will be independent of these choices.

Recall that if $\cA$ is a von Neumann superalgebra, with representations $\pi_1$ and $\pi_2$ on super Hilbert spaces $\cH_i$, then the graded intertwiners are given by
$$
\Hom_{\cA}(\cH_1,\cH_2) = \{ x \in \cB(\cH_1,\cH_2) : x \pi_1(y) = (-1)^{p(x)p(y)} \pi_2(y) x \mbox{ for all } y \in \cA \}.
$$
As usual, the expression $x \pi_1(y) = (-1)^{p(x)p(y)} \pi_2(y)x$ should be extended linearly for non-homogeneous $x$ and $y$.

\begin{Definition}
Let $V$ be a simple unitary vertex operator algebra, let $M$, $N$, and $K$ be unitary $V$-modules, and assume that $\pi^N$ and $\pi^K$ exist.
An intertwining operator $\cY \in I \binom{K}{M N}$ is said to be \emph{bounded} and \emph{localized} if for every interval $I$, every $(B,A) \in \scA_I$, and every $z \in \Int(B,A)$, we have $z \in \Int_{\cY}(B,A)$ and
$$
B \cY(a,z) A \in \Hom_{\cA_V(I^\prime)}(\cH_N, \cH_K)
$$
for all $a \in M$.
We write $I_{loc} \binom{K}{M N}$ for the space of bounded localized intertwining operators.
\end{Definition}

Observe that the condition $B \cY(a,z) A \in \Hom_{\cA_V(I^\prime)}(\cH_N, \cH_K)$ is independent of the required choices of $\log z$ and of the action of $A$ and $B$, as if one makes two different choices, the resulting operators differ by a unitary which is scalar on irreducible subrepresentations of $\cH_N$ and $\cH_K$.

Elements of $\Hom_{\cA_V(I)}(\cH_{\pi_1}, \cH_{\pi_2})$ are called local intertwiners for the representations $\pi_1$ and $\pi_2$.
The construction of large subspaces of local intertwiners which can be understood in terms of vertex operators is a key step in understanding the fusion of the corresponding representations.
In Wassermann's landmark analysis of the $SU(N)_k$ models \cite{Wa98}, he took advantage of the fact that certain smeared intertwining operators were bounded (at least for $SU(N)$-primary fields), and therefore produced enough local intertwiners.
The same strategy was later successfully employed by Toledano-Laredo \cite{TL97} and Loke \cite{Loke} for type $D$ WZW models and unitary minimal models, respectively, although Loke's analysis has certain unrelated technical gaps.
However, in general there will not be enough intertwining operators which become bounded when smeared, even for WZW models, and new ideas are required.

One approach, suggested by Wassermann, is to replace the smeared intertwining operators with the partial isometry of their polar decomposition.
This strategy has been successfully employed by Bin Gui \cite{GuiUnitarityI,GuiUnitarityII,GuiG2} to analyze WZW conformal nets with gauge group $B$, $C$, and $G_2$, and to construct a unitary structure on the representation category of the WZW VOAs of the same type.
However, at present the polar decomposition strategy crucially relies on certain estimates, called linear energy bounds, which are not expected to hold in general outside of WZW models.
A significant feature of our present approach is the construction of local intertwiners for representations $\pi^M$ in a manner which only uses bounded operators, and is not limited to nice classes of VOAs (such as WZW or more generally rational VOAs).
In a sequel article, we will exploit this to study the fusion of representations of conformal nets.

\begin{Theorem}\label{thm: intertwiners local}
Let $V$ be a simple unitary vertex operator superalgebra with bounded localized vertex operators, and let $\tilde M$ be a unitary $V$-module such that $\pi^{\tilde M}$ exists.
Let $W \subset V$ be a unitary subalgebra, let $N$ and $K$ be simple $W$-submodules of $\tilde M$, and let $M$ be a simple $W$-submodule of $V$.
Let $p_N$ and $p_K$ be the orthogonal projections of $\cH_{\tilde M}$ onto $\cH_{N}$ and $\cH_K$, respectively, and let $\cY \in I \binom{K}{M N}$ be given by
$$
\cY(a,x) = x^{\Delta} p_{K} Y^{\tilde M}(a,x) p_N
$$
for an appropriate $\Delta \in \bbR$.
Then $\cY \in I_{loc} \binom{K}{M N}$.
\end{Theorem}

We give the proof of Theorem \ref{thm: intertwiners local} after Lemma \ref{lem: relate annuli AV and AW}, which isolates a necessary technical observation.

\begin{Lemma}\label{lem: relate annuli AV and AW}
Let $V$ be a simple unitary vertex operator superalgebra with bounded localized vertex operators, and let $\tilde M$ be a unitary $V$-module such that $\pi^{\tilde M}$ exists.
Let $W \subset V$ be a unitary subalgebra, let $N$ be an irreducible $W$-submodule of $\tilde M$, and let $p_N$ be the orthogonal projection of $\cH_{\tilde M}$ onto $\cH_N$.
Let $A \in \scA^{in}$, and denote by $A^V$ and $A^W$ the two distinct actions of $A$ on $\cH_{\tilde M}$ coming from the Virasoro fields $Y^{\tilde M}(\nu_V,x)$ and $Y^{\tilde M}(\nu_W,x)$, respectively.
Then there exists a $q \in \cB(\cH_{\tilde M})$ such that $q$ commutes with the action of $\cA_W(I)$ on $\cH_{\tilde M}$ for all $I$, and 
$$
p_M A^W = A^W p_M = A^Vq
$$
\end{Lemma}
\begin{proof}
If $W$ is a conformal subalgebra then $A^W = A^V$ and we may take $q = p_M$.
Now consider if $W$ is not a conformal subalgebra,  in which case $V$ has a subalgebra isomorphic to $W \otimes W^c$ and we may decompose $\tilde M = \bigoplus M_i \otimes K_i$ as $W \otimes W^c$ modules, with $M_i$ irreducible and pairwise non-isomorphic. 
Under this decomposition, $A^V = \bigoplus A \otimes A$ and $A^W = \bigoplus A \otimes 1$.
Then for a certain $i$, we have $N = M_i \otimes \{v\}$, for a non-zero homogeneous $v \in K_i$ (recall that $W$-submodules are assumed $L^V_0$-invariant, by definition).
By Theorem \ref{thmConstructionOfNetRepresentations}, $\pi^{K_i}$ exists as a representation of $\cA_{W^c}$, and in particular $A \in \Ann^{in}(K_i)$.
Thus we may choose $\xi \in K_i$ such that $A\xi = v$.
Let $q$ be the operator supported on $\cH_M = \cH_{M_i} \otimes \{v\}$ which acts by $q(\eta \otimes v) = \eta \otimes \xi$.
By construction, $p_M A^W p_M = A^V q$.
Moreover, $q \in \bigoplus 1 \otimes \cB(\cH_{K_i})$, and by Corollary \ref{cor: action of subnet} the action of $\cA_W$ on $\cH_{\tilde M}$ is contained in $\bigoplus \cB(\cH_{M_i}) \otimes 1$.
Hence $q$ commutes with $\cA_W$, as desired.
\end{proof}

\begin{proof}[Proof of Theorem \ref{thm: intertwiners local}]
Recall that $\cY \in I \binom{K}{M N}$ by Proposition \ref{propIntertwinerDescent}.
Let $(B,A) \in \scA_I$ and $z \in \Int(B,A)$, and pick some choice of $\log z$.
Let $A^V$ and $A^W$ be the two different actions of $A$ on $\cH_{\tilde M}$, as in Lemma \ref{lem: relate annuli AV and AW}, and similarly for $B^V$ and $B^W$.
Let $q_N$ be the operators obtained by applying Lemma \ref{lem: relate annuli AV and AW} to $N$, so that $q_N$ commutes with $\cA_W$ on $\cH_{\tilde M}$ and $A^W p_N = A^V q_N$.
Similarly, let $q_K$ be the operator obtained by applying the same lemma to $B^*$, so that $q_K B^V = p_K A^W$.
Since $p_K$ and $p_N$ commute with $A^W$, we have for $a \in M$
\begin{equation}\label{eqn: non conformal compression}
q_K B^V Y^{\tilde M}(a,z) A^V q_N = z^{-\Delta} p_K B^W \cY(a,z) A^W p_N.
\end{equation}
By definition, $B^V Y^{\tilde M}(a,z) A^V = \pi^{\tilde M}_I( BY^V(a,z)A)$, and therefore $B^V Y^{\tilde M}(a,z) A^V$ supercommutes with $\pi^{\tilde M}_{I^\prime}(\cA_V(I^\prime))$.
In particular, it supercommutes with $\pi^{\tilde M}_{I^\prime}(\cA_W(I^\prime))$.
We chose $q_K$ and $q_N$ so that they commute with these algebras, so it follows from \eqref{eqn: non conformal compression} that 
$p_K B^W \cY(a,z) A^W p_N$
supercommutes with $\pi^{\tilde M}_{I^\prime}(\tilde BY^W(\tilde a,\tilde z)\tilde A)$ for all $\tilde a \in W$ and $(\tilde B, \tilde A) \in \scA_{I^\prime}$ and $\tilde z \in \Int(B,A)$%
\footnote{%
To be more precise, we should write $\pi^{\tilde M}_{I^\prime}(\pi^0_{I^\prime}(\tilde BY^W(\tilde a,\tilde z)\tilde A))$, where $\pi^0_{I^\prime}:\cA_W(I^\prime) \to \cB(\cH_V)$ is the representation of Proposition \ref{prop: BLVO subalgebra} exhibiting $\cA_W$ as a subnet of $\cA_V$.%
}%
.
By Corollary \ref{cor: action of subnet}, 
$$
\pi^{\tilde M}_{I^\prime}(\tilde BY^W(\tilde a,\tilde z)\tilde A) = \tilde B^W Y^{\tilde M}(\tilde a,\tilde z)\tilde A^W,
$$
and thus
\begin{align}
\big(\tilde B^W Y^{\tilde M}(\tilde a,\tilde z)\tilde A^W\big) &\big(p_K B^W \cY(a,z) A^W p_N\big) =\nonumber\\
=\label{eqn: commuting on tildeM}
(-1)^{p(a)p(\tilde a)} \big(p_K B^W \cY(a,z) A^W p_N\big)&\big(\tilde B^W Y^{\tilde M}(\tilde a,\tilde z)\tilde A^W\big).
\end{align}
By Proposition \ref{prop: reps from sums of modules}, since $\tilde A^W$ and $\tilde B^W$ commute with $p_M$ and $p_N$, and we have
$$
p_M \tilde B^W Y^{\tilde M}(\tilde a,\tilde z)\tilde A^W = p_M \tilde B^W Y^{\tilde M}(\tilde a,\tilde z)\tilde A^W p_M = \tilde B^W Y^{M}(\tilde a,\tilde z)\tilde A^W p_M,
$$
and similarly with $p_N$.
Plugging this into \eqref{eqn: commuting on tildeM} we obtain
$$
\big(\tilde B Y^{K}(\tilde a,\tilde z)\tilde A\big) \big(B \cY(a,z) A \big)
=
(-1)^{p(a)p(\tilde a)} \big(B \cY(a,z) A \big) \big(\tilde B Y^{N}(\tilde a,\tilde z)\tilde A\big).
$$
Setting $y = \tilde B Y^{W}(\tilde a,\tilde z)\tilde A$, we have shown that 
$$
\pi^K(y)(B \cY(a,z) A) = (-1)^{p(y)p(a)} (B \cY(a,z) A)\pi^N(y).
$$
Since elements of the form $y$ generate $\cA_W(I^\prime)$, we have $B \cY(a,z) A \in \Hom_{\cA_W(I^\prime)}(\cH_N,\cH_K)$, as desired.
\end{proof}

\newpage 

\section{Bounded localized vertex operators for code extensions}
\label{sec: code extensions}

The goal of this section is to greatly expand the class of vertex operator algebras known to have bounded localized vertex operators.
Historically, it has been fairly routine to infer analytic properties of subalgebras from an ambient algebra.
In this section, however, we show how to extend analytic properties of vertex operator algebras (and their intertwining operators) to certain simple current extensions.
As an application, in Section \ref{sec: examples} we will show that all WZW VOAs have bounded localized vertex operators.
The primary challenge is, as always, to relate locality in the sense of conformal nets and locality in the sense of vertex operator algebras.
The main results are Proposition \ref{prop: BLVO for easier codes} and Theorem \ref{thm: BLVO for harder codes} in Section \ref{sec: BLVO for codes}.

\subsection{Simple currents extensions of VOAs}

In this section we briefly recall the basics of simple current extensions of rational VOAs, and the interested reader may consult \cite[\S4]{YamauchiThesis} for more detail.

Let $W$ be a simple rational vertex operator algebra, with representatives of isomorphism classes of irreducible modules $W=M_0, \ldots, M_n$.
By definition, $M$ is a simple current if for all $0 \le i \le n$ there exists a unique $k$ such that $\dim I\binom{M_j}{M M_i} = \delta_{j,k}$.
Simple currents are necessarily irreducible.

A vertex operator superalgebra $V$ is a simple current extension of $W$ (by a finite abelian group $G$) if it decomposes $V = \bigoplus_{\alpha \in G} M^\alpha$ with $M^0 = W$, the $M^\alpha$ pairwise non-isomorphic simple currents, and $\{0\} \ne M^\alpha \cdot M^\beta \subset M^{\alpha + \beta}$ for all $\alpha,\beta \in G$, where $M^\alpha \cdot M^\beta = \Span \{a_{(k)}b : a \in M^\alpha, b \in M^\beta, k \in \bbZ\}$.
A simple current extension of a simple rational VOA is automatically simple.

\begin{Proposition}\label{prop: unique simple current VOA}
Let $W$ be a simple rational VOA, and let $V$ and $\tilde V$ be two vertex operator algebras which are simple current extensions of $W$.
If $V$ and $\tilde V$ are isomorphic as $W$-modules, then they are isomorphic as VOAs.
\end{Proposition}

A proof of this result is given in \cite[Prop. 4.2.3]{YamauchiThesis} (see also \cite[Prop. 5.3]{DongMason04}).

\begin{Remark}\label{rmk: unique simple current net}
The key ingredient in the proof of Proposition \ref{prop: unique simple current VOA} is that a symmetric cocycle in $Z^2(G,\bbC^\times)$ for $G$ abelian is automatically a coboundary (see \cite[Prop. 5.3]{Karpilovsky2}).
This argument is not special to vertex operator algebras, and readily generalizes to 
an abstract algebraic/categorical argument about $G$-graded algebras with simple components.
In particular, one may apply the same argument to conformal nets (and the $Q$-systems governing local extensions) to obtain uniqueness of simple current extensions in this case as well.
We will only need to discuss simple current extensions of conformal nets in the very special case of lattices (Proposition \ref{prop: SCE agrees for lattices}), and so we do not expand on this idea in more detail.
\end{Remark}

\subsection{Self-dual simple currents and code extensions}\label{sec: self dual simple currents and code extensions}

Fix a simple rational VOA $V$.
We will assume that $V$ is unitary as well, as this is will be the case in our applications, but none of the content of Section \ref{sec: self dual simple currents and code extensions} depends on that in an essential way.
Our results of this section will primarily be concerned with self-dual simple currents with certain self-braiding properties, which we formalize below.
In the following, if $\cY \in I \binom{K}{M N}$, we refer to $N$ as the input (space) of $\cY$, $K$ as the output, and $M$ as as the charge.

\begin{Definition}
Let $V$ be a unitary vertex operator algebra, let $M$, $N$, and $K$ be unitary $V$-modules.
Let $\cY_1, \cY_0, \tilde \cY_1, \tilde \cY_0$ be intertwining operators with only integral powers of the formal variable, and such that the output space of $\cY_1$ is the same as that of $\tilde \cY_1$, the input space of $\cY_0$ is the same as $\tilde \cY_0$, and for both $i$ the charge space of $\cY_i$ is the same as $\tilde \cY_{1-i}$.
Then we say that $\cY_1 \cdot \cY_0$ braid to $\tilde \cY_1 \cdot \tilde \cY_0$ if for every $a_1, a_2$ in the appropriate charge spaces, and every $b_1,b_2$ in the appropriate input/output spaces, the double series 
$$
\ip{\cY_1(a_1,z)\cY_0(a_2,w)b_1,b_2} \qquad \mbox{ and } \ip{\tilde \cY_{1}(a_2,w)\tilde \cY_0(a_1,z)b_1,b_2}
$$
converge absolutely on the domains $\abs{z} > \abs{w}$ and $\abs{w} > \abs{z}$, respectively, to rational functions which extend to the same element of $\bbC[z^{\pm 1}, w^{\pm 1}, (z-w)^{-1}]$.
We denote this relation by $\cY_1 \cdot \cY_0 \sim \tilde \cY_1 \cdot \tilde \cY_0$.
\end{Definition}

\begin{Definition}
A simple current $M$ for $V$ is called \emph{self-dual} if $\dim I \binom{V}{M M} = 1$.
\end{Definition}

Fix a choice of self-dual simple current $M_1$, and let $M_0 = V$.
Let $\bbF_2=\{0,1\}$ be the field with two elements, and for $i \in \bbF_2^n$ let $M_i = M_{i(1)} \otimes \cdots \otimes M_{i(n)}$.
We write $\vertex{k}{i}{j}$ for $I \vertex{M_k}{M_i}{M_j}$, and observe that $\dim  \vertex{k}{i}{j} = \delta_{i+j,k}$.
We say that $\vertex{k}{i}{j}$ is \emph{admissible} if $i + j = k$.

The spaces $\vertex{0}{0}{0}$ and $\vertex{1}{0}{1}$ have distinguished basis vectors, $Y^V$ and $Y^M$, and suppose that we fix some choice of basis vectors for $\binom{1}{1 0}$ and $\binom{0}{1 1}$.
Then for all admissible $\vertex{k}{i}{j}$ we have a distinguished basis vector $\cY_{i j}^{k} \in \vertex{k}{i}{j}$ given by the tensor product of our basis vectors.

In the following, we fix a simple unitary rational vertex operator algebra $V$ and a unitary simple current $M$ which is \emph{self-dual}, meaning that $\dim I\binom{V}{M M} = 1$.

\begin{Definition}\label{def: fermionic and bosonic braiding}
A self-dual simple current $M$ is called \emph{bosonic} (resp. \emph{fermionic}) if the conformal weights of $M$ lies in $\bbZ$ (resp. $\tfrac12 + \bbZ$) and for all $i,k \in \bbF_2$ with $k = i+1$ we have
\begin{equation}\label{eqn: bosonic and fermionic braiding}
\cY_{1k}^i \cdot \cY_{1i}^k \,\, \sim  \,\, \varepsilon \, \cY_{1i}^k \cdot \cY_{1k}^i
\end{equation}
where $\varepsilon = 1$ (resp. $\varepsilon = -1$), for some (equivalently, any) choice of bases for $\vertex{1}{1}{0}$ and $\vertex{0}{1}{1}$.
\end{Definition}

Observe that the powers of $x$ in the intertwining operators which arise in Definition \ref{def: fermionic and bosonic braiding} are integral because of our restriction on the conformal weights (see \cite[Rem. 5.4.4]{FHL93}).
We will be particularly interested in simple currents which are not bosonic or fermionic, but \emph{semionic}, which satisfy an analog of \eqref{eqn: bosonic and fermionic braiding} with $\varepsilon = i$ or $\varepsilon = -i$.
However, in this case the powers of $x$ which arise are not integral, but at the same time it is essential that both values of $\varepsilon$ in \eqref{eqn: bosonic and fermionic braiding} should in the appropriate sense correspond to the same fourth root of unity.
This is somewhat subtle to formalize, as all of the functions involved are multi-valued, and the different branches differ by signs.
We could do this via careful statements about analytic continuation, but instead we opt for a definition which only involves single-valued functions.

\begin{Definition}
A self-dual simple current $M$ is called \emph{semionic} if it has conformal weights lying in $\pm\tfrac14 + \bbZ$ and for $j \in \bbF_2^2$, $i=(1,1)$, and $k=i+j$ we have the fermionic braidings
$$
\cY_{i k}^j \cdot \cY_{i j}^k \,\, \sim \,\, - \, \cY_{i k}^j \cdot \cY_{i j}^k
$$
for some (equivalently, any) choice of bases for $\vertex{1}{1}{0}$ and $\vertex{0}{1}{1}$.
\end{Definition}
By construction, if $M$ is semionic then $M \otimes M$ is fermionic, and if $M$ is fermionic then $M \otimes M$ is bosonic.

For $\vertex{k}{i}{j}$ admissible and $M$ bosonic or fermionic, the intertwining operator $\cY_{i j}^k$ always has integral powers of $x$.
On the other hand, if $M$ is semionic then $\cY_{i j}^{k}$ has integral powers of $x$ if and only if $i \cdot j \in 2\bbZ$, where the dot product $i \cdot j$ is obtained by embedding $\bbF_2^n \subset \bbZ^n$ in the obvious way.
Our definition of semionic allows us to easily compute braid statistics for such intertwining operators.

\begin{Lemma}\label{lem: codeword braid statistics}
Let $M$ be a self-dual simple current, and let $i,j,k \in \bbF_2^n$.
\begin{enumerate}
\item If $M$ is bosonic or fermionic, then
$$
\cY^k_{q j} \cdot \cY^j_{p i} \,\,\sim\,\, (\pm 1)^{p \cdot q} \cY^k_{p \ell} \cdot \cY^\ell_{q i},
$$
where $p,q,k,$ and $\ell$ are determined by the assumption that all triples are admissible, and the sign $+1$ is taken if $M$ is bosonic, and $-1$ if $M$ is fermionic.

\item If $M$ is semionic and $p \cdot \ell$, $q \cdot i$, $p \cdot q \in 2 \bbZ$ then
$$
\cY^k_{q j} \cdot \cY^j_{p i} \,\,\sim\,\, (-1)^{\frac12(p \cdot q)} \, \cY^k_{p \ell} \cdot \cY^\ell_{q i} .
$$
\end{enumerate}
\end{Lemma}
\begin{proof}
We first consider the bosonic and fermionic case.
From the definition of an intertwining operator, if $r,s,t \in \bbF_2$ and $\vertex{t}{r}{s}$ is admissible, then we have (as in \cite[\S3.2]{FHL93})
$$
\cY^t_{0 t} \cdot \cY_{rs}^t \,\, \sim \,\, \cY_{rs}^t \cdot \cY^s_{0 s}.
$$
Thus when $M$ is bosonic, we can apply the braiding on each tensor factor to obtain the desired expression.
When $M$ is fermionic, observe that $p \cdot q$ is precisely the number of components for which $p$ and $q$ are both 1, and so applying the braid relations on each tensor factor we get $p \cdot q$ factors of $-1$.

Now consider the semionic case.
The assumption that $p \cdot \ell$ and $q \cdot i$ are even ensure that the intertwining operators have integral powers of $x$.
Since $p \cdot q$ is even, we can partition the $n$ components of $\bbF_2^n$ into (i) ones for which either $p$ and/or $q$ have a $0$ entry, or (ii) pairs on which $p$ and $q$ both restrict to $1$. 
It thus suffices to establish the desired braiding for each of these cases.
On pairs of type (ii), we can simply apply the definition of semionic simple current to obtain a braiding factor of $(-1)^{p \cdot q}$.
On the other hand, on tensor factors of type (i) we can use the commutativity of intertwining operators with the module operator to get a braiding factor of $1$, which completes the proof.
\end{proof}

It is easy to obtain bosonic and fermionic simple currents in the branching of conformal inclusions.

\begin{Lemma}\label{lem: easy fermionic and bosonic simple currents}
Let $V$ be a simple unitary vertex operator superalgebra, and let $W$ be an even unitary rational conformal subalgebra.
Let $M$ be a $W$-submodule of $V$ which is a self-dual simple current.
If $M \subset V^0$ then $M$ is bosonic, and if $M \subset V^1$ then $M$ is fermionic.
\end{Lemma}
\begin{proof}
Let $p_1$ and $p_0$ be the projections of $V$ onto $M$ and $W$, respectively.
Observe that if $a \in W$ then we have $p_i Y^V(a,x) = Y^V(a,x)p_i$, and if $a \in M$ then $p_i Y^V(a,x) = Y^V(a,x) p_{1-i}$.
The desired braiding now follows from the commutativity of products for $Y^V$.
\end{proof}

We will only use one semionic simple current in our examples:

\begin{Lemma}\label{lem: A1 module is semionic}
Let $V$ be the unitary WZW model of type $A_1$ at level $1$, and let $M$ be its non-trivial irreducible module.
Then $M$ is a semionic self-dual simple current.
\end{Lemma}
\begin{proof}
By the Frenkel-Kac-Segal construction (\cite[\S5.4]{Kac98}, \cite{FrenkelKac80}), $V$ is the lattice VOA associated to the $A_1$ lattice.
We embed 
$$
\Lambda_0:=A_1 \times A_1 \subset \bbZ^2 \cup \big((\tfrac12,\tfrac12) + \bbZ^2\big)=:\Lambda
$$ 
by sending the generator of the first copy of $A_1$ to $(1,1)$ and the generator of the second copy to $(1,-1)$.
We have $\Lambda/\Lambda_0 \cong \bbF_2^2$, and so by \cite[Thm. 5.2]{DongLepowsky94} we have the structure of an abelian intertwining algebra on $V_\Lambda := \bigoplus_{i \in \bbF_2^2} M_i$.
The indicated braiding of intertwining operators is \cite[Prop. 3.4]{DongLepowsky94}.
\end{proof}

Recall that a (binary, linear) code of length $n$ is a subspace $C \subset \bbF_2^n$.
For codewords $i,j \in C$, we denote by $i \cdot j$ the dot product in $\bbZ^n$.
A code $C$ is called \emph{even} if $i \cdot i \in 2\bbZ$ for all $i \in C$, and doubly even if the same holds with $4 \bbZ$ in place of $2 \bbZ$. 
It is called \emph{self-orthogonal} (written $C \subset C^\ast$) if $i \cdot j \in 2 \bbZ$ for all $i,j \in C$.
Observe that self-orthogonal codes are automatically even, but even (and doubly even) codes may fail to be self-orthogonal.

\begin{Definition}
Let $V$ be a simple unitary vertex operator algebra, and let $W^{\otimes n}$ be a rational unitary conformal subalgebra.
Let $M$ be a self-dual simple current for $W$ which is bosonic, fermionic, or semionic, and let $C$ be a binary code.
If $M$ is fermionic or semionic assume that $C \subset C^\ast$, and moreover if $M$ is semionic assume that $C$ is doubly even.
Then $V$ is called a simple current extension of $W$ of \emph{code type} with \emph{length $n$} with respect to $(C,W,M)$ if $V \cong \bigoplus_{i \in C} M_i$ as a $W^{\otimes n}$-module.
\end{Definition}

Note that the code type extensions considered have integral conformal dimensions by definition, and therefore must be (even) vertex operator algebras.

A typical example is given by code lattices.
If $C$ is a code of length $n$, define the corresponding (untwisted) code lattice:
\begin{equation}\label{eqn: code lattice}
\Lambda_C := \bigcup_{i \in C} \sqrt{2}\bbZ^n + \tfrac{i}{\sqrt{2}}.
\end{equation}
Then $C \subset C^\ast$ if and only if $\Lambda_C$ is integral, and $C$ is doubly even if and only if $\Lambda_C$ is even \cite[Thm. 7.2.2]{ConwaySloane}.
By construction, we have a sublattice $A_1^n = \sqrt{2}\bbZ^n \subset \Lambda_C$ corresponding to the codeword $i=0$, and $\Lambda_C/A_1^n \cong C$.
Moreover ${V_{A_1,1}}^{\otimes n} \subset V_{\Lambda_C}$ is a code type extension with respect to the code $C$ and the non-trivial module of $V_{A_1,1}$.

\begin{Lemma}\label{lem: code extension cocycle}
Let $W^{\otimes n} \subset V$ be a simple current of code type with respect to $(W,M,C)$.
For $i \in C$, let $p_i$ be the orthogonal projection of $V$ onto $M_i$.
Choose basis vectors for $\vertex{1}{1}{0}$ and $\vertex{0}{1}{1}$, and for $i,p \in C$ and $j = p+i$ let $\cY_{pi}^j \in \vertex{j}{p}{i}$ be the basis vector given as a tensor product of the distinguished bases.
Let $c(p,i) \in \bbC^\times$ be the scalars such that $p_{j}Y^V(a,x)p_i = c(p,i)\cY_{pi}^j(a,x)$ for all $a \in M_p$.
Then $c$ satisfies the twisted cocycle condition
$$
c(p,q+i)c(q,i) = \varepsilon^{p \cdot q} c(q,p+i)c(p,i)
$$
where $\varepsilon = 1$ if $M$ is bosonic, $\varepsilon=-1$ if $M$ is fermionic, and $\varepsilon = i$ if $M$ is semionic.
\end{Lemma}
\begin{proof}
Since $V$ is a VOA we have the commutativity relation
$$
p_{k} \, Y^V(a,x)Y^V(b,y) \, p_{i} \,\, \sim \,\, p_{k}\, Y^V(b,y)Y^V(a,x) \, p_{i}
$$
for any $i,k \in C$.
Let $p,q \in C$, let $j = p+i$, $\ell = q+i$ and suppose that $q + p + i = k$.
Then for $a \in M_p$ and $b \in M_q$ we have by Lemma \ref{lem: codeword braid statistics}
\begin{align*}
p_{k} Y^V(a,x)Y^V(b,y) p_{i} &= c(q,j)\,c(p,i) \,\, \cY_{qj}^{k}(a,x)\cY_{pi}^j(b,y)\\
&\sim \varepsilon^{p \cdot q} \, c(q,j) \, c(p,i) \,\, \cY_{p\ell}^k(b,y)\cY_{qi}^{\ell}(a,x)\\
&= \left(\frac{\varepsilon^{p \cdot q} c(q,j)c(p,i)}{c(p,\ell)c(q,i)}\right) p_{k} Y^V(b,y)Y^V(a,x) p_{i}.
\end{align*}
Hence $\varepsilon^{p \cdot q} c(q,j)c(p,i) = c(p,\ell)c(q,i)$, as desired.
\end{proof}

\subsection{Bounded localized vertex operators for code extensions}\label{sec: BLVO for codes}

We now state the main results of the section on bounded localized vertex operators for code extensions.
For bosonic and fermionic type, we have:

\begin{Proposition}\label{prop: BLVO for easier codes}
Let $W$ be a simple rational unitary VOA, and let $M$ be a self-dual simple current of $W$ which is bosonic or fermionic.
Suppose that there exists a simple unitary vertex operator superalgebra $\tilde V$ with bounded localized vertex operators containing $W$ as a conformal subalgebra, and such that $M$ is a $W$-submodule of $\tilde V$.
Then any simple current extensions of $W^{\otimes n}$ of code type with respect to $M$ has bounded localized vertex operators.
\end{Proposition}

And for semionic type, we have:

\begin{Theorem}\label{thm: BLVO for harder codes}
Let $W$ be a simple rational unitary VOA, and let $M$ be a self-dual simple current of $W$ which is semionic.
Suppose that there exists a simple unitary vertex operator superalgebra $\tilde V$  with bounded localized vertex operators containing $W \otimes W$ as a conformal subalgebra, and such that $M \otimes M$ is a $W \otimes W$-submodule of $\tilde V$.
Then any simple current extensions of $W^{\otimes n}$ of code type with respect to $M$ has bounded localized vertex operators.
\end{Theorem}

We only prove Theorem \ref{thm: BLVO for harder codes}, as Proposition \ref{prop: BLVO for easier codes} uses the same idea but is significantly simpler.
We prove some technical lemmas and then given the proof of Theorem \ref{thm: BLVO for harder codes} immediately following the proof of Lemma \ref{lem: ACFT braiding}.

\begin{Lemma}\label{lem: transport unitary}
Let $\cA$ be a Fermi conformal net, let $\cC$ be the even subnet of $\cA$, and let $\cC_0=\cC|_{\cH_\cA^0}$ be the restriction of $\cC$ to its vacuum Hilbert space.
Suppose that $\cC_0$ factors $\cC_0 = \cB \otimes \cB$, and that $\cH_{\cA}$ decomposes as $(\cH_0 \otimes \cH_0) \oplus (\cH_1 \otimes \cH_1)$ as irreducible $\cB \otimes \cB$ sectors, and denote by $\pi$ the action of $\cB$ on $\cH_1$.
Suppose there exists a non-zero
\begin{equation}\label{eqn: off diagonal matrix}
X = \begin{pmatrix} 0 & x \otimes x\\ y \otimes y & 0 \end{pmatrix} \in \cA(I).
\end{equation}
Then there exists unitaries $u \in \Hom_{\cB(I^\prime)}(\cH_0, \cH_1)$ and $v \in \Hom_{\cB(I^\prime)}(\cH_1,\cH_0)$ such that:
\begin{enumerate}
\item for every $a \in \cB(I)$, $\pi_I(u^*\pi_I(a)u) = v^*av$
\item for some $\varepsilon = \pm 1$, we have $\pi_I(v \pi_I(a) u) = \varepsilon uav$ for all $a \in \cB(I)$.
\item for every matrix of operators satisfying \eqref{eqn: off diagonal matrix}, $\pi_I(vx) = \pm uy$ and $\pi_I(yu) = \pm xv$ (for potentially distinct signs which may depend on $x$ and $y$).
\end{enumerate}
\end{Lemma}
\begin{proof}
Throughout the proof we will use the fact that if $A:=\begin{pmatrix}a \otimes a & 0\\ 0 & b\otimes b\end{pmatrix} \in \cA(I)$, then $A \in \cC(I)$ and thus $\pi_I(a) \otimes \pi_I(a) = b \otimes b$, which implies that $\pi_I(a) = \pm b$.
Since $A$ only determines $a$ and $b$ up to a sign, the value of $\pm$ is arbitrary and simply reflects our choice of $a$ and $b$.
The sign in (3) is of a similar nature, while the sign $\varepsilon$ in (2) does not depend on any arbitrary choices, and we expect that it is related to the Frobenius-Schur indicator of $\pi$.

We first show that we can find an operator of the form \eqref{eqn: off diagonal matrix} in $\cA(I)$ with $x$ and $y$ unitary.
Observe that $XX^* = \begin{pmatrix}xx^* \otimes xx^* & 0\\ 0 & yy^* \otimes yy^*\end{pmatrix} \in \cA(I)$, and thus by the preceding observation, if $x$ is unitary then, so is $y$.
Replacing $X$ with the partial isometry in its polar decomposition, we may assume that $x$ is a partial isometry.
Then since the source and target projections $xx^*$ and $x^*x$ lie in $\cB(I)$ and $\cB(I)$ is a type $III$ factor, there are partial isometries $w_1,w_2 \in \cB(I)$ such that $w_1xw_2$ is unitary.
Setting $W_i = \begin{pmatrix} w_i \otimes w_i & 0\\ 0 & \pi_I(w_i) \otimes \pi_I(w_i) \end{pmatrix} \in \cA(I)$, we have $U:=W_1XW_2 \in \cA(I)$, and by construction $U$ is of the form $U = \begin{pmatrix} 0 & u \otimes u\\ v \otimes v & 0\end{pmatrix}$ with $u=w_1xw_2$ unitary.
By the above remarks, $v$ is unitary as well.

Now let $a \in \cB(I^\prime)$, and set $A = \begin{pmatrix} a \otimes a & 0\\ 0 & \pi_{I^\prime}(a) \otimes \pi_{I^\prime}(a) \end{pmatrix} \in \cA(I^\prime)$.
Thus $U$ and $A$ commute, and examining the identity $UA = AU$ we obtain
$$
\pi_{I^\prime}(a)u \otimes \pi_{I^\prime}(a)u = ua \otimes ua.
$$
Thus $\pi_{I^\prime}(a)u = \pm ua$, and since both sides are linear in $a$ the sign $\pm$ must be independent of $a$.
Evaluating at $a=1$ we see that the sign is $+1$, and we conclude $u \in \Hom_{\cB(I^\prime)}(\cH_0,\cH_1)$.

Next, let $a \in \cB(I)$ and let $A$ be as above.
Then $UAU$ and $U^*AU$ are diagonal operators in $\cC(I)$, and examining the diagonal elements we obtain
$$
\pi_I(v\pi_I(a)u) = \pm uav \quad \mbox{ and } \quad \pi_I(u^*\pi_I(a)u) = \pm v^* a v,
$$
respectively.
In both cases, the same linearly argument shows that the sign $\pm$ is independent of $a$, and in the latter case evaluating at $a=1$ reveals that the sign in that case is $+1$.
This establishes (1) and (2).
To establish (3), apply the same argument to $UX \in \cC(I)$ when $X$ is as in \eqref{eqn: off diagonal matrix}.
\end{proof}

Continuing with the setup of Lemma \ref{lem: transport unitary}, we define a map 
$$
\tau_I: \Hom_{\cB(I^\prime)}(\cH_0, \cH_1) \to \Hom_{\cB(I^\prime)}(\cH_1, \cH_0)
$$ 
by $\tau_I(x) = u^* \pi_I(vx)$.
By Haag duality, $vx \in \cB(I)$ and so $\pi_I(vx)$ is defined.
By construction, if $X$ is as in \eqref{eqn: off diagonal matrix} then $\tau_I(x) = \pm y$.

\begin{Lemma}\label{lem: tranport for tau}
Under the setup of Lemma \ref{lem: transport unitary}, 
$$
\pi_I(\tau_I(x_2)x_1) = \varepsilon \, x_2 \tau_I(x_1)
$$
for all $x_1,x_2 \in \Hom_{\cB(I^\prime)}(\cH_0,\cH_1)$.
\end{Lemma}
\begin{proof}
By definition, for $x \in \Hom_{\cB(I^\prime)}(\cH_0, \cH_1)$ we have 
\begin{equation}\label{eqn: left transport}
\pi_I(vx) = u \tau_I(x).
\end{equation}
On the other hand, applying \eqref{eqn: left transport} and (1) of Lemma \ref{lem: transport unitary} we have
\begin{equation}\label{eqn: right transport}
\pi_I(\tau_I(x)u) = 
\pi_I(u^*(u \tau_I(x))u) =
\pi_I(u^* \pi_I(vx) u) = 
v^*(vx)v = 
xv.
\end{equation}
Now combining \eqref{eqn: left transport} and \eqref{eqn: right transport} we have
\begin{equation}\label{eqn: almost done transport}
\pi_I\big(\tau_I(x_2)x_1\big) = 
\pi_I\big(\tau_I(x_2)u\big) \pi_I\big(u^*v^*\big) \pi_I\big(vx_1\big) =
x_2\big(v\pi_I(u^* v^*)u\big) \tau_I(x_1).
\end{equation}
Applying (2) of Lemma \ref{lem: transport unitary} to $a=1$ we see $\pi_I(vu) = \varepsilon uv$, and thus taking adjoints we see that $v \pi_I(u^* v^*) u = \varepsilon$.
Plugging this into \eqref{eqn: almost done transport} completes the proof.
\end{proof}

%
%
%
%
%
%

We are now ready to establish the necessary braiding statistics for operators of the form $B\cY_{ij}^k(a,z)A$ for $i,j,k \in \bbF_2^n$.
Recall that these operators depend on a choice of $\log z$, and also sometimes on choices of sign for $A$ and $B$, as we now explain.
If $(B,A) \in \scA_I$, then as in the discussion at the beginning of Section \ref{sec: def of module} (and in particular \eqref{eqn: unpacking module def}), the unpacked meaning of $B\cY(a,z)A$ is
$$
B\cY(a,z)A = \pi^K_I(BU(\gamma)^*)U^K(\tilde \gamma) \cY(a,z) U^N(\tilde \gamma)^* \pi^M_I(U(\gamma)A)
$$
for an appropriate $\tilde \gamma \in \Diff^{(\infty)}_c(S^1)$.
When the $L_0$-eigenvalues of $K$ and $N$ differ by integers, this expression is independent of the choice of $\tilde \gamma$, but in the case of our intertwiners $\cY_{11}^0$ and $\cY^1_{10}$, the choice of $\tilde \gamma$ introduces a sign $\pm$ of ambiguity to $B\cY^0_{11}(a,z)A$ and $B\cY^1_{10}(a,z)A$, since $L_0$-eigenvalues of $M_0$ and $M_1$ differ by half-integers.
In the lemma below, we will refer to the choice of $\tilde \gamma$ (and thus of this sign for both $\cY_{11}^0$ and $\cY^1_{10}$) as choosing compatible actions of $A$ and $B$.

\begin{Lemma}\label{lem: ACFT braiding}
Let $W$ be a simple rational unitary VOA, and let $M$ be a semionic self-dual simple current of $W$.
Suppose that $\tilde V$ is a vertex operator superalgebra with bounded localized vertex operators such that $W \otimes W \subset \tilde V$, and such that $M \otimes M$ is a $W \otimes W$-submodule of $\tilde V$.
Let $(B,A) \in \scA_I$ and $z \in \Int(B,A)$.
Choose compatible actions of $A$ and $B$ (as in the preceding paragraph).
\begin{enumerate}
\item For all $i \in \bbF_2^n$, $A \in \Ann^{in}(\cH_{M_i})$ and $B \in \Ann^{out}(\cH_{M_i})$.
\item For all admissible $i,j,k \in \bbF_2^n$, $z \in \Int_{\cY}(B,A)$ for $\cY = \cY_{ij}^k$.
\item Let $J$ be an interval disjoint from $I$, and suppose that there exists an interval $K$ containing $I \cup J$.
Let $(\tilde B, \tilde A) \in \scA_J$, let $\tilde z \in \Int(\tilde B, \tilde A)$, and choose compatible actions of $\tilde B$ and $\tilde A$.
Then for a certain choice of $\log z$ and $\log \tilde z$,  every choice of basis vectors for $\vertex{0}{1}{1}$ and $\vertex{1}{1}{0}$, and for all $a \in M_p$ and $\tilde a \in M_q$
\begin{equation}\label{eqn: ACFT braiding}
\big(\tilde B\cY_{qj}^{k}(\tilde a, \tilde z)\tilde A\big) 
\big(B \cY_{pi}^j(a,z) A)
=
(-1)^{\frac12(p \cdot q)}
\big(B \cY_{p\ell}^k(a,z) A)
\big(\tilde B\cY_{qi}^{\ell}(\tilde a, \tilde z)\tilde A\big)
\end{equation}
for all $i,j,k,\ell,p,q \in \bbF_2^n$ for which all triples appearing are admissible and for which $p \cdot q$ is even.
%
\end{enumerate}
\end{Lemma}
\begin{proof}
We first fix choices of basis vectors for intertwiner spaces, as follows.
Since all intertwiner spaces are one dimensional, we may choose basis vectors $\cY_{11}^0 \in \vertex{0}{1}{1}$ and $\cY_{10}^1 \in \vertex{1}{1}{0}$ such that for all $a \in M$ we have
\begin{equation}\label{eqn: fixing basis vector 1}
\cY_{11}^0(a,x) \otimes \cY_{11}^0(a,x) = p_{W \otimes W} Y^{\tilde V}(a \otimes a,x) p_{M \otimes M}
\end{equation}
and
\begin{equation}\label{eqn: fixing basis vector 2}
\cY_{10}^1(a,x) \otimes \cY_{10}^1(a,x) = p_{M \otimes M} Y^{\tilde V}(a \otimes a,x) p_{W \otimes W}.
\end{equation}
As usual, we are required to make the canonical choices for $\cY_{0i}^i$.

Let us first consider (1).
Let $(B,A) \in \scA_I$ and $z \in \Int(B,A)$.
Since $W \otimes \{\Omega\} \subset \tilde V$ and $M$ is a $W \otimes \{\Omega\}$-submodule of $\tilde V$, it follows that $W$ has bounded localized vertex operators (Proposition \ref{prop: BLVO subalgebra}) and $\pi^M$ exists (Theorem \ref{thmConstructionOfNetRepresentations}).
Thus $B^*,A \in \Ann^{in}(\cH_{M_i})$ for $i=0,1$, and the same holds for any $i \in \bbF_2^{n}$.

Now consider (2).
Clearly the conclusion does not depend on the choice of basis vectors, so we prove it for the choices already made.
It suffices to show that for some $s > 0$ the operators 
$$
p_k B Y^V(s^{L_0}p_i-,z)p_j
$$
are bounded for all admissible $i,j,k \in \bbF_2$, and the corresponding statement for $i,j,k \in \bbF_2^n$ follows by taking tensor products.
So let $i,j,k \in \bbF_2$ be admissible, and let $p_{ii}$ be the projection of $\tilde V$ onto $M_i \otimes M_i$, and similarly for $p_{jj}$ and $p_{kk}$.
Then by construction
\begin{equation}\label{eqn: code compression}
p_{kk}(B \otimes B)(Y^{\tilde V}(s^{L_0}p_{ii}-,z)(A \otimes A)p_{ii} 
=
(B \cY_{ij}^k(s^{L_0}-,z) A) \otimes (B \cY_{ij}^k(s^{L_0}-,z) A).
\end{equation}
But since $\tilde V$ has bounded localized vertex operators, the left-hand side is bounded for appropriate $s$, so the right-hand side is as well.
It follows that each tensor factor is bounded, as desired.

Finally, we establish (3).
In the end, we will see that the conclusion will be independent of the choice of basis vectors, but for the present we maintain the choices fixed in \eqref{eqn: fixing basis vector 1} and \eqref{eqn: fixing basis vector 2}. 
Since $M \otimes M$ is a self-dual simple current, $(W \otimes W) \oplus (M \otimes M)$ is a subalgebra of $\tilde V$, and we assume without loss of generality that $\tilde V = (W \otimes W) \oplus (M \otimes M).$
In this case the even subnet $\cB \subset \cA_{\tilde V}$ satisfies $\cB_0 = \cA_W \otimes \cA_W$, as in Lemma \ref{lem: transport unitary}.
Observe that for $a \in M$, we have 
$$
X(a)=\begin{pmatrix}0 & y(a) \otimes y(a)\\ x(a) \otimes x(a)\end{pmatrix} \in \cA_{\tilde V}(I)
$$
where
$$
x(a) = B\cY_{10}^1(a,z)A, \qquad y(a) = B\cY_{11}^0(a,z)A.
$$
For $\tilde a \in M$ we define $\tilde X(\tilde a)$, $\tilde x(\tilde a)$ and $\tilde y(\tilde a)$ in the same way, and we have $\tilde X(\tilde a) \in \cA_{\tilde V}(J)$.

Since $I,J \subset K$, we may apply Lemma \ref{lem: transport unitary} to the interval $K$ and obtain unitaries $u \in \Hom_{\cA_W(K^\prime)}(\cH_W,\cH_M)$ and $v \in \Hom_{\cA_W(K^\prime)}(\cH_M,\cH_W)$ satisfying the conclusion of that Lemma.
As before, define 
$$
\tau_K: \Hom_{\cA_W(K^\prime)}(\cH_W,\cH_M) \to \Hom_{\cA_W(K^\prime)}(\cH_M,\cH_W)
$$
by $\tau_K(x) = u^*\pi_K(vx)$.
By Lemma \ref{lem: transport unitary}, since $X(a), \tilde X(\tilde a) \in \cA_{\tilde V}(K)$ we have 
\begin{equation}\label{eqn: transport for intertwiners}
\tau_K(x(a)) = \pm y(a) 
\quad \mbox{ and } \quad 
\tau_K(\tilde x(\tilde a)) = \pm \tilde y(\tilde a).
\end{equation}
Observe that replacing $\log z$ by $\log z + 2 \pi i$ changes the sign of $y(a)$ but not of $x(a)$, and similarly $\log \tilde z$ affects $\tilde y(\tilde a)$ but not $\tilde x(\tilde a)$.
Thus we may make choices of $\log z$ and $\log \tilde z$ such that \eqref{eqn: transport for intertwiners} holds with the sign $+$ in both cases.

Since $I$ and $J$ are disjoint and $X(a)$ and $\tilde X(\tilde a)$ are odd, we have $\tilde X(\tilde a) X(a) = - X(a) \tilde X(\tilde a)$.
Expanding the top left entry of the resulting matrix we see
$$
\big(\tilde y(\tilde a)x(a) \otimes \tilde y(\tilde a)x(a)\big) 
= 
- \big(y(a)\tilde x(\tilde a) \otimes y(a)\tilde x(\tilde a)\big).
$$
Thus 
\begin{equation}\label{eqn: first braiding}
\tilde y(\tilde a)x(a) = \omega \, \big(y(a)\tilde x(\tilde a)
\end{equation}
where $\omega$ is a primitive fourth root of unity.

Examining the bottom right entry of the same matrix would give us a braiding relation $\tilde x(\tilde a)y(a) = \omega^\prime \,  x(a) \tilde y(\tilde a)$, but this argument provides no way of obtaining the crucial fact that $\omega^\prime = \omega$.
Instead, we let $\varepsilon$ be as in Lemma \ref{lem: transport unitary} and apply Lemma \ref{lem: tranport for tau} to calculate
\begin{align}
\tilde x(\tilde a)y(a) &= \tilde x(\tilde a) \tau_K(x(a))\nonumber\\
&= \varepsilon \, \pi_K(\tau_K(\tilde x(\tilde a)) x(a))\nonumber\\
&= \varepsilon \, \pi_K(\tilde y(\tilde a) x(a)) \nonumber\\
&= \varepsilon \, \omega \, \pi_K(y(a) \tilde x(\tilde a)) \nonumber\\
&= \omega \, x(a) \tilde y(\tilde a). \label{eqn: second braiding}
\end{align}
The first equality is \eqref{eqn: transport for intertwiners} (with our adjustment to make the signs $\pm$ positive), the second is Lemma \ref{lem: tranport for tau}, the third is \eqref{eqn: transport for intertwiners}, the fourth is \eqref{eqn: first braiding}, and the fifth is Lemma \ref{lem: tranport for tau} and \eqref{eqn: transport for intertwiners} again.

Thus from \eqref{eqn: first braiding} and \eqref{eqn: second braiding} we have for $i,j,k,\ell \in \bbF_2$
\begin{equation}\label{eqn: more braiding}
(\tilde B\cY_{1j}^k(\tilde a,z)\tilde A)(B\cY_{1i}^j(a,z)A) 
= 
\omega (B\cY_{1\ell}^k(a,z)A)(\tilde B\cY_{1i}^\ell(\tilde a,z)\tilde A)
\end{equation}
provided the relevant triples are admissible, for a certain choices of bases for $\vertex{1}{1}{0}$ and $\vertex{0}{1}{1}$.
However \eqref{eqn: more braiding} is preserved under change of basis, so it holds for all choices of basis vectors.
On the other hand, since these basis vectors lie in $I_{loc}$ by Theorem \ref{thm: intertwiners local}, we have
$$
(\tilde B\cY_{qj}^k(\tilde a,z)\tilde A)(B\cY_{pi}^j(a,z)A) 
= 
 (B\cY_{p\ell}^k(a,z)A)(\tilde B\cY_{qi}^\ell(\tilde a,z)\tilde A)
$$
for admissible labels in $\bbF_2$, provided $p=0$ or $q=0$, and again these identities are independent of basis choice.
Thus in general we have
$$
(\tilde B\cY_{qj}^k(\tilde a,z)\tilde A)(B\cY_{pi}^j(a,z)A) 
= 
 \omega^{pq} (B\cY_{p\ell}^k(a,z)A)(\tilde B\cY_{qi}^\ell(\tilde a,z)\tilde A).
$$
The desired braiding \eqref{eqn: ACFT braiding} now follows by taking tensor products of this identity.
\end{proof}

\begin{proof}[Proof of Theorem \ref{thm: BLVO for harder codes}]
Let $V$ be a simple current extension of $W^{\otimes 2n}$ of code type with respect to $(W,M,C)$.
Combining (1) and (2) of Lemma \ref{lem: ACFT braiding}, we see that $\tilde V$ has bounded insertions, as $BY^V(s^{L_0}-,z)A$ is given by $2n$-by-$2n$ matrix of bounded operators.
We now need to show that $\cA_V(I)$ and $\cA_V(J)$ commute when $I$ and $J$ are disjoint, and without loss of generality we may assume that $I \cup J$ is not dense in $S^1$.
Let $(B,A) \in \scA_I$ and $z \in \Int(B,A)$ and let $(\tilde B, \tilde A) \in \scA_J$ and $\tilde z \in \Int(\tilde B, \tilde A)$.
Let $i,j,k,\ell,p,q$ be as in (3) of Lemma \ref{lem: ACFT braiding}.
It suffices to show that 
\begin{equation}\label{eq: code blvo goal}
p_k(\tilde B Y(\tilde a, \tilde z) \tilde A)(BY(a,z)A)p_i
=
p_k(BY(a,z)A)(\tilde B Y(\tilde a, \tilde z) \tilde A)p_i
\end{equation}
for $a \in M_p$ and $\tilde a \in M_q$.

Choose compatible actions of $A$,$B$, $\tilde A$, and $\tilde B$, and choose $\log z$ and $\log \tilde z$ and bases for intertwiner spaces as in (3) of Lemma \ref{lem: ACFT braiding}.
Observe that \eqref{eq: code blvo goal} has no fractional powers of the variables, so cannot depend on the choices of logs.
Moreover, the value of $BY(a,z)A$ is independent of the choice of $\tilde \gamma$ made to compatibly define the actions of $A$ and $B$ since $V$ has integer conformal weights.
However, we will decompose both sides of \eqref{eq: code blvo goal} into tensor factors which do depend on these choices, and we are careful to use the same choice on each tensor factor.

Let $c:C \times C \to \bbC^\times$ be the function satisfying $p_\gamma Y^V (p_\alpha \otimes p_\beta) = c(\alpha, \beta) \cY_{\alpha \beta}^\gamma$ where $\gamma = \alpha + \beta$.
Then we have
\begin{align*}
p_k(\tilde B Y(\tilde a, \tilde z) \tilde A)(BY(a,z)A)p_i
&= 
c(q,j)c(p,i)(\tilde B \cY_{qj}^k(\tilde a, \tilde z) \tilde A)(B\cY_{pi}^j(a,z)A)\\
&=(-1)^{\frac12(p \cdot q)} c(q,j)c(p,i) (B\cY_{p\ell}^k(a,z)A) (\tilde B \cY_{qi}^\ell(\tilde a, \tilde z) \tilde A)\\
&= c(p,\ell)c(q,i)(B\cY_{p\ell}^k(a,z)A) (\tilde B \cY_{qi}^\ell(\tilde a, \tilde z) \tilde A)\\
&=p_k(BY(a,z)A)(\tilde B Y(\tilde a, \tilde z) \tilde A)p_i.
\end{align*}
The first equality is by definition of $c$ and the fact that $p_k$ and $p_i$ commute with $\tilde B$ and $A$, the second equality is Lemma \ref{lem: ACFT braiding}, the third is Lemma \ref{lem: code extension cocycle}, and the fourth is the same as the first.
Since $i$ and $k$ were arbitrary, $\tilde B Y(\tilde a, \tilde z) \tilde A$ and $BY(a,z)A$ commute.
This completes the proof.
\end{proof}

\newpage

\section{Examples and applications} \label{sec: examples}

In this section, we always discussed bounded localized vertex operators with respect to the system of generalized annuli constructed in Section \ref{sec: example of system of generalized annuli}, for which the free fermion superalgebra has bounded localized vertex operators (Theorem \ref{thm: FF BLVO}).
We fix the notation $V(\frg,k)$ for the WZW VOA corresponding to $\frg$ at positive integer level $k$.

\subsection{Bounded localized vertex operators}

\subsubsection{Lattices and code extensions}\label{sec: examples lattices in codes}

Observe that the embedding of lattices $A_1^2 \subset \bbZ^2$ yields a unitary conformal embedding $W^{\otimes 2} \subset V$ where $W = V(A_1,1)$ and $V = \cF^{\otimes 2}$ (two copies of the (complex) free fermion).
In fact, $W^{\otimes 2}$ is the even part of $\cF^{\otimes 2}$, and it decomposes as $(W \otimes W) \oplus (M \otimes M)$, where $M$ is the non-trivial irreducible $W$-module.
By Lemma \ref{lem: A1 module is semionic}, $M$ is a semionic simple current and since $M$ is the only non-trivial irreducible $W$-module, it is easy to check that any conformal extension of $W^{\otimes n}$ is of code type.
Since we have $\cA_W \otimes \cA_W \subset \cA_V$, we may apply Theorem \ref{thm: BLVO for harder codes} to obtain:

\begin{Proposition}\label{prop: A1 framed have bounded localized vertex operators}
Let $V$ be a simple unitary vertex operator algebra of central charge $n$ containing $V(A_1,1)^{\otimes n}$ as a unitary subalgebra.
Then $V$ has bounded localized vertex operators.
\end{Proposition}

\begin{Example}[Code lattices]\label{ex: code lattices BLVO}
As described in Section \ref{sec: self dual simple currents and code extensions} (see \eqref{eqn: code lattice}), starting with a self-orthogonal, doubly even binary code $C$ of length $n$, we may form a  lattice  
$$
\Lambda_C = \bigcup_{c \in C} \sqrt{2}\bbZ^n + \frac{c}{\sqrt{2}}.
$$
See \cite[\S7.2,\S5.2]{ConwaySloane} for more details on the construction.
By construction $\Lambda_C$ contains $A_1^n$ as a sublattice, and so $V_{\Lambda_C}$ contains $V(A_1,1)^{\otimes n}$ as a unitary subalgebra.
Thus by Proposition \ref{prop: A1 framed have bounded localized vertex operators} $V_{\Lambda_C}$ has bounded localized vertex operators.
Starting with the length 8 Hamming code \cite[\S3.2.4.2]{ConwaySloane}, this construction produces the $E_{8}$ lattice \cite[Ex. 7.2.5]{ConwaySloane}, and the two doubly even self-dual codes of length $16$ produce the $E_8$ and $D_{16}^+$ lattices.
There are $9$ doubly even self-dual codes of length $24$, and each produces a distinct Niemeier lattice.
In particular, the Golay code produces the Niemeier lattice with root system $A_1^{24}$  \cite[p.86-87]{DolanGoddardMontague96}.
\end{Example}
Thus all three holomorphic (even) vertex operator algebras of central charge less than $24$ have bounded localized vertex operators, along with any (not necessarily conformal) subalgebra, their tensor products, their subalgebras, their cosets, and so on.

\subsubsection{The Ising model}

The Ising model $L(\frac12, 0)$ has two non-trivial irreducible modules, one of which, $L(\frac12,\frac12)$, is a self-dual simple current.
The Ising model is the even part of the real free fermion vertex operator superalgebra $\cF_\bbR$, and we have a decomposition $\cF_\bbR = L(\frac12,0) \oplus L(\frac12,\frac12)$.
Any simple current extension of $L(\frac12,0)$ is automatically of code type since $L(\frac12,\frac12)$ is the only possible simple current.
Thus by Lemma \ref{lem: easy fermionic and bosonic simple currents} and Proposition \ref{prop: BLVO for easier codes}, we have:

\begin{Proposition}\label{prop: simple current Ising framings}
Let $V$ be a unitary vertex operator algebra containing $L(\frac12,0)^{\otimes n}$ as a unitary conformal subalgebra, and suppose that the inclusion is a simple current extension.
Then $V$ has bounded localized vertex operators.
\end{Proposition}

The VOAs considered in Proposition \ref{prop: simple current Ising framings} are a special case of what are called \emph{framed VOAs}, which admit $L(\frac12,0)^{\otimes n}$ as a conformal subalgebra.
The most famous framed VOA is the Moonshine VOA, but the extension is not a simple current extension, and therefore the Moonshine VOA is not covered by our results.
However, the structure of framed VOAs was described by Lam and Yamauchi \cite{LamYamauchi08}, and any framed VOA admits an intermediate algebra $L(\frac12,0)^{\otimes n} \subset W \subset V$ such that both smaller inclusions are simple current extensions.
By our result, $W$ would have bounded localized vertex operators, and it would be a natural direction for further study to try to study simple current extensions of $W$.

\subsubsection{WZW models}

The most difficult WZW model to establish analytic properties for is $E_{8}$ at level $1$, as it does not embed non-trivially in any other VOA.
By Example \ref{ex: code lattices BLVO}, however, $E_{8,1}$ has bounded localized vertex operators, and from there we may show:

\begin{Theorem}\label{thm: WZW BLVO}
Let $\frg$ be a simple Lie algebra of compact type, let $k$ be a positive integer, and let $V(\frg,k)$ be the associated WZW model.
Then $V(\frg, k)$ has bounded localized vertex operators.
\end{Theorem}
\begin{proof}
By Example \ref{ex: code lattices BLVO}, the $E_8$ lattice VOA has bounded localized vertex operators, and this VOA coincides with $V(E_8,1)$ by the Frenkel-Kac-Segal construction \cite[\S5.4]{Kac98}.
By Theorem \ref{thm: FF BLVO}, the complex free fermion $\cF$ has bounded localized vertex operators, and by the results of Section \ref{sec: BLVO for subalgbras} bounded localized vertex operators are inherited by subalgebras and extend to tensor products.
There are natural inclusions $V(A_{n-1},1) \subset V(D_n,1) \subset \cF^{\otimes n}$ which come from inclusions of lattices.
There are inclusions of groups which produce $V(B_n,1) \subset V(D_{n-1},1)$ and $V(C_n,1) \subset V(D_{2n},1)$.
Inclusions of lattices produce $V(E_{6},1) \subset V(E_7,1) \subset V(E_8,1)$, and moreover we have a conformal inclusion $V(F_4,1) \otimes V(G_2,1) \subset V(E_8,1)$.
Thus the theorem has been proven when $k=1$.
For $k > 1$, we appeal to the diagonal inclusion $V(\g,k) \subset V(\g,1)^{\otimes k}$.
\end{proof}

\subsubsection{Comparison to CKLW nets}

In \cite{CKLW18}, Carpi-Kawahigashi-Long-Weiner constructed a conformal net $\widetilde \cA_V$ from a unitary $V$ which was \emph{strongly local}.
Most examples of VOAs with bounded localized vertex operators are also strongly local, and we would like to know that the two constructions produce the same conformal net.
This is especially important for WZW models, where the CKLW construction produces the net that one would usually call $\cA_{\frg,k}$, where the local algebras are generated by unitary representations of the local loop groups.

While the CKLW construction is not explicitly described for vertex operator superalgebras, the main results are expected to go through unchanged in the super case; this is the subject of work in progress by Carpi, Gaudio, and Hillier.
There is also work in progress of Carpi, Weiner, and Xu which will demonstrate the strong locality of unitary subalgebras $V$ of free fermions, and for such models it would then be clear that $\cA_V = \widetilde \cA_V$.
However, for the present we will only require this fact when $V$ is one of $V(A_1,1)$ or $L(\frac12,0)$.
For $V(A_1,1)$ this is clear by Proposition \ref{propSubnetTrivial}, as the CKLW net for this example coincides with the loop group net, which is a subtheory of free fermions (see e.g. \cite{Wa98}).
For the Ising model, we refer to Lemma \ref{lem: Virasoro is right}.
We then have:

\begin{Proposition}\label{prop: SCE agrees for lattices}
Let $V$ be a unitary simple current extension of $V(A_1,1)^{\otimes n}$, and let $\Lambda$ be the associated code lattice as in Section \ref{sec: examples lattices in codes} so that $V = V_{\Lambda}$.
Let $\cA_\Lambda$ be the lattice conformal net of \cite{DongXu06}.
Then $\cA_V \cong \cA_\Lambda$
\end{Proposition}
\begin{proof}
Recall that by the above discussion $\cA_{V(A_1,1)} = \widetilde \cA_{V(A_1,1)}$ and that these are isomorphic to $\cA_{A_1}$ by \cite[\S3.1.1]{Xu09} (also \cite[Prop. 4.1.17]{BischoffThesis}).
If $C$ is the code associated to the inclusion $A_1^n \subset \Lambda$, then by \cite[Prop. 4.1.14]{BischoffThesis} we have that the vacuum Hilbert space for $\cA_\Lambda$ decomposes as $\bigoplus_{i \in C} M_i$, where $M_0$ is the vacuum representation, $M_1$ is the non-trivial representation, and for $i \in \bbF_2^n$ we have $M_i = M_{i(1)} \otimes \cdots \otimes M_{i(n)}$.
We have the same branching as VOA modules for $V(A_1,1)^{\otimes n} \subset V_\Lambda$, and thus by Proposition \ref{prop: reps from sums of modules} we have the same branching for $\cA_{V(A_1,1)}^{\otimes n} \subset \cA_{V}$.
Thus the vacuum representations of $\cA_V$ and $\cA_{\Lambda}$ are isomorphic as sectors of $\cA_{V(A_1,1)}^{\otimes n}$, and therefore by Remark \ref{rmk: unique simple current net} they are isomorphic as conformal nets.
\end{proof}

\begin{Lemma}\label{lem: subalgebra agree with CKLW}
Let $V$ be a unitary VOA which is strongly local and has bounded localized vertex operators, and such that $\cA_V \cong \widetilde \cA_V$.
If $W$ is a unitary subalgebra of $V$, then $\cA_W \cong \widetilde \cA_W$.
\end{Lemma}
\begin{proof}
Let $u: \cH_V \to \cH_V$ be a unitary isomorphism of sectors $u:\cA_V \to \widetilde \cA_{V}$.
The representations of $\Diff_c(I)$ obtained from local algebras of $\cA_V$  and $\widetilde \cA_V$ both coincide with the one obtained from the Virasoro field $Y(\nu,x)$, and thus $u$ commutes with the common Virasoro subnet of the two nets.
In particular, $u$ commutes with the action of rotation, and $uV = V$.
Define a new VOA structure on $V$ by $\tilde Y(a,x) = u Y(u^*a,x) u^*$, which we refer to as $\tilde V$.
By construction, $u: \cA_V \to \cA_{\tilde V}$ is an isomorphism of sectors, and thus $\cA_{\tilde V} = \widetilde \cA_{V}$.
Let $\tilde W = uW$, and observe that $\tilde W$ is a subalgebra of $\tilde V$ which is isomorphic to $W$.
On the other hand, by construction $u$ induces an equivalence of 
$$
\cA_W = \cA_V|_{\cH_W} \cong \cA_{\tilde V}|_{\cH_{\tilde W}} = \widetilde \cA_V|_{\cH_{\tilde W}} = \widetilde \cA_{\tilde W}.
$$
Since $\tilde W \cong W$, we have $\widetilde \cA_{\tilde W} \cong \widetilde \cA_{W} $, which completes the proof.
\end{proof}

\begin{Corollary}\label{cor: BLVO WZW is CKLW WZW}
Let $V = V(\frg,k)$ for some simple Lie algebra $\frg$ of compact type, and $k$ a positive integer.
Then $\cA_V \cong \widetilde \cA_V$.
\end{Corollary}
\begin{proof}
We have the result for $\cF$ by Theorem \ref{thm: WZW BLVO}, and for $V(E_8,1)$ by Proposition \ref{prop: SCE agrees for lattices} and \cite[Prop. 4.1.17]{BischoffThesis}.
Since $\cA_{V\otimes V} = \cA_{V} \otimes \cA_V$ and similarly for $\widetilde \cA_{V \otimes V}$, applying Lemma \ref{lem: subalgebra agree with CKLW} to each of the inclusions used in the proof of Theorem \ref{thm: WZW BLVO} gives the desired result.
\end{proof}

Of course, one would like the stronger result that $\cA_V = \widetilde \cA_V$, but that is outside the scope of the technique of Lemma \ref{lem: subalgebra agree with CKLW}.
It would also be interesting to verify that the conformal nets constructed here from simple current extensions of $L(\frac12,0)^{\otimes n}$ agree with the corresponding construction in \cite{KawahigashiLongo06}.

\subsection{Modules and local equivalence}
\label{sec: modules and local equivalence}
Let $V$ be a simple rational unitary VOA with bounded localized vertex operators.
Consider the following three properties that $V$ might enjoy:
\begin{enumerate}
\item[] Property 1: Every irreducible $V$-module $M$ admits a unitary structure.
\item[] Property 2: For every $M$, $\pi^M$ exists.
\item[] Property 3: Every irreducible sector of $\cA_V$ is of the form $\pi^M$.
\end{enumerate}

We conjecture that all three properties always hold.
It seems plausible that the results of this article suffice for a direct attack on establishing Property 3; given a sector $\pi$, one must construct a field from the holomorphic function $\pi_I(BY^V(a,z)A)$, and prove that it satisfies the properties of a VOA module.
This is similar in spirit to the technique of `local energy bounds' developed by Carpi and Weiner.
We hope to consider this problem in future work.

For WZW models $V(\frg,k)$, Property 1 is known to hold as a result of the classification of irreducible modules.
For the WZW conformal nets, it was shown by Henriques \cite[Thm. 26]{HenriquesColimits} that every irreducible representation of $\cA_{\frg,k}$ is obtained from a level $k$ irreducible positive energy representation $\pi_{\lambda}: LG_k  \to \cU(\cH_{k,\lambda})$ where $G$ is the compact simple simply connected Lie group of type $\frg$.
As a result, there are at most as many irreducible sectors of $\cA_{\frg,k}$ as there are irreducible modules of $V(\frg,k)$.
The question of whether each such sector exists is called the \emph{local equivalence problem}:

\begin{Problem*}[Local equivalence]
Let $G$ be a compact simple simply connected Lie group, let $\pi_0$ be its level $k$ vacuum representation, and let $\pi_\lambda$ be some irreducible level $k$ representation.
Let $L_IG$ be the subgroup of loops which are the identity on $I^\prime$.
Show that the map $\pi_0(g) \mapsto \pi_\lambda(g)$, for $g \in L_IG$, extends continuously to an isomorphism of von Neumann algebras $\pi_0(L_IG)^{\prime\prime} \cong \pi_\lambda(L_IG)^{\prime\prime}$.
\end{Problem*}
The analogous problem for irreducible positive energy representations of $\Diff(S^1)$ was recently solved by Mih\'{a}ly Weiner \cite{Weiner17}.

\begin{Lemma}\label{lem: Prop 2 is equivalent to local equivalence}
Let $\frg$ be a simple Lie algebra of compact type, let $G$ be the associated compact simple simply connected Lie group, and let $k$ be a positive integer.
Then Property 2 holds for $V(\frg,k)$ if and only if Property 3 holds and the local equivalence problem for $(G,k)$ has a positive answer.
\end{Lemma}
\begin{proof}
First assume Property 2.
By Henriques' work, we know that there are at most as many irreducible sectors of $\cA_{\frg,k}$ as there $V$-modules.
If $\pi^M$ exists for every $M$, then we know that there is a one-to-one correspondence between $\cA_{\frg,k}$ sectors and irreducible positive energy representations.
By a pigeonhole argument, the local equivalence problem must have a positive solution, and each of the corresponding sectors must be equivalent to some $\pi^M$.
The converse is proven similarly.
\end{proof}

Our result on existence of $\pi^M$ for submodules of a larger VOA (Theorem \ref{thmConstructionOfNetRepresentations}) provides a tool for establishing Property 2, and thereby solving the local equivalence problem.
The following theorem of Krauel and Miyamoto has been slightly restated for the case of unitary VOAs, but does not require unitarity.
It is a result about \emph{regular} VOAs, which are VOAs which enjoy a strong semisimplicity property.
By \cite[Thm. 4.5]{ABD04}, for simple unitary VOAs regularity is equivalent to being rational and $C_2$-cofinite; see \cite{ABD04} for more detail.

\begin{Theorem}[{\cite[Thm. 2]{KrauelMiyamoto15}}]\label{thm: KM module existence}
Let $V$ be a simple unitary VOA, let $U$ be a unitary subalgebra, and assume that $(U^c)^c = U$.
Assume that $V$, $U$, and $U^c$ are regular.
Then every irreducible $U$-module and every irreducible $U^c$ module occurs as a submodule of some irreducible $V$-module.
\end{Theorem}

Thus if $V$, $U$, and $U^c$ are as in Theorem \ref{thm: KM module existence} and $V$ has Property 1 and Property 2, then so do both $U$ and $U^c$ by Theorem \ref{thmConstructionOfNetRepresentations}.
A simple consequence is the following.

\begin{Proposition}\label{prop: induction for module existence}
Let $\frg$ be a simple Lie algebra of compact type.
Suppose that for every irreducible $V(\frg,1)$-module $M$, $\pi^M$ exists, and that for every positive $k$ the subalgebra $U = V(\frg,k+1) \subset V(\frg,k) \otimes V(\frg,1)$ has the property that $U^c$ is regular and $(U^c)^c = U$.
Then for all positive integers $k$ and every $V(\frg,k)$ module $M$, $\pi^M$ exists.
Thus every irreducible sector of $\cA_{V(\frg,k)}$ is of the form $\pi^M$, and the local equivalence problem has a positive solution for $\frg$.
\end{Proposition}
\begin{proof}
The proof is a straightforward induction on $k$, with the base case being a given hypothesis of the Proposition.
The inductive step follows easily from Theorem \ref{thm: KM module existence}, which would say that every irreducible $V(\frg,k+1)$ module $M$ is a submodule of some irreducible $V(\frg,k) \otimes V(\frg,1)$ module, which exists by the inductive hypothesis. 
Then $\pi^M$ exists by Theorem \ref{thmConstructionOfNetRepresentations}.
The remaining properties now follow from Lemma \ref{lem: Prop 2 is equivalent to local equivalence}.
\end{proof}
Of course, one could replace the diagonal inclusion $V(\frg,k+1) \subset V(\frg,k) \otimes V(\frg,1)$ by a related one, such as $V(\frg,k) \subset V(\frg,1)^{\otimes k}$, in the statement of Proposition \ref{prop: induction for module existence}.

In practice, the problem of verifying that $\pi^M$ exists for every level 1 representation $M$ is not too imposing.
For type $A$ it follows from the inclusion $A_n \subset \cF^{\otimes n+1}$.
For type $E$ it follows from the inclusions $E_6 \subset E_7 \subset E_8$ of lattices (and the fact that $E_8$ is unimodular).
For types $F$ and $G$, one can employ the inclusion $V(F_4,1) \otimes V(G_2,1) \subset E_8$.
For type $D$, one should be able to use the free fermion $\cF$ along with its Ramond sector, but for simplicity we have avoided discussing the Ramond sector.
For type $B$ and $C$, one should then be able to use the level one inclusions into VOAs of type $D$ employed in the proof of Theorem \ref{thm: WZW BLVO}.

On the other hand, the problem of showing that the diagonal cosets are regular is a difficult and important one in the theory of VOAs.
When $\frg$ is of type $ADE$, we may combine major theorems of Arakawa and Arakawa-Creutzig-Linshaw to obtain such a result.
Specifically,  they study the inclusion $U \subset V$ where $V = V(\frg,k) \otimes V(\frg,1)$ and $U = V(\frg,k+1)$.
Arakawa-Creutzig-Linshaw show that in this case $U^c$ is the minimal series $W$-algebra $W_{\ell}(\frg)$ for 
\begin{equation}\label{eqn: W algebra parameter}
\ell + h\check{} = \frac{k + h\check{}}{k + h\check{} + 1},
\end{equation}
where $h\check{}$ is the dual Coxeter number of $\frg$, and $(U^c)^c = U$ \cite[Main Thm. 1]{ArakawaCreutzigLinshaw19}.
On the other hand, Arakawa shows that $W_{\ell}(\frg)$ is rational and $C_2$-cofinite  \cite{Arakawa15PrincipalNilpotent,Arakawa15AssociatedVarieties} and thus regular by \cite{ABD04}.
Of particular interest is the case when $\frg = A_1$, in which case $W_\ell(\frg)$ recover the discrete series of unitary minimal models.

Combining these results, we have:

\begin{Theorem}\label{thm: local equivalence and existence}
Let $\frg$ be of type $A$ or $E$, and let $k$ be a positive integer.
Then for every irreducible $V(\frg,k)$-module $M$, $\pi^M$ exists.
Moreover, every irreducible sector of $\cA_{V(\frg,k)}$ is of the form $\pi^M$, and the local equivalence problem for $(\frg, k)$ has a positive answer.
\end{Theorem}
\begin{proof}
By the discussion following Proposition \ref{prop: induction for module existence}, the theorem is true when $k=1$.
By the discussion preceding this theorem, we may combine \cite{Arakawa15PrincipalNilpotent,Arakawa15AssociatedVarieties,ArakawaCreutzigLinshaw19} and Theorem \ref{thm: KM module existence} to verify the hypotheses of Proposition \ref{prop: induction for module existence} regarding cosets.
The theorem now follows from that proposition.
\end{proof}

The solution to the local equivalence problem for type $A$ was originally given by Wassermann \cite[\S17]{Wa98}, who also gave a proof of local equivalence in the more general ADE case in unpublished notes \cite{Wassermann90}. 
The proof in Theorem \ref{thm: local equivalence and existence} is essentially the same, albeit obtained in a more general framework.
For example, we may also prove:

\begin{Proposition}\label{prop: W algebra modules exist}
Let $\frg$ be a simple Lie algebra of type $A$ or $E$, let $k$ be a positive integer, and let $\ell$ be as in \eqref{eqn: W algebra parameter}.
Then $W_\ell(\frg)$ has bounded localized vertex operators, every irreducible $W_\ell(\frg)$-module $M$ admits a unitary structure, and $\pi^M$ exists.
\end{Proposition}
\begin{proof}
The inclusion $W_\ell(\frg) \subset V(\frg,k) \otimes V(\frg,1)$ shows that $W_\ell(\frg)$ has bounded localized vertex operators.
Since $\pi^M$ exists for every irreducible $V(\frg,k) \otimes V(\frg,1)$-module $M$, and every irreducible $W_\ell(\frg)$-module may be found inside some $V(\frg,k) \otimes V(\frg,1)$-module by \cite{ArakawaCreutzigLinshaw19} and Theorem \ref{thm: KM module existence}, the desired result follows.
\end{proof}

In particular, as noted above Proposition \ref{prop: W algebra modules exist} applies to the discrete series of unitary minimal Virasoro models $L(c,0)$.

\newpage
\bibliography{gracft2}

\end{document}